\documentclass[11pt]{article}

\input{Qcircuit}

\input{ArticlePreamble.sty}

\author[1]{Juan Bermejo-Vega\thanks{juan.bermejovega@mpq.mpg.de}}
\author[2]{Cedric Yen-Yu Lin\thanks{cedricl@mit.edu}}
\author[1]{Maarten Van den Nest}
\affil[1]{\small Max-Planck-Institut f\"ur Quantenoptik, Theory Division, Garching, Germany.}
\affil[2]{\small \ Center for Theoretical Physics, Massachusetts Institute of Technology,  Cambridge, MA, USA}

\title{ Normalizer circuits and a Gottesman-Knill theorem \\ for infinite-dimensional systems}

\begin{document}

\maketitle

\begin{abstract}
\textbf{\emph{Normalizer circuits}} \cite{VDNest_12_QFTs,BermejoVega_12_GKTheorem}  are generalized Clifford circuits that act on  arbitrary finite-dimensional systems  $\mathcal{H}_{d_1}\otimes \cdots \otimes \mathcal{H}_{d_n}$ with a standard basis labeled by the elements of a finite Abelian group $G=\mathbb{Z}_{d_1}\times\cdots \times \mathbb{Z}_{d_n}$. Normalizer gates implement operations associated with the  group $G$ and can be of three types: quantum  Fourier transforms, group automorphism gates and quadratic phase gates. In this work, we extend the normalizer formalism \cite{VDNest_12_QFTs,BermejoVega_12_GKTheorem} to \emph{infinite dimensions}, by allowing normalizer gates to act on systems of the form  $\mathcal{H}_\mathbb{Z}^{\otimes a}$:  each factor $\mathcal{H}_\mathbb{Z}$ has a  standard basis labeled by \emph{integers} $\mathbb{Z}$, and a  Fourier basis  labeled by \emph{angles}, elements of the \emph{circle group} $\mathbb{T}$. Normalizer circuits become hybrid quantum circuits acting both on continuous- and discrete-variable systems. We show that infinite-dimensional normalizer circuits can be efficiently simulated classically with a generalized \textbf{\emph{stabilizer formalism}} for Hilbert spaces associated with groups of the form $\mathbb{Z}^a\times \mathbb{T}^b \times \mathbb{Z}_{d_1}\times\cdots\times \mathbb{Z}_{d_n}$. We develop new techniques to track stabilizer-groups based on  \emph{normal forms} for group automorphisms and quadratic functions. We use our normal forms to reduce  the problem of simulating normalizer circuits to that of finding general solutions of  systems of mixed real-integer linear equations  \cite{BowmanBurget74_systems-Mixed-Integer_Linear_equations} and exploit this fact to devise a robust simulation algorithm: the latter remains efficient even in pathological cases where stabilizer groups become \emph{infinite}, \emph{uncountable} and \emph{non-compact}.  The techniques developed in this paper might find applications  in the study of fault-tolerant quantum computation with  superconducting qubits \cite{Kitaev06_Protected_Qubit_Supercond_Mirror,Brooks13_Protected_gates_for_superconducting_qubits}.
\end{abstract}

\tableofcontents

\section{Introduction}\label{sect:Previous Work}

Normalizer circuits  \cite{VDNest_12_QFTs,BermejoVega_12_GKTheorem} are a family of quantum circuits that generalize the so-called Clifford circuits \cite{Gottesman_PhD_Thesis,Gottesman99_HeisenbergRepresentation_of_Q_Computers,Knill96non-binaryunitary,Gottesman98Fault_Tolerant_QC_HigherDimensions} to Hilbert spaces associated with  finite Abelian groups $G$. Normalizer circuits are composed of  normalizer gates, important examples of which are quantum Fourier transforms (QFTs) over finite Abelian groups, which play a central role in quantum algorithms such as Shor's factoring and discrete-log algorithms \cite{Shor}. In Refs.\ \cite{VDNest_12_QFTs,BermejoVega_12_GKTheorem} it was shown that every normalizer circuit can be efficiently classically simulated. This result can be regarded as a natural generalization of the celebrated Gottesman-Knill theorem \cite{Gottesman_PhD_Thesis,Gottesman99_HeisenbergRepresentation_of_Q_Computers}, which asserts that every Clifford circuit (i.e., a circuit composed of Hadamard, CNOT and $\pi/2$-phase gates acting on qubit systems) can be efficiently classically simulated.

In this work we further generalize the normalizer circuit framework. In particular, we introduce normalizer circuits  where the associated Abelian group $G$  can be \emph{infinite}. We focus on groups of the form  $G=F\times \Z^a$,  where $F=\DProd{d}{n}$ is a finite Abelian group (decomposed into cyclic groups) as in the normalizer circuit setting considered in \cite{VDNest_12_QFTs, BermejoVega_12_GKTheorem}, and where  $\Z$ denotes the additive group of integers---the latter being an infinite group. The motivation for adding $\Z$ is that several number theoretical problems are naturally connected to problems over the integers. For example, it is well known that the factoring problem is related to the hidden subgroup problem over $\mathbb{Z}$ \cite{Brassard_Hoyer97_Exact_Quantum_Algorithm_Simons_Problem,Hoyer99Conjugated_operators,MoscaEkert98_The_HSP_and_Eigenvalue_Estimation,Damgard_QIP_note_HSP_algorithm}. Similarly to the finite Abelian case, normalizer circuits over an infinite group $G$  are composed of normalizer gates. These gates come in three types, namely automorphism gates (which are generalizations of CNOT gates), quadratic phase gates (which are generalizations of $\pi/2$-phase gates) and the quantum Fourier transform.  The main result of this paper is a proof that all  normalizer circuits over infinite groups $G$ can be simulated classically in polynomial time, thereby extending the classical simulation results obtained in \cite{VDNest_12_QFTs, BermejoVega_12_GKTheorem} for normalizer circuits over finite Abelian groups.

In extending normalizer circuits  to infinite groups $G$, several issues arise that are not present in the finite group setting. First of all, the physical system associated with $G$ is a system with a standard basis vectors $|g\rangle$ where $g$ ranges over all elements in $G$. This implies that the Hilbert space is infinite-dimensional. Another important point is related to the quantum Fourier transform (QFT). Roughly speaking, the QFT over a group $G$ is a transformation which relates two bases of the Hilbert space: namely, the standard basis $\{|g\rangle\}$ and the \emph{Fourier basis}. The Fourier basis vectors are related to the character group of $G$. If $G$ were finite Abelian, it would be isomorphic to its own character group. This feature is however no longer true for infinite groups such as $\Z$. Indeed, the character group of $\Z$ is (isomorphic to) a different group, namely, the circle group $\T=[0, 1)$ with addition modulo 1. This group represents the addition of angles in a circle, up to a constant rescaling of $2\uppi$. As we will see below, this phenomenon has important consequences for the treatment of normalizer gates over $G$. In particular, in order to construct a closed normalizer formalism, we must in fact consider groups of the more general form $\Z^a\times\T^b \times  F$ and their associated normalizer circuits. Note that $\T$ is a continuous group, whereas  $\Z$ is discrete (finitely generated).

To achieve an efficient classical simulation of normalizer circuits over $\Z^a\times\T^b \times  F$ (\textbf{theorem \ref{thm:Main Result}}),  we develop new \emph{\textbf{stabilizer formalism}} techniques which extend the stabilizer formalism for finite Abelian groups developed in  \cite{VDNest_12_QFTs,BermejoVega_12_GKTheorem} (which was in turn a generalization of the well-known stabilizer formalism for qubit/qudit systems  \cite{Gottesman_PhD_Thesis,Gottesman99_HeisenbergRepresentation_of_Q_Computers,Knill96non-binaryunitary,dehaene_demoor_coefficients,dehaene_demoor_hostens,AaronsonGottesman04_Improved_Simul_stabilizer,deBeaudrap12_linearised_stabiliser_formalism}) . The stabilizer formalism is a paradigm where a quantum state may be described by means of a group of unitary (Pauli) operators (the Pauli stabilizer group) which leave the state invariant. An appealing feature of the stabilizer formalism for finite-dimensional systems is that each stabilizer group is finite and fully determined by specifying a small list of group generators.
This list of generators thus forms a concise representation of the corresponding quantum state. Furthermore, if a Clifford gate is applied to the state, the list of generators transforms in a transparent way which can be efficiently updated. Performing such updates throughout the computation yields a stabilizer description of the output state. Finally, one can show that the statistics arising from a final measurement can be efficiently reproduced classically by suitably manipulating the stabilizer generators of the final state. 

In generalizing the \emph{stabilizer formalism} to normalizer circuits over groups of the form $G=\Z^a\times\T^b \times  F$, several complications arise since the groups in question are no longer finite nor even finitely generated. One immediate consequence is that the associated Pauli stabilizer groups are  \emph{no longer finitely generated} either. This means that the paradigm of representing the group in terms of a small list of generators no longer applies, and a different method needs to be used.  We will show how \emph{stabilizer groups can be efficiently represented in terms of certain linear maps}; these maps have concise matrix representations, which will form efficient classical descriptions of the associated stabilizer group. 

An important technical ingredient in our simulation is a proof that both quadratic functions and homomorphisms on $G$ have certain concise \textbf{\emph{normal forms}}. The latter is a purely group-theoretic result that may find interesting applications elsewhere in quantum information. For instance, our normal form can be applied to describe the relative phases of \emph{stabilizer states}, since it was shown in \cite{BermejoVega_12_GKTheorem} that such phases are quadratic\footnote{The result in \cite{BermejoVega_12_GKTheorem} is for finite dimensional systems. However, adopting the definitions in section \ref{sect:Stabilizer States}, it is easy to check that the proof extends step-by-step to the infinite dimensional case.}; as a result, one may use it to generalize part of Gross' discrete Hudson theorem \cite{Gross06_discrete_Hudson_theorem} to our setting\footnote{Gross's theorem  provides a normal form for odd-dimensional-qudit stabilizer states in terms of quadratic functions. In addition, it states that a pure state is an odd-dim qudit stabilizer state iff it has non-negative Wigner function \cite{Gross06_discrete_Hudson_theorem}. The second statement  cannot hold in our set-up, due to the presence of non-local effects (cf.\ next section).}.

A crucial ingredient in the last step of our simulation is a  polynomial-time classical algorithm that computes the \textbf{\emph{support}} of a stabilizer state, given a stabilizer group that describes it. This algorithm exploits a classical reduction of this problem to solving systems of \textbf{\emph{mixed real-integer linear equations}}  \cite{BowmanBurget74_systems-Mixed-Integer_Linear_equations}, which can be efficiently solved classically. To find this reduction, we make crucial use of  the afore-mentioned normal forms and our infinite-group stabilizer formalism  (section \ref{sect:Proof of theorem 1}).

Lastly, we mention a technical issue that arises in the simulation of the final measurement of a normalizer computation: the basis in which the measurement is performed may be \emph{continuous} (stemming again from the fact that $G$ contains factors of $\T$). As a result, accuracy issues need to be taken into account in the simulation. For this purpose, we develop \textbf{\emph{$\boldsymbol{\varepsilon}$-net techniques}} to sample the support of stabilizer states.

\subsection*{Relationship to previous work}

In the particular case when $G$ is finite, our results completely generalize the results in \cite{VDNest_12_QFTs} and some of the results in \cite{BermejoVega_12_GKTheorem}. Our normal forms for quadratic functions/stabilizer states  improve  those in \cite{BermejoVega_12_GKTheorem}. However,  the simulations in \cite{BermejoVega_12_GKTheorem} are stronger in the  finite group case, allowing  non-adaptive Pauli measurements in the middle of a normalizer computation. It should be possible to extend the techniques in \cite{BermejoVega_12_GKTheorem} to our regime, which we leave for future investigation. Normal forms for qudit stabilizer states based on quadratic functions were also given in \cite{dehaene_demoor_hostens}.

Prior to our work, an  infinite-dimensional stabilizer formalism best-known as ``the \emph{continuous variable} (CV) \emph{stabilizer formalism}'' was developed for systems that can be described in terms of harmonic oscillators \cite{Braunstein98_Error_Correction_Continuous_Quantum_Variables,Lloyd98_Analog_Error_Correction,Gottesman01_Encoding_Qubit_inan_Oscillator,Bartlett02Continuous-Variable-GK-Theorem,BartlettSanders02Simulations_Optical_QI_Circuits,Barnes04StabilizerCodes_for_CV_WEC}, which can be used as ``continuous variable'' carriers of quantum information. The CV stabilizer formalism is  used in the field of  continuous-variable  quantum information processing \cite{Braunstein98_Error_Correction_Continuous_Quantum_Variables,Lloyd98_Analog_Error_Correction,Gottesman01_Encoding_Qubit_inan_Oscillator,Bartlett02Continuous-Variable-GK-Theorem,BartlettSanders02Simulations_Optical_QI_Circuits,Barnes04StabilizerCodes_for_CV_WEC,LloydBraunstein99_QC_over_CVs,BraunsteinLoock05QI_with_CV_REVIEW,GarciaPatron12_Gaussian_quantum_information}, being a key ingredient in  current schemes for  CV quantum error correction \cite{Gottesman01_Encoding_Qubit_inan_Oscillator,Menicucci14FaultTolertant_MBQC_CV_ClusterState} and CV measurement-based quantum computation with CV cluster states \cite{ZhangBraunstein13_CV_Gaussian_Cluster_States,Menicucci06_UQC_with_CV_cluster_states,GuMileWeedbrook09_QC_with_CV_cluster,Menicucci14FaultTolertant_MBQC_CV_ClusterState}. A CV version of the Gottesman-Knill theorem \cite{Bartlett02Continuous-Variable-GK-Theorem,BartlettSanders02Simulations_Optical_QI_Circuits} for simulations of Gaussian unitaries (acting on Gaussian states) has been derived in this framework. 

We  stress that, although our infinite-group stabilizer formalism and  the CV stabilizer formalism share some similarities, they are physically and mathematically \emph{inequivalent} and should not be confused with each other. Ours is applied to describe Hilbert spaces of the form $\mathcal{H}_{\Z}^{\otimes a}\otimes\mathcal{H}_{\T}^{\otimes b}\otimes\mathcal{H}_{N_1} \otimes \cdots \otimes \mathcal{H}_{N_c}$  with a basis  $\ket{g}$ labeled by the elements of $\T^a\times \Z^b\times \DProd{N}{c}$: the last $c$ registers correspond to finite-dimensional ``discrete variable'' systems;  the first $a+b$ registers can be thought of infinite-dimensional ``\textbf{rotating-variable}'' systems that are best described in terms of \textbf{quantum rotors}\footnote{\noindent The quantum states of a \textbf{quantum fixed-axis rigid rotor} (a quantum particle that can move in a circular orbit around a fixed axis) live in a Hilbert space with position and momentum bases labeled by $\T$ and $\Z$: the position is given by a continuous angular coordinate and the angular momentum is quantized in $\pm 1$ units (the sign indicates the direction in which the particle rotates \cite{aruldhasquantum}).}. In the CV formalism \cite{Bartlett02Continuous-Variable-GK-Theorem}, in contrast, the Hilbert space is $\mathcal{H}_\R^m$ with a standard basis labeled by $\R^m$ (explicitly constructed as a product basis of    position and momentum eigenstates of $m$ harmonic oscillators). Due to these differences, the available families of normalizer gates and Pauli operators in each framework (see sections \ref{sect_normalizer_circuits}, \ref{sect:Pauli operators over Abelian groups} and \cite{Bartlett02Continuous-Variable-GK-Theorem}) are simply inequivalent.

Furthermore, dealing with continuous-variable stabilizer groups as in \cite{Gottesman01_Encoding_Qubit_inan_Oscillator, Bartlett02Continuous-Variable-GK-Theorem, Barnes04StabilizerCodes_for_CV_WEC} is sometimes simpler, from the simulation point of view, because the group $\R^m$ is also a finite-dimensional \textbf{vector space} with  
a \emph{finite basis}. In our setting, in turn, $G$ is \emph{no longer} a vector space but a \textbf{group} that may  well be \emph{uncountable} yet having \emph{neither a basis} \emph{nor a  finite generating set}; on top of that, our groups contain \emph{zero divisors} and are not \emph{compact}. These differences require new techniques to track stabilizer groups as they \emph{inherit} all these rich properties. For further reading on these issues we refer to the discussion in \cite{BermejoVega_12_GKTheorem}, where the differences between prime-qudit stabilizer codes \cite{Gottesman_PhD_Thesis,Gottesman99_HeisenbergRepresentation_of_Q_Computers,Gottesman98Fault_Tolerant_QC_HigherDimensions} (which can described in terms of fields and vector spaces) and  stabilizer codes over arbitrary spaces $\mathcal{H}_{d_1}\otimes\cdots\otimes\mathcal{H}_{d_n}$ (which are associated to a finite Abelian group) are explained in detail.

Finally, we mention some related work on the classical simulability of Clifford circuits based on different techniques other than stabilizer groups: see \cite{VdNest10_Classical_Simulation_GKT_SlightlyBeyond} for simulations of  qubit non-adaptive Clifford circuits in the Schrödinger picture based on the stabilizer-state normal form of \cite{dehaene_demoor_coefficients}; see \cite{Veitch12_Negative_QuasiProbability_Resource_QC,MariEisert12_Positive_Wigner_Functions_Quantum_Computation} for phase-space simulations of odd-dimensional qudit Clifford operation exploiting  a local hidden variable theory based on the discrete Wigner function of \cite{Gibbons04_Discrete_Phase_Space_Finite_Fields,Gross06_discrete_Hudson_theorem,Gross_PhD_Tehsis}; last, see \cite{Delfosse14_Wigner_Function_Rebits} for phase-space simulations of qubit CSS-preserving Clifford operations based on a Wigner function for rebits. 

It should be insightful at this point to discuss briefly  whether the  latter results may extend to our set-up.  In this regard, it seems plausible to the authors that efficient simulation schemes for normalizer circuits analog to those in \cite{VdNest10_Classical_Simulation_GKT_SlightlyBeyond} might exist and may even benefit from the techniques  developed in the present work (specifically, our normal forms, as well as those given in \cite{BermejoVega_12_GKTheorem}). Within certain limitations, it might also be possible to extend the results in \cite{Veitch12_Negative_QuasiProbability_Resource_QC,MariEisert12_Positive_Wigner_Functions_Quantum_Computation} and \cite{Delfosse14_Wigner_Function_Rebits} to our setting. We are aware  that local hidden variable models for the full-fledged normalizer formalism studied here cannot exist due to the existence of stabilizer-type Bell inequalities \cite{Mermin90_extreme_quantum_entanglement,Scarani05_Nonlocality_Cluster_states,OtfriedToth05_Bell_Ineq_for_Graph_States}, which  can be violated already within the qubit stabilizer formalism (the $G=\Z_2^n$ normalizer formalism). Consequently, in order to find a hypothetical non-negative quasi-probability representation of normalizer circuits with  properties analogue to those of the standard discrete Wigner function of \cite{Gross06_discrete_Hudson_theorem,Gross_PhD_Tehsis}, one would necessarily need to specialize to restricted normalizer-circuit models\footnote{Note that this might not be true for all quasi-probability representations. The locality of the hidden variable models given in \cite{Gibbons04_Discrete_Phase_Space_Finite_Fields,Gross_PhD_Tehsis,Veitch12_Negative_QuasiProbability_Resource_QC,Delfosse14_Wigner_Function_Rebits} comes both from the positivity of the Wigner function and an additional factorizability property (cf. \cite{Gross_PhD_Tehsis} and \cite{Delfosse14_Wigner_Function_Rebits} page 5, property 4): in principle, classical simulation approaches that sample non-negative quasi-probability distributions without the factorizability property are well-defined and could also work, even if they do not  lead to local hidden variable models.} with, e.g., fewer types of gates, input states or measurements; this, in fact, is part of the approach followed in \cite{Delfosse14_Wigner_Function_Rebits}, where the positive Wigner representation for qubit CSS-preserving Clifford elements is given.

Currently, there are no good candidate Wigner functions for extending the results of \cite{Veitch12_Negative_QuasiProbability_Resource_QC,MariEisert12_Positive_Wigner_Functions_Quantum_Computation} or \cite{Delfosse14_Wigner_Function_Rebits} to systems of the form $\mathcal{H}_\Z^{\otimes a}$: the proposed ones (see  \cite{Rigas_Soto10_nonnegative_Wigner_OAM_states,Rigas11_OAM_in_phase_space,HinarejosPerezBanuls12_Wigner_function_Lattices} and references therein)  associate negative Wigner values to  Fourier basis states (which are allowed inputs in our formalism and also in \cite{Veitch12_Negative_QuasiProbability_Resource_QC,MariEisert12_Positive_Wigner_Functions_Quantum_Computation,Delfosse14_Wigner_Function_Rebits}) that we introduce in section \ref{sect:Quantum states over infinite Abelian groups}; for one qubit, these  are the usual $\ket{+}$, $\ket{-}$ states. The existence of a  non-negative Wigner representation for this individual case has, yet, not been ruled out by Bell inequalities, up to our best knowledge.

\subsection*{Discussion and outlook}

Finally, we discuss some open questions suggested by our work as well as a few potential avenues for future research.

First, we mention an upcoming work \cite{BermejoLinVdN13_BlackBox_Normalizers} where we will draw a rigorous connection between the normalizer circuit framework developed here and a large family of quantum algorithms. This includes Shor's factoring and discrete-log algorithms \cite{Shor},  Cheung-Mosca's algorithm for decomposing finite Abelian groups \cite{mosca_phd, cheung_mosca_01_decomp_Abelian_groups},  Deutsch's algorithm \cite{Deutsch85quantumtheory}, Simon's \cite{Simon94onthe} and other algorithms for solving Abelian hidden subgroup problems \cite{Boneh95QCryptanalysis,Grigoriev97_testing_shift_equivalence_polynomials,kitaev_phase_estimation,Kitaev97_QCs:_algorithms_error_correction,Brassard_Hoyer97_Exact_Quantum_Algorithm_Simons_Problem,Hoyer99Conjugated_operators,MoscaEkert98_The_HSP_and_Eigenvalue_Estimation,Damgard_QIP_note_HSP_algorithm}. Note that these  quantum algorithms achieve superpolynomial speed-ups over the best known classical algorithms. In \cite{BermejoLinVdN13_BlackBox_Normalizers} we will consider normalizer circuits over groups $G=\Z^a\times\T^b \times  F$, as in the present paper, with the important distinction that the finite group $F$ need not necessarily be presented in its canonical decomposition into cyclic groups $\DProd{d}{n}$. In particular, $F$ may be a finite Abelian black-box group \cite{BabaiSzmeredi_Complexity_MatrixGroup_Problems_I}. The resulting class of normalizer circuits---which we call black-box normalizer circuits---can no longer be classically simulated. In particular, we will show that each of the quantum algorithms mentioned above \emph{is} in fact an instance of a black box normalizer circuit.  We will thus obtain a precise  connection between this class of powerful quantum algorithms, the framework of normalizer circuits and the Gottesman-Knill theorem.

Also, we point out that, due to the presence of Hilbert spaces of the form $\mathcal{H}_\Z$, our stabilizer formalism over infinite groups yields a natural framework to study continuous-variable error correcting codes for quantum computing architectures based on superconducting qubits. Consider, for instance, the so-called  \emph{$0$-$\uppi$ qubits} \cite{Kitaev06_Protected_Qubit_Supercond_Mirror,Brooks13_Protected_gates_for_superconducting_qubits}. These are encoded qubits that, in our formalism, can be written as eigenspaces of groups of (commuting) generalized Pauli operators associated to  $\Z$ and $\T$ (cf.\ sections \ref{sect:Pauli operators over Abelian groups}-\ref{sect:Stabilizer States} and also the definitions in \cite{Kitaev06_Protected_Qubit_Supercond_Mirror,Brooks13_Protected_gates_for_superconducting_qubits}). Hence, we can interpret them as instances of generalized \emph{stabilizer codes}\footnote{A generalized notion of stabilizer code over an Abelian group was introduced in \cite{BermejoVega_12_GKTheorem} by two of us.} over the groups $\Z$ and $\T$. The authors believe that it should be possible to apply the simulation techniques in this paper (e.g., our generalized Gottesman-Knill theorem) in the study of fault-tolerant quantum-computing schemes  that employ this form of generalized stabilizer codes: we remind the reader that the standard Gottesman-Knill theorem \cite{Gottesman_PhD_Thesis,Gottesman99_HeisenbergRepresentation_of_Q_Computers} is often  applied in fault-tolerant schemes for quantum computing with  traditional qubits, in order to   delay recovery operations and track the evolution of Pauli errors (see, for instance, \cite{Steane03_Overhead_Threshold_FTQEC, Knill05_QComp_RealisticallyNoisy,DiVincenzo07_effectiveFTQC_Slow_Measurements,Paler14ErrorTracking}).

Also in relation with quantum error correction, it would interesting to improve our stabilizer formalism in order to describe \emph{adaptive Pauli measurements}; this would  extend the results in \cite{BermejoVega_12_GKTheorem}, where two of us showed how to do that for finite Abelian groups.

In connection with previous works, it would be interesting to study normalizer circuits over the group $\R^m$, to understand how they compare to the Gaussian formalism and to analyze in a greater level of detail the full hybrid scenario $\mathcal{H}_{\R}^{\otimes a}\otimes\mathcal{H}_{\Z}^{\otimes b}\otimes\mathcal{H}_{\T}^{\otimes c}\otimes\mathcal{H}_{N_1} \otimes \cdots \otimes \mathcal{H}_{N_d}$.  We have left this question to future investigation. However, we  highlight that the formalism of  normalizer circuits we present can be extended transparently to study this case, by considering additional Hilbert spaces of the form  $\mathcal{H}_{\R}^{\otimes m}$ with associated groups $\R^m$;  the stabilizer formalism over groups that we develop, our normal forms and simulation techniques can be applied to study that context with little (or zero) modifications. 

Another question that remains  open after our work is whether normalizer circuits over finite Abelian groups constitute all possible Clifford operations (see a related conjecture in \cite{BermejoVega_12_GKTheorem}, for finite dimensional systems).

Lastly, we  mention that an important ingredient underlying the consistency of our normalizer/stabilizer formalism is the fact that the groups associated to the Hilbert space fulfill the so-called \textbf{Pontryagin-Van Kampen duality} \cite{Morris77_Pontryagin_Duality_and_LCA_groups,Stroppel06_Locally_Compact_Groups,Dikranjan11_IntroTopologicalGroups,rudin62_Fourier_Analysis_on_groups,HofmannMorris06The_Structure_of_Compact_Groups,Armacost81_Structure_LCA_Groups,Baez08LCA_groups_Blog_Post}. From a purely-mathematical point of view, it is possible to associate a family of normalizer gates to every group in such class, which accounts for all possible Abelian groups that are locally compact Hausdorff (often called LCA groups). Some LCA groups are notoriously complex objects and remain unclassified to date. Hilbert spaces associated to them can exhibit exotic properties, such as   \emph{non-separability}, and may not always be in correspondence with natural quantum mechanical systems. In order to construct a physically relevant model of quantum circuits, we have restricted ourselves to groups of the form $\Z^a \times \T^b \times \DProd{N}{c}$, can be naturally associated to known quantum mechanical systems. We believe that the results presented in this paper can be easily extended to all groups of the form $\Z^a \times \T^b \times \DProd{N}{c}\times \R^d$, which we call ``\emph{elementary}'', and form a well-studied class of groups known as ``\emph{compactly generated Abelian Lie groups}'' \cite{Stroppel06_Locally_Compact_Groups}. Some examples of LCA groups that are not elementary are the $p$-adic numbers $\mathbb{Q}_p$ and the adele ring $\mathbb{A}_F$ of an algebraic number field $F$ \cite{Dikranjan11_IntroTopologicalGroups}.

\section{Outline of the paper}

Section \ref{sect:Summary} contains a non-technical \textbf{summary of concepts} and several examples of normalizer gates. We explain the more technical aspects of \textbf{our setting} in sections \ref{sect:Quantum states over infinite Abelian groups},\ref{sect_normalizer_circuits}, where we introduce  normalizer circuits and the Hilbert spaces where they act in full detail.

In section \ref{sect:Main Result} we state the \textbf{main result} (theorem \ref{thm:Main Result}). 

In the remaining sections we develop the techniques we use to prove the main result.

Sections  \ref{sect:Group Theory}-\ref{sect:quadratic_functions} contain \textbf{classical techniques}. The proofs of the theorem in this section have been moved to appendices \ref{appendix:Supplement to section Homomorphisms}-\ref{app:Supplement Quadratic Functions}, in order to give more attention  to the quantum techniques of the paper. Section \ref{sect:Group Theory} surveys some necessary notions of group and character theory. In section \ref{sect:Homomorphisms and matrix representations} we develop a  theory of matrix representations of group homomorphisms.  In section \ref{sect:Systems of linear equations over groups} we study systems of linear equations over groups (including systems of mixed real-integer linear equations) and algorithms to solve these. In section \ref{sect:quadratic_functions} we present normal forms for quadratic functions.

Sections \ref{sect:Pauli operators over Abelian groups}-\ref{sect:Proof of theorem 1} contain \textbf{quantum techniques}. In section \ref{sect:Pauli operators over Abelian groups} we study the properties of Pauli operators over Abelian groups. In section \ref{sect:Stabilizer States} we present stabilizer group techniques based on these operators. Finally, in \ref{sect:Proof of theorem 1}, we prove our main result.

\section{Summary of concepts}\label{sect:Summary}

In this section we give a rough intuitive definition of our circuit model and provide examples of normalizer gates to illustrate their operational meaning. Our model is presented in full detail in section \ref{sect_normalizer_circuits}, after reviewing some necessary notions of group-theory in section  \ref{sect:Quantum states over infinite Abelian groups}. The readers interested in understanding the proofs of our main results should consult these sections.

\subsection{The setting}

In short, normalizer gates are quantum gates that act on a Hilbert space $\mathcal{H}_G$ which has an orthonormal standard basis $\{\ket{g}\}_{g\in G}$ labeled by the elements of an Abelian group $G$. The latter can be finite or infinite, but it must have a well-defined integration (or summation) rule\footnote{All groups studied later have  a well-defined Haar measure.} so that the group has well-defined Fourier transform. We define a \emph{normalizer circuits over $G$} to be any quantum circuits built of the following \emph{normalizer gates}:
\begin{enumerate}
\item \textbf{Quantum Fourier transforms}. These gates implement the (classical) Fourier transform of the group $\psi(x)\rightarrow\hat{\psi}(p)$ as a quantum operation $\int \psi(x)\ket{x}\rightarrow \int \hat{\psi}(p)\ket{p}$. Here, $\psi$ is a  complex function acting on the group and $\hat{\psi}$ is its Fourier transform.
\item \textbf{Group automorphism gates}. These implement  group automorphisms $\alpha:G \rightarrow G$, at the quantum level $\ket{g}\rightarrow\ket{\alpha(g)}$. Here, $g$ and $\alpha(g)$ denote elements of $G$.
\item \textbf{Quadratic phase gates} are diagonal gates  that multiply standard basis states with \emph{quadratic} phases: $\ket{g}\rightarrow \xi(g)\ket{g}$. This means that $g\rightarrow \xi(g)$ is a quadratic (``almost multiplicative'')  function with the property $\xi(g+h)=\xi(g)\xi(h)B(g,h)$, where $B(g,h)$ is a bi-multiplicative correcting term.
\end{enumerate}

\subsection{Examples}

In order to illustrate these definitions, we  give  examples of normalizer gates for finite groups and infinite groups of the form $\Z^a$ (integer lattices) and $\T^b$ (hypertori).

\subsubsection*{Normalizer circuits over finite groups}

First we give a few examples of normalizer circuits over finite Abelian groups. (We also refer the reader to  \cite{VDNest_12_QFTs,BermejoVega_12_GKTheorem}, where these circuits have been  extensively studied.)

First consider $G=\Z_d^m$. The Hilbert space in this case is $\mathcal{H}=\mathcal{H}_d^m$, corresponding to a system of $m$ qudits. The following generalizations of the CNOT and Phase gates are automorphism and quadratic phase gates, respectively \cite{VDNest_12_QFTs}:
\begin{equation} \nonumber \text{SUM}_{d,a} = \sum_{x,y\in\Z_d} \ket{x,x+ay} \bra{x,y} ,\qquad  S_d = \sum_{x\in\Z_d} \exp{\left(\tfrac{\pii}{d} x(x+d)\right)} \ket{x} \bra{x} \end{equation}
where $a\in\Z_d$ is arbitrary. The quantum Fourier transform can also be expressed similarly as a gate:
\begin{equation}  \nonumber  F_d = \frac{1}{\sqrt{d}}\sum_{x,y\in\Z_d} \exp{\left(\tfrac{2\pii xy}{d}\right)} \ket{x} \bra{y} \end{equation}
For the qubit case, $d=2$, we recover the definitions of CNOT, phase gate $S=\mathrm{diag}(1,\imun)$ and Hadamard gate, and our main result (theorem \ref{thm:Main Result}) yields the (non-adaptive) Gottesman-Knill theorem \cite{Gottesman_PhD_Thesis,Gottesman99_HeisenbergRepresentation_of_Q_Computers}. 

Another normalizer gate for qudits is the  multiplication gate $M_{d,a}= \sum_{x\in\Z_d} \ket{ax}\bra{x}$, where $a$ is coprime to $d$, which is an automorphism gate. This single-qudit gate acts non-trivially only for $d\geq 2$.

As an interesting subcase, we can choose the group to be of the form $G=\Z_{2^n}$ with exponentially large $d$. Normalizer gates in this example are $M_{2^n,a}$, $S_{2^n}$, and $\mathcal{F}_{2^n}$, as defined above. The quantum Fourier transform in this case, $\mathcal{F}_{2^n}$, corresponds to the standard discrete QFT used in Shor's algorithm for factoring \cite{Shor}.

\subsubsection*{The infinite case $G=\Z^m$}

Let us now move to  an infinite case, choosing the group $G=\Z^m$ to be an integer lattice. Examples of automorphism gates and quadratic phase gates are, respectively, 
\begin{equation}\nonumber \text{SUM}_{\Z,a} = \sum_{x,y\in\Z} \ket{x,x+ay} \bra{x,y},\qquad S_p = \sum_{x\in\Z} \exp{(\pii p x^2)} \ket{x} \bra{x} \end{equation}
where $a$ is an arbitrary integer and $p$ is an arbitrary real number. The fact that these gates are indeed normalizer gates follows from  general normal forms for matrix representations group homomorphisms (lemma \ref{lemma:Normal form of a matrix representation}) and quadratic functions (theorem \ref{thm:Normal form of a quadratic function}) that we later develop.

The case of quantum Fourier transforms is more involved in this case, since the QFT over $\Z$ can no longer be regarded as a quantum logic gate, i.e., a unitary rotation \cite{nielsen_chuang}. This happens  because the QFT now performs a non-canonical change of the standard integer basis $\{z\}_{z\in \Z}$ of the space, into a new basis $\{t\}_{t\in \T}$ labeled by the elements of the circle group $\T=[0,1)$. The latter property is due to the fact that the QFT over $\Z$ is nothing but the quantum version of the discrete-time Fourier transform \cite{oppenheim_Signals_and_Systems} (the inverse of the well-known Fourier series) which sends  functions over $\Z$ to \emph{periodic functions over the reals}. In this paper, we will understand the  QFT over $\Z$ as  change of basis  between two orthonormal bases of the Hilbert space $\mathcal{H}_\Z$. We discuss these technicalities in detail  in section \ref{sect_normalizer_circuits}.

The QFT over $\Z$ acts on quantum states as in the following examples (cf.\ section \ref{sect_normalizer_circuits}):
\begin{center}
\begin{minipage}{0.45\linewidth}
\centering
{\it State before QFT over $\Z$}
\begin{equation*} \ket{x} \end{equation*}
\begin{equation*} \sum_{x\in\Z} \overline{\euler^{2\pii\, px}}\ket{x} \end{equation*}
\begin{equation*} \sum_{\substack{x\in\Z\\\:}} \ket{rx} \end{equation*}
\end{minipage}
\begin{minipage}{0.45\linewidth}
\centering
{\it State after QFT over $\Z$}
\begin{equation*} \int_{\mbb{T}} dp\, {\euler^{2\pii\, px}}\ket{p} \end{equation*}
\begin{equation*} \ket{p} \end{equation*}
\begin{equation*} \frac{1}{r}\sum_{\substack{k\in\Z :\\ k/r\in \mbb{T}}}\ket{k/r}  \end{equation*}
\end{minipage}
\end{center}
These transformations  can be found in standard signal processing textbooks  \cite{oppenheim_Signals_and_Systems}.

\subsubsection*{The infinite case: $G = \T^m$}

Finally, we choose the group $G=\T^m$ be an $m$-dimensional torus. Two examples of automorphism gates are the sum and sign-flip gates:
\begin{equation} \nonumber \text{SUM}_{\T,b} = \sum_{p,q\in\T} \ket{p,q+bp} \bra{p,q},\qquad M_{\T,s} = \sum_{p\in\T} \ket{sp} \bra{p} \end{equation}
where $b$ is an arbitrary integer and $s=\pm1$ (again, these formulas come from lemma \ref{lemma:Normal form of a matrix representation}). 

Unlike the previous examples we have considered, any quadratic phase gate over $G$  is \emph{purely multiplicative} (i.e., the bi-multiplicative function $B(g,h)$ is always trivial\footnote{This fact can be understood in the light of a later result, theorem \ref{thm:Normal form of a quadratic function} and it is related to  nonexistence of nontrivial group homomorphisms from $\T^m$ to $\Z^m$, the latter being the character group of $\T^m$ up to isomorphism.}). In the case $m=1$, this is equivalent to saying that any such gate is of the form
\begin{equation}\nonumber \sum_{p\in\T} \exp{(2\pi i b p)}\ket{p} \bra{p} \end{equation}
with $b$  an arbitrary integer.

Lastly, we take a look at the effect of the quantum Fourier transform over $\T$ in some examples. Similarly to the  QFT over $\Z$, this gate performs a non-canonical change of the standard basis (now  $\{t\}_{t\in \T}$ changes to $\{z\}_{z\in \Z}$) and should  be understood as a change of basis that does not correspond to a gate. The QFT over $\T$ is the quantum version of the Fourier series \cite{oppenheim_Signals_and_Systems}.
\begin{center}
\begin{minipage}{0.45\linewidth}
\centering
\vspace{+5pt}{ \em State before QFT over $\T$}
\begin{equation*} \ket{p} \end{equation*}
\begin{equation*} \int_{\mbb{T}} dp\, \euler^{2\pii\, px}\ket{p} \end{equation*}
\begin{equation*} \frac{1}{r}\sum_{\substack{k\in\Z :\\ k/r\in \mbb{T}}}\ket{k/r} \end{equation*}
\end{minipage}
\begin{minipage}{0.45\linewidth}
\centering
{\it State after  QFT over $\T$}
\begin{equation*} \sum_{x\in\Z} \euler^{2\pii\, px}\ket{x} \end{equation*}
\begin{equation*} \ket{-x} \end{equation*}
\begin{equation*} \sum_{\substack{x\in\Z\\\:}} \ket{rx} = \sum_{\substack{x\in\Z\\\:}} \ket{-rx} \end{equation*}
\end{minipage}
\end{center}
Comparing the effect of the QFT on ${\cal H}_{\Z}$ and the QFT on ${\cal H}_{\T}$, we see that the former is the ``inverse'' of the latter up to a change of sign of the group elements labeling the basis; concatenating the two of them yields the transformation $|x\rangle$ to $|-x\rangle$. This is a general phenomenon, which we shall observe again in the proof of some results (namely,  theorem \ref{thm:Normalizer gates are Clifford}).

\section{Preliminaries: the Hilbert space of a group}\label{sect:Quantum states over infinite Abelian groups}

In this section we introduce Hilbert spaces associated with Abelian groups of the form 
\begin{equation}\label{group_hilbert_space} 
G=\DProd{N}{a} \times \Z^b
\end{equation}
 where $\mathbb{Z}_N$ is the additive group of integers modulo $N$ and $\Z$ is the additive group of integers.

\subsection{The integers modulo $N$}

First we consider $\mathbb{Z}_N$. With this group, we associate an $N$-dimensional Hilbert space ${\cal H}_{N}$ having a basis $\{|x\rangle: x\in\mathbb{Z}_N\}$, henceforth called the standard basis of ${\cal H}_N$. A state in ${\cal H}_{N}$ is (as usual) a unit vector $|\psi\rangle = \sum \psi(x)|x\rangle$ with $\sum |\psi(x)|^2=1$, where $\psi(x)\in\mathbb{C}$ and where the sum is over all $x\in\mathbb{Z}_N$.

\subsection{The integers $\Z$}\label{sect_Hilbert_space_Z}

An analogous construction is made for the group $\mathbb{Z}$, although an important distinction with the former case is that $\mathbb{Z}$ is infinite. We consider the infinite-dimensional Hilbert space ${\cal H}_{\Z}= \ell_2(\mathbb{Z})$ with standard basis states $|z\rangle$ where $z\in \mathbb{Z}$. A state in ${\cal H}_\Z$ has the form \begin{equation}\label{psi_Z} 
|\psi\rangle = \sum_{x\in\Z} \psi(x)|x\rangle
\end{equation}
 with $\sum |\psi(x)|^2=1$, where the (infinite) sum is over all $x\in\mathbb{Z}$. More generally, any normalizable sequence $\{\psi_x: x\in\mathbb{Z}\}$ with $\sum |\psi_x|^2<\infty$ can be associated with a quantum state in ${\cal H}_{\Z}$ via normalization. Sequences whose sums are not finite can give rise to \emph{non-normalizable} states. These states do not belong to ${\cal H}_{\Z}$ and are hence unphysical; however it is often convenient to consider such states nonetheless. Some examples are the \emph{plane-wave states}, 
\begin{equation}\label{eq:Fourier basis state over Z}
|p\rangle:=\sum_{z\in \Z} \overline{e^{2\pi i zp}}|z\rangle \quad p\in [0, 1).
\end{equation}
Define the \emph{circle group}\footnote{Since  the group $\T^n$ is nothing but an $n$-dimensional torus,  $\T$ is sometimes referred as  ``the torus group'' \cite{Dikranjan11_IntroTopologicalGroups}. $\T$ should not be confused with $\T^2$ the usual (two-dimensional) torus.} $\T$  to be the group consisting of all real numbers in the interval $[0, 1)$ with the addition modulo 1 as group operation. Even though the  $|p\rangle$ themselves do not belong to ${\cal H}_\Z$, every state (\ref{psi_Z}) in this Hilbert space can be re-expressed as \begin{equation}
\label{psi_T}|\psi\rangle = \int_{\mathbb{T}} \mbox{d}p \ \phi(p) |p\rangle\end{equation}
for some complex function $\phi:\T\to \mathbb{C}$,  where d$p$ denotes the Haar measure on $\T$. Thus the states $|p\rangle$ form an alternate basis\footnote{Although we use this terminology, the $|p\rangle$ do not form a ``basis'' in the usual sense since these states are unnormalizable and lie outside the Hilbert space ${\cal H}_{\Z}$.  Rigorously, the $|p\rangle$ kets should be understood as Dirac-delta measures, or as Schwartz-Bruhat tempered distributions \cite{Bruhat61_Schwatz-Bruhat-functions,Osborne75_Schwartz_Bruhat}. The theory of rigged Hilbert spaces \cite{delaMadrid05_roleofthe_riggedHilbert, Antoine98_QM_beyond_Hilber_space, Gadella02_unified_Dirac_formalism, gadella12_Riggings_LCA_Groups} (often used to study observables with continuous spectrum) establishes that the $|p\rangle$ kets  can be used as a  basis for all practical purposes that concern us.}  of ${\cal H}_\Z$ which is both infinite and continuous; we call this the Fourier basis. The Fourier basis is orthonormal in the sense that 
\begin{equation} 
\langle p|p'\rangle =  \delta(p-p'),
\end{equation}
where $\delta(\cdot)$ is the Dirac delta function.

In the following, when dealing with the Hilbert space ${\cal H}_{\Z}$, we will use both the standard and Fourier basis. This is different from the finite space ${\cal H}_N$ where we only use the standard basis. The motivation for this is the following: note that, similarly to ${\cal H}_{\Z}$, the space ${\cal H}_N$ has a Fourier basis, given by the vectors 
\begin{equation} 
|\tilde y\rangle:=\sum_{x\in \Z_N}\overline{e^{2\pi i \frac{xy}{N}}} |x\rangle\quad \mbox{ for every } y\in\Z_N.
\end{equation} 
Since both the Fourier basis and the standard basis have equal cardinality, there exists a unitary operation mapping $|y\rangle\to |\tilde y\rangle$. This operation is precisely the quantum Fourier transform over $\Z_N$ (see section \ref{sect_normalizer_circuits}). Thus, instead of working with both the standard and Fourier basis, we may equivalently use the standard basis alone plus the possibility of applying this unitary rotation (which is the standard approach). In contrast, the standard and Fourier basis in the infinite-dimensional space ${\cal H}_\Z$ have different cardinality, since the former is indexed by the discrete group $\Z$ and the latter is indexed by the continuous group $\T$. Thus there does not exist a unitary operation which maps one basis to the other. For this reason, we will work with both bases.

\subsection{Total Hilbert space}

In general we will consider groups of the form (\ref{group_hilbert_space}). The associated Hilbert space has the tensor product form
\begin{equation}
\label{total_H} {\cal H}_G={\cal H}_{{N_1}}\otimes \dots\otimes {\cal H}_{{N_a}}\otimes {\cal H}_{\Z}\otimes\dots\otimes {\cal H}_{\Z}
\end{equation} with  $b$ occurrences of the infinite-dimensional space ${\cal H}_{\Z}$. For each finite-dimensional space ${\cal H}_{{N_i}}$ we consider its standard basis as above. For each ${\cal H}_\Z$ we may consider either its standard basis or its Fourier basis. The total Hilbert space has thus a total $2^b$ possible choices that are labeled by different groups: we call these bases the \emph{\textbf{group-element bases}} of $\mathcal{H}_G$. For example, choosing the standard basis everywhere yields the basis states
\begin{equation}
\label{standard_basis_all} |x(1)\rangle\otimes \cdots\otimes |x(a)\rangle\otimes |y(1)\rangle\otimes\dots\otimes |y(b)\rangle \quad x(i)\in\mathbb{Z}_{N_i}; \ y(j)\in \mathbb{Z},
\end{equation}
which is labeled by the elements of the group $\DProd{N}{a}\times \Z^b$. By choosing the Fourier basis in the $(a+b)$-th space we obtain in turn
\begin{equation}
|x(1)\rangle\otimes \cdots\otimes |x(a)\rangle\otimes |y(1)\rangle\otimes \dots\otimes|y(b-1)\rangle\otimes |p\rangle \quad x(i)\in\mathbb{Z}_{N_i}; \ y(j)\in \mathbb{Z};\ p\in \mathbb{T},
\end{equation}
which is labeled by the elements of $\DProd{N}{a}\times \Z^{b-1}\times \T$.

More generally, each of the $2^b$ tensor product bases of (\ref{total_H}) is constructed as follows. Consider any Abelian group of the form
\begin{equation}\label{group_labels_basis} G'=\mathbb{Z}_{N_1}\otimes \dots\mathbb{Z}_{N_a} \times G_1\times\cdots \times G_b \quad G_i \in\{\mathbb{Z}, \mathbb{T}\}.
\end{equation}
The associated \textbf{\emph{group-element basis}} of $\mathcal{B}_{G'}$ of $\mathcal{H}_G$ is 
 \begin{equation}\label{basis} {\cal B}_{G'}:=\{ |g\rangle:=|g(1)\rangle\otimes \dots|g(a+b)\rangle; \quad g=(g(1), \dots, g(a+b))\in G'\}.
 \end{equation}
The notation $G_i = \T$ indicates that  $\ket{g(i)}$ is a Fourier state of $\Z$ (\ref{eq:Fourier basis state over Z}). The states $\ket{g}$ are product-states with respect to the tensor-product decomposition of (\ref{total_H}). There are $2^b$ possible choices of groups in (\ref{group_labels_basis}) (all of them related via Pontryagin duality\footnote{From a mathematical point of view, all groups (\ref{group_labels_basis}) form a family (in fact, a category) which is generated by replacing the factors $G_i$ of the  group $\Z^{a}\times \DProd{N}{b}\times \mathbf{B}$ with their character groups  $G_i^*$, and identifying isomorphic groups. Pontryagin duality \cite{Morris77_Pontryagin_Duality_and_LCA_groups,Stroppel06_Locally_Compact_Groups,Dikranjan11_IntroTopologicalGroups,rudin62_Fourier_Analysis_on_groups,HofmannMorris06The_Structure_of_Compact_Groups,Armacost81_Structure_LCA_Groups,Baez08LCA_groups_Blog_Post} then tells us that there are $2^b$ different groups and bases. Note that this  multiplicity is a purely \emph{infinite-dimensional feature}, since  all finite groups are isomorphic to their character groups; consequently, this feature does not play a role in the study of finite-dimensional normalizer circuits   (or Clifford circuits) \cite{VDNest_12_QFTs,BermejoVega_12_GKTheorem}.}) and $2^b$ inequivalent group-element basis of the Hilbert space.

Finally, we note that, relative to any basis ${\cal B}_G$,  every quantum state (normalizable or not) can be expressed as
\begin{equation}
\ket{\psi}=\int_X\mathrm{d}g\, \psi(g)\ket{g},
\end{equation} where the $\psi(g)$ are complex coefficients and where $X$ is a subset of $X$ (if $X$ is a discrete set, the integral should be replaced by a sum over all elements in $X$).

\section{Normalizer circuits}\label{sect_normalizer_circuits}

In this section we generalize the  normalizer circuits \cite{VDNest_12_QFTs,BermejoVega_12_GKTheorem} to general infinite-dimensional Hilbert spaces of the form (\ref{total_H}). 

We remind that, in previous works \cite{VDNest_12_QFTs,BermejoVega_12_GKTheorem}, a finite Abelian group $G$ determines both the standard basis of $\mathcal{H}_G$ and the allowed types of gates, called normalizer gates. A technical complication that arises when one tries to define normalizer gates over  infinite dimensional groups of the form $\Z$ and $\T$ is that there are several group-elements basis that can be associated to the Hilbert space $\mathcal{H}_G$ (see previous section). As a result, there is not a natural preferred choice of standard basis\footnote{In the usual sense of the word, which comes from  the standard model of quantum circuits  \cite{nielsen_chuang}.} in a  normalizer circuit over an infinite Abelian group. To deal with this aspect of the Hilbert space, we allow the ``standard'' basis of the computation to be \emph{time-dependent}: at every time step $t$ in a normalizer circuit there is a designated ``standard'' basis ${\cal B}_{G_t}$, which must be a group-element basis, and is subject to change along the computation.\\

\noindent\textbf{Designated basis $\mathcal{B}_G$.} Let $G'$  be a member of the family of $2^b$ groups defined  in (\ref{group_labels_basis})---note that $G$ thus generally contains factors $\Z_{N_i}$, $\Z$ and $\T$.  We consider the corresponding group-element basis ${\cal B}_{G'}= \{|g\rangle: g\in G'\}$ as defined in (\ref{basis}). When we  set ${\cal B}_{G'}$ to be the \emph{designated basis} of the computation, it is meant that the group $G$ that labels  $\mathcal{B}_{G'}$ will determine the allowed  normalizer gates (read below), and that $\mathcal{B}_{G'}$ will be the basis in which measurements are performed.

\subsection{Normalizer gates}\label{sect Normalizer gates}

We now define the allowed gates of the computation, called \emph{normalizer gates}. These can be of three types, namely automorphism gates, quadratic phase gates and quantum Fourier transforms. To define these gates, we assume that the designated basis of the computation is $\mathcal{B}_G$ for some  $G$ of the form (\ref{group_labels_basis}).

\

\noindent \textbf{Automorphism gates.} Consider a continuous group automorphism $\alpha:G\rightarrow G$, i.e. $\alpha$ is an invertible map satisfying $\alpha(g+h)  = \alpha(g) + \alpha(h)$ for every $g, h\in G$. Define a unitary gate $U_{\alpha}$ by its action on the basis ${\cal B}_G$, as follows: $ U_\alpha:\ket{h}\rightarrow\ket{\alpha(h)}$. Any such gate is called an automorphism gate. Note that $U_{\alpha}$ acts as a permutation on the basis ${\cal B}_{G}$ and is hence a unitary operation.

\

\noindent \textbf{Quadratic phase gates.}  A function $\chi:G\rightarrow U(1)$  (from the group $G$ into the complex numbers of unit modulus) is called a character if $\chi(g+h) =\chi(g)\chi(h)$ for every $g, h\in G$ and $\chi$ is continuous. A function $B:G\times G\to U(1)$ is said to be a bicharacter if it is continuous and if  it is a character in both arguments. A function $\xi:G\rightarrow U(1)$   is called \emph{quadratic} if it is continuous and if
\begin{equation}\label{eq:Quadratic Function}
\xi(g+h)=\xi(g)\xi(h)B(g,h),\quad \text{for every $g$, $h\in G$}
\end{equation}
for some bicharacter $B(g,h)$. A quadratic phase gate is any diagonal unitary operation acting on ${\cal B}_G$ as $D_{\xi}: \ket{h}\rightarrow \xi(h)\ket{h}$, where $\xi$ is a quadratic function of $G$.

\

\noindent\textbf{Quantum Fourier transform.} In contrast to both automorphism gates and quadratic phase gates, which leave the designated basis unchanged, the quantum Fourier transform (QFT)  is a basis-changing operation: its action is \emph{precisely} to change the designated basis ${\cal B}_G$ (at a given time) into another group-element basis ${\cal B}_{G'}$, according to certain rules described next.

Roughly speaking, the QFT  realizes the unitary change of basis between standard and Fourier basis (section \ref{sect:Quantum states over infinite Abelian groups}). More precisely, there are different types of Fourier transforms, which act on the individual spaces ${\cal H}_N$ and ${\cal H}_{\Z}$.

\begin{itemize}
\item \textbf{QFT over $\Z_N$} In the case of ${\cal H}_N$, both bases have the same cardinality (since the ${\cal H}_N$ is finite-dimensional) and such a change of basis can be actively performed by means of a unitary rotation; the QFT over $\Z_N$  (which we denote ${\cal F}_{N}$) is precisely defined to be this unitary operation. In the standard basis  $|x\rangle$ with $x\in \Z_N$, the action of this gate on a state $|\psi\rangle = \sum \psi_x|x\rangle$ is 
\begin{equation} {\cal F}_N|\psi\rangle= \sum_{y\in\mathbb{Z}_N} \hat\psi(y) |y\rangle \quad\mbox{ with } \hat\psi(y):= \frac{1}{\sqrt{N}} \sum_{x\in\mathbb{Z}_N} {\euler^{2\pi i xy}}\psi(x).
\end{equation}

\item \textbf{Infinite-dimensional QFTs}. Quantum Fourier transforms acting on spaces $\mathcal{H}_\Z$ have  more exotic features than their finite dimensional counterparts. In first palce, these \emph{{QFTs are not gates}} in the strict sense: since the standard basis $\{|x\rangle: x\in \Z\}$ and Fourier basis $\{|p\rangle: p\in \T\}$ have different cardinality, they  cannot be rotated into each other (section \ref{sect_Hilbert_space_Z}); thus, the QFT is a change of basis between two orthonormal basis, but it does not define a  unitary rotation. In addition, due to this asymmetry, there are \emph{two distinct types of QFTs}, defined as follows.

 \textbf{QFT over $\Z$}. If the standard basis ($\Z$ basis) is the designated basis of ${\cal H}_{\Z}$, states are represented as 
 \begin{equation}|\psi\rangle= \sum_{x\in\mathbb{Z}} \psi(x) |x\rangle.\end{equation}
Gates are defined according to this (integer) basis, which is also our measurement basis. When we say that \emph{the QFT over $\Z$ is applied to $|\psi\rangle$}, we mean that the designated basis is changed from the standard basis to the Fourier basis.  The state does not physically change\footnote{This somewhat metaphoric notation is chosen to be consistent with previously existing terminology \cite{VDNest_12_QFTs,BermejoVega_12_GKTheorem}.}, but  normalizer gates applied after the QFT  will be related to the circle group $\T$ (and \emph{not} $\Z$); if we would measure the state right after the QFT, we would also do it in the $\T$ basis.  As a result, the relevant amplitude-expansion of  $|\psi\rangle$ is
    \begin{equation}\label{eq:QFT over Z}
    |\psi\rangle= \int_{\mathbb{T}}\mbox{d}p\  \hat\psi(p) |p\rangle \quad\mbox{ with } \hat\psi(p):=  \sum_{x\in\mathbb{Z}} {\euler^{2\pi i px}} \psi(x).
    \end{equation}
    Some readers may  notice, at this point, that the QFT over $\Z$ is nothing but the inverse Fourier series, also commonly known as the \emph{discrete-time Fourier transform} \cite{oppenheim_Signals_and_Systems}.

     \textbf{QFT over $\boldsymbol{\mathbb{T}}$.}  On the other hand, if the designated basis of ${\cal H}_{\Z}$ is the Fourier basis, states are represented as 
    \begin{equation}
    |\psi\rangle= \int_{\mathbb{T}}\mbox{d}p\  \psi(p) |p\rangle \end{equation} 
    When we say that \emph{the QFT over $\T$ is applied to $|\psi\rangle$}, we mean that (i) the designated basis is changed from the Fourier basis to the standard basis and, therefore, (ii) we re-express the state $|\psi\rangle$  as
    \begin{equation}\label{eq:Fourier coefficients}
    |\psi\rangle= \sum_{x\in\mathbb{Z}} \hat\psi(x) |x\rangle \quad\mbox{ with } \hat\psi(x):=  \int_{\mathbb{T}}\mbox{d}p\ {\euler^{2\pi i px} }\psi(p).
    \end{equation}
    Note that the coefficients $\hat{\psi}(x)$ in (\ref{eq:Fourier coefficients}) are the Fourier coefficients appearing in the Fourier series \cite{oppenheim_Signals_and_Systems}. Due to this fact, one can always regard $\psi(p)$ as a periodic function over the real numbers (with period 1), and identify the QFT over $\T$ as the quantum version of  the Fourier series \cite{oppenheim_Signals_and_Systems}.
    
Lastly, note that the QFT over $\Z$ may only be applied (as an operation in a quantum computation) if the designated basis is the standard basis. Conversely, the QFT over $\T$ may only be applied if the designated basis is the Fourier basis.

\end{itemize}
We now consider the total Hilbert space ${\cal H}_G$ with designated basis ${\cal B}_G$, where $G= G_1\times \dots\times G_m$ with each $G_i$ being a group of the form $\Z_N$, $\Z$ or $\T$.  The total \emph{QFT over $G$} is obtained by applying the QFT over $G_i$, as defined above, to all individual spaces in the tensor product decomposition of ${\cal H}_G$. In particular, application of the QFT over $G$ implies that the designated basis is changed from ${\cal B}_G$ to ${\cal B}_{G'}$, where  $G'=G_i'\times \dots G_m'$ is defined as follows: if $G_i$ has the form $\Z_N$ then $G_i'=G_i$: if $G_i=\Z$ then $G_i':= \T$ and, vice versa, if  $G_i=\Z$ then $G_i':= \T$.

Similarly, one may perform a \emph{partial QFT} by applying the QFT over $G_i$ for a subset of the $G_i$, and leaving the other systems unchanged. The designated basis is changed accordingly on all subsystems where a QFT is applied.

\subsection{Normalizer circuits}\label{sect_normalizer_circuits_sub}

Roughly speaking, a normalizer circuit of size $T$ is a quantum circuit ${\cal C}=U_T\cdots U_1$ composed of $T$ normalizer gates $U_i$. More precisely, the definition is as follows.

A \emph{normalizer circuit over $G=\Z^{a+b}\times\DProd{N}{c}$} acts on a Hilbert space of the form
\begin{equation*}
\mathcal{H}_G=\mathcal{H}_\Z^{a}\otimes\mathcal{H}_\Z^{b}\otimes \left(\mathcal{H}_{N_1}\otimes \cdots \otimes\mathcal{H}_{N_c}\right),
\end{equation*}
with arbitrary  parameters $a$, $b$, $c$, $N_i$. At time $t=0$ (before gates are applied), the designated basis of the computation is the group-element basis ${\cal B}_{G(0)}$ where $G(0)$ is a group from the family (\ref{group_labels_basis}); without loss of generality, we set $G(0)=\Z^a \times \T^b \times \DProd{N}{c}$. The input, output, and gates of the circuit are constrained as follows:
\begin{itemize}
\item {\textbf{Input states}}. The input states are  elements of the  designated group basis $\mathcal{B}_{G(0)}$,  i.e.\, $|g\rangle$ with $g\in G(0)$. The \emph{registers}\footnote{In a quantum computation we call each physical subsystem in (\ref{total_H}) a ``register''.}  $\mathcal{H}_\Z^a$ and $\mathcal{H}_\Z^b$ are initialized to be in standard-basis  $\ket{n}$, $n\in \Z$ and Fourier-basis states  $\ket{p}$, $p\in \T$, respectively\footnote{The results in this paper can be easily generalized to input states that are stabilizer states (section \ref{sect:Stabilizer States}), given that we know the stabilizer group of the state.}. It is assumed that Fourier basis inputs can be prepared within any finite (yet arbitrarily high) level of precision (cf.\  section \ref{sect:Main Result}).
\item\textbf{Gates.} At time $t=1$, the gate $U_1$ is applied, which is either an automorphism gate, quadratic phase gate or a partial QFT over $G(0)$. The designated basis is changed from ${\cal B}_{G(0)}$ to ${\cal B}_{G(1)}$, for some group $G(1)$ in the family (\ref{group_labels_basis}), according to the rules for changing the designated basis as described in section \ref{sect Normalizer gates}.

At time $t=2$, the gate $U_1$ is applied, which is either an automorphism gate, quadratic phase gate or a partial QFT over $G(1)$. The designated basis is changed from ${\cal B}_{G(1)}$ to ${\cal B}_{G(2)}$, for some group $G(2)$.

The gates $U_3, \dots, U_t$ are considered similarly. We denote by ${\cal B}_{G(t)}$ the designated basis after application of $U_t$ (for some group $G(t)$ in the family (\ref{group_labels_basis})), for all $t=3, \dots, T$. Thus, after all gates have been applied, the designated basis is $\mathcal{B}_{G(T)}$.

\item\textbf{Measurement.} After the circuit, a measurement in the final designated basis $\mathcal{B}_{G(T)}$ is performed.
\end{itemize}

For \emph{finite} Abelian groups and their associated Hilbert spaces (i.e. the Hilbert space has the form ${\cal H}_{N_1}\otimes\dots\otimes {\cal H}_{N_a}$), the above definitions of normalizer circuits and normalizer gates specialize to the previously defined notion of normalizer circuits over finite Abelian groups, as done in \cite{VDNest_12_QFTs,BermejoVega_12_GKTheorem}.

\section{Main result}\label{sect:Main Result}

In our main result (theorem \ref{thm:Main Result} below) we show that any polynomial-size normalizer circuit (cf.\ section \ref{sect_normalizer_circuits_sub} for definitions) associated to  any group of the form (\ref{group_labels_basis})  can be simulated efficiently classically. Before stating the result, we will rigorously state what it is meant in our work by an efficient classical simulation of a normalizer circuit, in terms of computational complexity.

In short, the  \textbf{computational problem} we consider is the  following: given a classical description of a normalizer quantum circuit, its input quantum state and the measurement performed at the end of the computation (see section \ref{sect_normalizer_circuits_sub} for details on our computational model), our task is to sample the probability  distribution of final measurement outcomes with a classical algorithm. Any classical algorithm to solve this problem in polynomial time (in the bit-size of the input) is said to be \emph{efficient}.

We specify next how an instance of the computational problem is presented to us.

First, we introduce \textbf{standard encodings} that we use to describe normalizer gates. Our encodings  are  \emph{efficient}, in the sense that the number of bits needed to store a description of a normalizer gate scales as $O(\poly{m},\polylog{N_i})$, where $m$ is the total number of {registers} of the Hilbert space (\ref{total_H}) and $N_i$ are the local dimensions of the finite dimensional registers (the memory size of each normalizer gate in these encodings is given in table \ref{table:Main Result INPUT SIZE}). This polynomial (as opposed to exponential) scaling in $m$ is crucial in our setting, since normalizer gates may act non-trivially on all $m$ registers of the Hilbert space (\ref{total_H})---this is an important difference between our computational model (based on normalizer gates) and the standard quantum circuit model \cite{nielsen_chuang}, where a  quantum circuit is always given as a sequence of one- and two-qubit gates. 
\begin{enumerate}
\item[(i)] A partial quantum Fourier transform $\mathcal{F}_{i}$ over $G_i$ (the $i$th factor of $G$) is described by the index $i$ indicating the register where the gate acts non-trivially.
\item[(ii)] An automorphism gate $U_\alpha$ is described by  what we call a \emph{matrix representation $A$}  of the automorphism $\alpha$ (definition \ref{def:Matrix representation}): an $m\times m$ real matrix $A$ that specifies the action of the map $\alpha$.
\item[(iii)] A quadratic phase gate $D_\xi$ is described by an $m\times m$ real matrix $M$ and an $m$-dim real vector $v$. The pair $(M,v)$ specifies the action of the quadratic function $\xi$ associated to $D_\xi$. Here we exploit a normal form for quadratic functions  given later in theorem \ref{thm:Normal form of a quadratic function}. 
\end{enumerate}
 In our work, we assume that all maps $\alpha$ and $\xi$ can be represented \emph{exactly} by rational matrices and vectors $A$, $M$, $v$, which are explicitly given to us\footnote{Some automorphisms and quadratic functions can only be represented by matrices with irrational entries (cf.\ the normal forms in sections \ref{sect:Homomorphisms and matrix representations},\ref{sect:quadratic_functions}). Restricting ourselves to study the rational ones allows us to develop \emph{exact simulation algorithms}. We believe irrational matrices (even with transcendental entries) could also be handled by taking into account floating-point errors. We highlight that the stabilizer formalism in this paper and all of our normal forms are developed \emph{\textbf{analytically}}, and hold even if transcendental numbers appear in the matrix representations of $\alpha$ and $\xi$. (It is an good question to explore whether an exact simulation results may hold for matrices with algebraic coefficients.)}.
 
Second, a normalizer circuit is specified as a list of normalizer gates given to us in their standard encodings.

The efficiency of our standard encodings relies strongly on results presented in sections \ref{sect:Homomorphisms and matrix representations} and \ref{sect:quadratic_functions}. In section \ref{sect:Homomorphisms and matrix representations}, we develop a (classical) \textbf{ theory of matrix representations} of group homomorphisms, proving their existence (lemma \ref{lemma:existence of matrix representations}) and providing a normal form that characterizes the structure of these matrices (lemma \ref{lemma:Normal form of a matrix representation}). In section \ref{sect:quadratic_functions}, we develop analytic normal forms for bicharacter functions (lemmas \ref{lemma:Normal form of a bicharacter 1}, \ref{lemma symmetric matrix representation of the bicharacter homomorphism}) and quadratic functions (theorem \ref{thm:Normal form of a quadratic function}). These results are also main contributions of our work. 

We would like to highlight, in particular, that our \textbf{normal form for quadratic functions} (theorem \ref{thm:Normal form of a quadratic function}) should be of interest to a quantum audience. It was recently shown  in \cite{BermejoVega_12_GKTheorem} that  quadratic functions over an Abelian group describe the quantum wave-functions of the so-called \emph{stabilizer states}. Hence, our normal form can be used to characterize the complex amplitudes of such states\footnote{As mentioned in the introduction, the result in \cite{BermejoVega_12_GKTheorem} is for  stabilizer states over finite dimensional Hilbert spaces but it can be easily generalized.}.

Lastly, we mention that, in our work, we allow the matrices $A$, $M$ and the vector $v$ in (i-iii) to contain \emph{arbitrarily large \emph{and} arbitrarily small} coefficients. This degree of generality is necessary in the setting we consider, since we allow \emph{all} normalizer gates to be valid components of a normalizer circuit. However,  the presence of infinite groups in (\ref{group_labels_basis}) implies that  that there exists an infinite number of normalizer gates (namely, of automorphism and quadratic gates, which follows from the our  analysis in sections  \ref{sect:Homomorphisms and matrix representations} and \ref{sect:quadratic_functions}). This is in contrast with the settings considered in \cite{VDNest_12_QFTs,BermejoVega_12_GKTheorem}, where both the group (\ref{group_hilbert_space}) and the associated set of  normalizer gates are finite. As a result, the arithmetic precision needed to store the coefficients of $A$, $M$, $v$ in our standard encodings becomes a variable of the model (just like in the standard problem of multiplying two integer matrices).

We state now our main result.
\begin{theorem}[\textbf{Main result}]\label{thm:Main Result} Let $\mathcal{C}$ be any normalizer circuit over any group $G=\Z^a \times \Z^b \times \DProd{N}{c}$ as defined in section \ref{sect_normalizer_circuits_sub}. Let $\mathcal{C}$ act on a input state $\ket{g}$ in the designated standard basis at time zero, and be followed by a final  measurement in  the  designated basis at time $T$. Then the output probability distribution can be sampled classically  efficiently  using epsilon-net methods. 
\end{theorem}
We remind that, in theorem \ref{thm:Main Result}, both standard and Fourier basis states of $\mathcal{H}_{\Z}$ are allowed inputs (cf.\ section \ref{sect_normalizer_circuits_sub}). 

In theorem \ref{thm:Main Result}, the state $\ket{g}$ is described by the group element $g$, which is encoded as a tuple of $m$ rational\footnote{In this work we do not use floating point arithmetic.} numbers of varying size (see table \ref{table:Main Result INPUT SIZE}, row 1).  The memory needed to store the normalizer gates comprising $\mathcal{C}$ is summarized in table \ref{table:Main Result INPUT SIZE}, row 2. By ``\emph{classically efficiently}'' it is meant that there exists a classical algorithm (\textbf{theorem \ref{thm:Algorithm to sample subgroups}}) to perform the given task whose worst-time running time scales \emph{polynomially} in the input-size (namely, in the number of subsystems $m$, the number of normalizer gates of $\mathcal{C}$) and of all other variables listed in the ``bits needed'' column of table \ref{table:Main Result INPUT SIZE}), and \emph{polylogarithmically}  in the parameters ${\frac{1}{\varepsilon}}$, ${\Delta}$ that specify the number of points in a \emph{$(\Delta, \varepsilon)$-net} (read below) and their geometrical arrangement. 
\begin{table}[h]
\begin{center}   
  \begin{tabular}{| c | c |  c|}
           \hline 
           Input element & Description needs to  & Bits needed\\ 
           \hline \multirow{2}{*}{Input state $\ket{g}$}  & Specify element $g(i)$ of infinite group $\Z$, $\T$ & variable\\
          & Specify element $g(j)$ of finite group $\Z_{N_j}$ & $\log N_j$   \\ 
          \hline \multirow{3}{*}{Normalizer circuit $\mathcal{C}$}  & Specify  quantum Fourier transform $\mathcal{F}_i$ & $\log{m}$ \\ 
          & Specify automorphism gate $U_\alpha$ & $m^2 \|A\|_{\mathbf{b}}$  \\
          & Specify quadratic phase gate $D_\xi$ & $ m^2 \|M\|_{\mathbf{b}}+m \|v\|_{\mathbf{b}}$  \\ \hline     
    \end{tabular}
  \end{center}
   \caption{The input-size in theorem \ref{thm:Main Result}. $\|X\|_{\mathbf{b}}$ denotes the number of bits used to store one single coefficient of $X$, which is always assumed to a rational matrix/vector. Formulas in column 3 are written in Big Theta $\Theta$ notation and do not  include constant factors (which are anyway small).} \label{table:Main Result INPUT SIZE}
\end{table}

\subsubsection*{Sampling techniques}

We finish this section by saying a few words about the $(\Delta, \epsilon)$-net methods used in the proof of theorem \ref{thm:Main Result}. These techniques are fully developed in sections \ref{sect:Proof of theorem 1} and \ref{sect:Sampling the support of a state}.

We shall show later (lemma \ref{lemma:StabStates_are_Uniform_Superpositions}) that the final quantum state $\ket{\psi}$ generated by a normalizer circuit is always a uniform quantum superposition in any of the computational basis we consider (\ref{basis}): if $G$ is the group associated to our designated basis, and if $X$ is the set of $x\in G$ such that $\psi (x)\neq 0$, then $|\psi(x)|=|\psi(y)|$ for all $x,y\in X$.  As a result the final distribution of measurement outcomes is always a flat distribution over some set $X$. 

Moreover, we show in section that $X$ is always isomorphic to a group of the form $K \times \Z^\mathbf{r}$ where $K$ is compact, and that such isomorphism can be efficiently computed: as a result, we see that, although $X$ is not compact, the non-compact component of $X$ has an Euclidean geometry. Our sampling algorithms are based on this fact: to sample $X$ in an approximate sense, we construct a subset $\mathcal{N}_{\Delta,\varepsilon} \subset X$ of the form
\begin{equation}
\mathcal{N}_{\Delta,\varepsilon} =\mathcal{N}_{\varepsilon}\oplus\mathcal{P}_\Delta,
\end{equation}
where $\mathcal{N}_{\varepsilon}$ is an $\varepsilon$-net (definition \ref{def:Epsilon Net}) of the compact component $K$ of $X$ and $\mathcal{P}_\Delta$ is a $\mathbf{r}$-dimensional parallelotope contained in the Euclidean component $\Z^\mathbf{r}$, centered at 0, with edges of length $2\Delta_1,\ldots, 2\Delta_\mathbf{r}$. We call $\mathcal{N}_{\Delta,\varepsilon}$ a ($\Delta, \varepsilon$)-net  (definition \ref{def:Delta Epsilon Net}). The algorithm in theorem \ref{thm:Main Result} can efficiently construct and sample such sets for any $\varepsilon$ and $\Delta:=\Delta_1,\ldots, \Delta_\mathbf{r}$ of our choice: its worse running-time is $O(\polylog{\tfrac{1}{\varepsilon}},\polylog{\Delta_i}) $,  as a function of these parameters. We refer the reader to section \ref{sect:Proof of theorem 1} and theorem \ref{thm:Algorithm to sample subgroups} for more details.

\subsubsection*{Treatment of finite-squeezing errors}

It follows from the facts that we have just dicussed that when $G$ is not a compact group (i.e.\, $G$ contains $\Z$ primitive factors) the support $X$ of the quantum state $\ket{\psi}$ can be an unbounded set. In such case, it follows from the fact that $\ket{\psi}$ is a uniform superposition that the quantum state is \emph{unphysical} and that the physical preparation of such a state requires infinite energy; in the continuous-variable quantum information community, states like $\ket{\psi}$ are often called \emph{infinitely squeezed states} \cite{KokLovett10Intro_to_Optical_Quantum_Information_Processing}. In a physical implementation (cf.\  \cite{BermejoLinVdN13_BlackBox_Normalizers}), these states can be replaced by physical finitely-squeezed states, whose amplitudes will decay towards the infinite ends\footnote{The particular form of the damping depends on the implementation. These effects vanish in the limit of infinite squeezing.}  of the support set $X$. This leads to finite-squeezing errors, compared to the ideal scenario.

In this work, we consider normalizer circuits to work perfectly in the ideal infinite-squeezing scenario. Our simulation algorithm in theorem \ref{thm:Main Result} samples the ideal distribution that one would obtain in the infinite precision limit, neglecting the presence of  finite-squeezing errors. This is achieved with the  $(\Delta, \epsilon)$-net methods described above, which we use to discretize and sample the  manifold $X$ that supports the ideal output state $\ket{\psi}$;  the output of this procedure  reveals the information  encoded in the  wavefunction of the state.

In our upcoming work \cite{BermejoLinVdN13_BlackBox_Normalizers}, we  make use of this simulation algorithm to study quantum algorithms based on normalizer circuits. We also study  in \cite{BermejoLinVdN13_BlackBox_Normalizers}  how  information can be represented with finitely-squeezed states in a computation.

\section{Group and Character Theory}\label{sect:Group Theory}

\subsection{Elementary Abelian groups}\label{sect:elementary_Abelian_groups}

A group of the form \be G=\Z^a \times \mathbb{R}^b\times  \DProd{N}{c} \times \T^d\ee will be called an \emph{elementary Abelian group}. We will often use the shorthand notation $F= \DProd{N}{c}$ for the finite subgroup in the above decomposition. Note that the Hilbert space formalism introduced in section \ref{sect:Quantum states over infinite Abelian groups} does not refer to  groups of the form $\mathbb{R}^b$ (i.e. direct products of the group of real numbers). However for some of our calculations in later sections, it will be convenient to include these types of groups in the analysis.

An elementary Abelian group of the form $\Z$, $\R$, $\T$ or $\Z_N$ is said to be primitive. Thus every elementary Abelian group can be written as $G=G_1\times\dots \times G_m$ with each $G_i$ primitive; we will often use this notation. We will also use the notation $G_{\Z}$, $G_{\R}$, $G_{F}$, $G_{\T}$ to denote elementary Abelian groups that are, respectively, integer lattices $\Z^a$, real lattices $\R^b$, finite groups $F$ and tori $\T^d$. We will also assume that the  factors $G_i$ of $G$ are arranged so that $G=G_{\Z}\times G_{\R}\times G_{F}\times G_{\T}$.

Next we introduce the notion of group characteristic $\charac{G}$ for primitive groups:
\begin{equation}
\charac{\Z}:=0, \quad \charac{\R}:=0,\quad \charac{\Z_N}:=N, \quad \charac{\T}:=1.
\end{equation}
Alternatively, we can also define the characteristic as follows:  $\charac{G}$ is the number that coincides with (a) the order of $1$ in $G$ if 1 has finite order (which is the case for $\Z_N$ and $\T$); (b) zero, if $1$ has infinite order in $G$ (which is the case for $\Z$ and $\R$).

Consider an elementary Abelian group $G= G_1\times\dots \times G_m$ where $c_i$ is the characteristic of $G_i$. Each element $g\in G$ can be represented as an $m$-tuple $g=(g_1, \dots, g_m)$ of real numbers. If $x=(x_1, \dots, x_m)$ is an arbitrary $m$-tuple of real numbers, we say that $x$ is congruent to $g$, denoted by  $x\equiv g$ $(\mbox{mod } G)$, if \be x_i \equiv  g_i\  (\mbox{mod } c_i) \quad \mbox{ for every } i=1, \dots, m.\ee
For example, every string of the form $x=(\lambda_1c_1, \dots, \lambda_mc_m)$ with $\lambda_i\in \Z$ is congruent to $0\in G$.

\subsection{Characters}\label{sect:characters}

\begin{definition}[\textbf{Character} \cite{Morris77_Pontryagin_Duality_and_LCA_groups,Moreno05_Analytic_Number_Theory_L_Functions}] \label{def:Characters} Let $G$ be an elementary Abelian group. A character of $G$ is  a continuous homomorphism $\chi$ from  $G$  into the group $U(1)$ of  complex numbers with modulus one. Thus $\chi(g+h)=\chi(g)\chi(h)$ for every $g, h\in G$.
\end{definition}
If $G$ is an elementary Abelian group, the set of all of its characters is again an elementary Abelian group, called the dual group of $G$, denoted by $\widehat{G}$. Moreover, $\widehat{G}$ is isomorphic to another elementary Abelian group, according to the following rule:
\begin{equation}\label{eq:Elementary Character Group}
G=\R^a\times \T^b \times \Z^c \times F \qquad \longrightarrow\qquad  \widehat{G}\cong\R^a\times \Z^b \times \T^c \times F.
\end{equation}
Thus, in particular, $\widehat{\mathbb{R}}$ is isomorphic to $\mathbb{R}$ itself and similarly $\widehat{F}$ is isomorphic to $F$ itself; these groups are called autodual. On the other hand, $\widehat{\mathbb{Z}}$ is isomorphic to $\mathbb{T}$ and, conversely, $\widehat{\mathbb{T}}$ is isomorphic to $\mathbb{Z}$. We also note from the rule (\ref{eq:Elementary Character Group}) that the dual group of $\widehat{G}$ is isomorphic to $G$ itself. This is a manifestation of the Pontryagin-Van Kampen duality \cite{Morris77_Pontryagin_Duality_and_LCA_groups,Stroppel06_Locally_Compact_Groups,Dikranjan11_IntroTopologicalGroups}.

We now give explicit formulas for the characters of any primitive Abelian group.
\begin{itemize}
\item The characters of $\mathbb{R}$ are
\begin{equation}\label{char_R}
 \chi_x(y):=\exp{\left(2\pi i xy\right)}, \quad\text{for every $x$, $y\in\R$}.
\end{equation}
Thus each character is labeled by a real number. Note that $\chi_x\chi_{x'}= \chi_{x+x'}$ for all $x, x'\in \mathbb{R}$. The map $x\to \chi_x$ is an isomorphism from $\R$ to $\widehat{\R}$, so that $\R$ is autodual.
\item The characters of $\mathbb{Z}_N$ are
\begin{equation}\label{char_F}
\chi_x(y):=\exp{\left(\frac{2\pi i}{N}\ xy\right)}, \quad\text{for every $x$, $y\in\mathbb{Z}_N$}.
\end{equation}
Thus each character is labeled by an element of $\mathbb{Z}_N$. As above, we have $\chi_x\chi_{x'}= \chi_{x+x'}$ for all $x, x'\in\mathbb{Z}_N$. The map $x\to \chi_x$ is an isomorphism from $\mathbb{Z}_N$ to $\widehat{\mathbb{Z}}_N$, so that $\mathbb{Z}_N$ is autodual.
\item  The characters of $\Z$ are
\begin{equation}\label{char_Z}
\chi_p(m):=\exp{(2\pi i pm)}, \quad\text{for every $p\in\T$, $m\in\Z$},
\end{equation}
Each character is labeled by an element of $\mathbb{T}$. Again we have $\chi_p\chi_{p'}= \chi_{p+p'}$ for all $p, p'\in\mathbb{T}$ and the map $p\to \chi_p$ is an isomorphism from $\mathbb{T}$ to $\widehat{\mathbb{Z}}$.
\item The characters of $\T$ are
\begin{equation}\label{char_T}
\chi_m(p):=\exp{(2\pi i pm)}, \quad\text{for every $p\in\T$, $m\in\Z$};
\end{equation}
Each character is labeled by an element of $\mathbb{Z}$. Again we have $\chi_m\chi_{m'}= \chi_{m+m'}$ for all $m, m'\in\mathbb{Z}$ and the map $m\to \chi_m$ is an isomorphism from $\mathbb{Z}$ to $\widehat{\mathbb{T}}$.
\end{itemize}
If $G$ is a general elementary Abelian group, its characters are obtained by taking products of the characters described above. More precisely, if $A$ and $B$ are two elementary Abelian groups, the character group of $A\times B$ consists of all products $\chi_A\chi_B$ with $\chi_A\in \widehat{A}$ and $\chi_B\in \widehat{B}$, and where $\chi_A\chi_B(a, b):= \chi_A(a)\chi_B(b)$ for every $(a, b)\in A\times B$. To obtain all characters of a group $G$ having the form (\ref{eq:Elementary Character Group}), we denote
\be G^*:=\R^a\times \Z^b \times \T^c \times F.\ee Considering an arbitrary element \be\label{mu} \mu= (r_1, \dots, r_a ,z_1, \dots, z_b, t_1, \dots, t_c, f_1, \dots, f_d)\in G^*,\ee the associated character is given by the product \be \label{mu_character} \chi_{\mu}:=\chi_{r_1}\dots\chi_{r_a}\ \chi_{z_1}\dots\chi_{z_b}\ \chi_{t_1}\dots\chi_{t_c}\ \chi_{f_1}\dots\chi_{f_d}\ee where the individual characters $\chi_{r_i}, \chi_{z_j}, \chi_{t_k}, \chi_{f_l}\dots$ of $\mathbb{R}, \mathbb{Z}, \mathbb{T}$ and $\mathbb{Z}_{N_l}$  are defined above. The character group of $G$ is given by \be \widehat{G}= \{\chi_{\mu}: \mu\in G^*\}.\ee
The rule (\ref{eq:Elementary Character Group}) immediately implies that $(G^*)^* = G$. This implies that the character group of $G^*$ is $\{\chi_{g}: g\in G\}$, where $\chi_g$ is defined in full analogy with (\ref{mu_character}). Two elementary but important  
features are the following:
\begin{lemma}\label{lemma:Pontryagin duality for characters}
For every $g\in G$ and $\mu\in G^*$ we have
\be \chi_{\mu}(g)= \chi_g(\mu).\ee
\end{lemma}
\begin{lemma}\label{lemma:Character Multiplication}
For every $\mu,\nu\in G^*$ and every $g\in G$ it follows that
\begin{equation}
\chi_{\mu+\nu}(g)=\chi_{\mu}(g)\chi_{\nu}(g)
\end{equation}
\end{lemma}
Both lemmas \ref{lemma:Pontryagin duality for characters} follow from inspection of the characters of $\mathbb{R}, \mathbb{Z}, \mathbb{T}$ and $\mathbb{Z}_{N}$ defined in (\ref{char_R})-(\ref{char_T}). The lemmas also reflect the strong duality between $G$ and $G^*$.

 Finally, the definition of every character function $\chi_{a}(b)$ as given in (\ref{char_R})-(\ref{char_T}), which is in principle defined for $a$ in  $\R, \Z_N, \Z, \T$  and $b$ in $\R, \Z_N, \T, \Z$, respectively, can be readily extended to the entire domain of real numbers, yielding functions $\chi_x(y)$ with $x ,y\in \R$.  
Consequently, the character functions (\ref{mu_character}) of general elementary Abelian groups $G=G_1\times \dots\times G_m$ can also be extended to a larger domain, giving rise to functions $\chi_x(y)$ where $x, y\in \R^m$. With this extended notion, we have the following basic property:

\begin{lemma}\label{thm_extended_characters}
Let $g\in G$ and $\mu\in G^*$. For every  $x, y\in \R^m$ such that $x\equiv g$ $(\mbox{mod } G)$ and $y\equiv \mu$ $(\mbox{mod } G^*)$, we have
\be
\chi_y(x)=\chi_{\mu}(g).
\ee
\end{lemma}
The proof is easily given for primitive groups, and then extended to general elementary Abelian groups.

\subsection{Simplifying characters via the bullet group}

Let $G= G_1\times \dots \times G_m$ be an elementary Abelian group where each $G_i$ is of primitive type. Recall that in section \ref{sect:characters} we have introduced the definition of the elementary Abelian group $G^*$, which is isomorphic to the character group of $G$. Here we define another group $G^{\bullet}$ (which is Abelian but not elementary), called the bullet group of $G$. The bullet group is isomorphic to $G^*$ (and hence also to the character group of $G$) and is mainly introduced for notational convenience, as working with $G^{\bullet}$ will turn out to simplify calculations. The bullet group is defined to be $G^\bullet= G_1^{\bullet}\times \dots\times G_m^{\bullet}$, where
\begin{align}
\Z_N^{\bullet} &:= \left\{ 0, \frac{1}{N}, \frac{2}{N}, \dots, \frac{N-1}{N}  \bmod{1}\right\},\notag\\
\R^{\bullet} &:= \R^*= \R; \quad \Z^{\bullet} := \Z^* = \T; \quad
\T^{\bullet} := \T^* = \Z.\label{eq:Bullet Group}
\end{align}
Thus the only difference between the groups $G^*$ and $G^\bullet$ is in the $\Z_{N_i}$ components. The groups $G^*$ and $G^\bullet$ are isomorphic via the ``bullet map''
\begin{equation}
\mu\in G^*\to \mu^\bullet:= (\mu_1^\bullet, \dots, \mu_m^\bullet) \in G^\bullet,
\end{equation}
where $\mu_i^\bullet:= \mu_i/N$ if $\mu_i\in \Z_N$ and $\mu_i^\bullet = \mu_i$ if $\mu_i$ belongs to either $\R, \Z$ or $\T$. The bullet map is easily seen to be an isomorphism.

One of the main purposes for introducing the bullet group is to simplify calculations with characters. In particular, for every $g\in G$ and $\mu\in G^*$ we have
\be\label{eq:definition of bullet map}
\chi_{\mu}(g)=\exp{\left(2\pii\, \sum_{i=1}^m \mu_i^\bullet g_i \right)}.
\ee

\subsection{Annihilators}\label{sect:Annihilators}

Let $G$ be an elementary Abelian group. Let $X$ be any subset of $G$. The\emph{ annihilator} $X^\perp$ is the subset
\begin{equation}
X^\perp:=\{\mu\in G^* :\: \chi_\mu(x)=1\textnormal{ for every }x\in X\}.
\end{equation}
We can define the annihilator $Y^\perp$ of a subset $Y\subseteq G^*$ analogously as \be Y^\perp:=\{x\in G :\: \chi_\mu(x)=1\textnormal{ for every }\mu\in Y\}.\ee By combining the two definitions it is possible to define  double annihilator sets $X^\Perp:= (X^{\perp})^{\perp}$, which is a subset of the initial group $G$, for every set $X\subseteq G$. Similarly,   $Y^{\Perp}\subseteq G^*$  for every $Y\subseteq G^*$. The following lemma states that $X$ and $X^\Perp$ are, in fact, identical sets \emph{iff} $X$ is closed as a set, and related to each other in full generality:
\begin{lemma}[\cite{Stroppel06_Locally_Compact_Groups}]\label{lemma:Annihilator properties}
Let $X$ be an arbitrary subset of an elementary Abelian group $G$. Then the following assertions hold:
\begin{enumerate}
\item[(a)] The annihilator $X^\perp$  is a closed subgroup of $G^*$ (and  $X^\Perp$ is a closed subgroup of $G$).
\item[(b)] $X^\Perp$  is the smallest closed subgroup of $G$ containing $X$.
\item[(c)] If $Y$ is a subset of $G$ such that $X\subseteq Y$ then $X^\perp \supseteq Y^\perp$ and $X^\Perp\subseteq Y^\Perp$.
\end{enumerate}
\end{lemma}
We mention that in the quantum computation literature (see e.g. \cite{Brassard97_Exact_Simons_Algorithm,lomont_HSP_review,VDNest_12_QFTs,BermejoVega_12_GKTheorem}) the annihilator $H^\perp$ of a subgroup $H$ is more commonly known as the \emph{orthogonal subgroup} of $H$.

Let $\alpha:G\rightarrow H$ be a continuous group homomorphism between two elementary Abelian groups $G$ and $H$. Then \cite[proposition 30]{Morris77_Pontryagin_Duality_and_LCA_groups} there exists a unique continuous group homomorphism $\alpha^*:H^*\rightarrow G^*$, which we call the \textbf{dual homomorphism} of $\alpha$,  defined via the equation
\begin{equation}\label{eq:Dual Automorphism DEF}
\chi_{\alpha^*(\mu)}(g)=\chi_\mu(\alpha(g)).
\end{equation}
Note that $\alpha^{**}=\alpha$ by duality.

\section{Homomorphisms and matrix representations}\label{sect:Homomorphisms and matrix representations}

\subsection{Homomorphisms}\label{sect:Homomorphisms}
 
Let $G=G_1\times \ldots \times G_m$ and $H=H_1\times \ldots \times H_n$ be two elementary finite Abelian groups, where $G_i$, $H_j$ are primitive subgroups. As discussed in section \ref{sect:elementary_Abelian_groups}, we assume that the $G_i$ and $H_j$ are ordered so that $G=G_{\Z}\times G_{\R}\times G_{F}\times G_{\T}$
and $H=H_{\Z}\times H_{\R}\times H_{F}\times H_{\T}$.

Consider a continuous group homomorphism $\alpha: G\to H$. Let $\alpha_{\mathbb{Z}\mathbb{Z}}: {G_{\Z}}\to {G_{\Z}}$ be the map obtained by restricting the input and output of $\alpha$ to ${H_{\Z}}$. More precsiely, for $g\in G_{\Z}$ consider the map \be (g, 0, 0, 0)\in G \to \alpha(g, 0, 0, 0)\in H\ee and define $\alpha_{\mathbb{Z}\mathbb{Z}}(g)$ to be the $G_{\Z}$-component of $\alpha(g, 0, 0, 0)$. The resulting map $\alpha_{\mathbb{Z}\mathbb{Z}}$ is a continuous homomorphism from $\Z$ to $\Z$. Analogously, we define the continuous group homomorphisms $\alpha_{XY}: G_Y\rightarrow H_X$ with $X, Y = \Z, \R, \T, F$. It follows that, for any $g=(z, r, f, t)\in G$, we have
\be\label{eq:block decomp Homomorphism}
\alpha(g) = \begin{pmatrix}
      \alpha_{\Z\Z}(z) + \alpha_{\Z\R}(r) + \alpha_{\Z F}(f) + \alpha_{\Z\T}(t) \\
      \alpha_{\R\Z}(z) + \alpha_{\R\R}(r) + \alpha_{\R F}(f) + \alpha_{\R\T}(t) \\
      \alpha_{F\Z}(z) + \alpha_{F\R}(r) + \alpha_{FF}(f) + \alpha_{F\T}(t) \\
      \alpha_{\T\Z}(z) + \alpha_{\T\R}(r) + \alpha_{\T F}(f) + \alpha_{\T\T}(t)
    \end{pmatrix} \leftrightarrow \begin{pmatrix}
      \alpha_{\Z\Z} & \alpha_{\Z\R} & \alpha_{\Z F} & \alpha_{\Z\T} \\
      \alpha_{\R\Z} & \alpha_{\R\R} & \alpha_{\R F} & \alpha_{\R\T} \\
      \alpha_{F\Z} & \alpha_{F\R} & \alpha_{FF} & \alpha_{F\T} \\
      \alpha_{\T\Z} & \alpha_{\T\R} & \alpha_{\T F} & \alpha_{\T\T}
    \end{pmatrix}
    \begin{pmatrix}
    z\\ r\\f\\t
    \end{pmatrix}
\ee
 
$\alpha$ is therefore naturally identified with the the $4\times 4$ ``matrix of maps'' given in the r.h.s of (\ref{eq:block decomp Homomorphism}).

The following lemma (see e.g. \cite{PrasadVemuri08_classification_heisenberg_groupss} for a proof) shows that homomorphisms between elementary Abelian groups must have a particular block structure.
\begin{lemma}\label{corollary:block-structure of Group Homomorphisms} Let $\alpha: G\to H$ be a continuous group homomorphism. Then $\alpha$ has the following block structure
\begin{equation}\label{eq:block decomposition of a homomorphism}
    \alpha \leftrightarrow
    \begin{pmatrix}
      \alpha_{\Z\Z} & \mathpzc{0}  &\mathpzc{0}  &  \mathpzc{0} \\
      \alpha_{\R\Z} & \alpha_{\R\R} &  \mathpzc{0} &  \mathpzc{0} \\
      \alpha_{F\Z} &  \mathpzc{0} & \alpha_{FF} &  \mathpzc{0} \\
      \alpha_{\T\Z} & \alpha_{\T\R} & \alpha_{\T F} & \alpha_{\T\T}
    \end{pmatrix}
\end{equation}
where $\mathpzc{0}$ denotes the trivial group homomorphism.
\end{lemma}
 
The lemma shows, in particular, that there are no non-trivial continuous group homomorphisms between certain pairs of primitive groups: for instance, continuous groups cannot be mapped into discrete ones, nor can finite groups be mapped into zero-characteristic groups.

\subsection{Matrix representations}\label{sect:Matrix Representations}

\begin{definition}
[\textbf{Matrix representation}]\label{def:Matrix representation}
Consider elementary Abelian groups $G= G_1\times\dots \times G_m$ and $H= H_1\times\dots \times H_n$ and a group homomorphism $\alpha:G\rightarrow H$. A \emph{matrix representation of $\alpha$} is an $n\times m$ real matrix  $A$ satisfying the following property:
 \be\label{def_matrix_rep} \alpha (g)\equiv Ax \ (\mbox{mod } H) \quad \mbox{ for every } g\in G \mbox{ and } x\in \R^m \mbox{ satisfying } x\equiv g \ (\mbox{mod } G)\ee
Conversely, a real $n\times m$ matrix $A$ is said to define a group homomorphism if there exists a group homomorphism $\alpha$ satisfying (\ref{def_matrix_rep}).
\end{definition}
It is important to highlight that in the definition of matrix representation we impose that the identity $\alpha(g)= Ax \pmod{H}$ holds in a very general sense: the output of the map must be equal for inputs $x,\,x'$ that are \emph{different} as strings of real numbers but correspond to the \emph{same} group element $g$ in the group $G$. In particular, all strings that are congruent to zero in $G$ must be mapped to strings congruent to zero in $H$. Though these requirements are (of course) irrelevant when we only consider groups of zero characteristic (like $\Z$ or $\R$), they are crucial when quotient groups are involved (such as $\Z_N$ or $\T$).

As a simple example of a matrix representation, we consider the bullet map\footnote{Strictly speaking, definition \ref{def:Matrix representation} cannot be applied to the bullet map, since $G^\bullet$ is not an elementary Abelian group. However the definition is straightforwardly extended to remedy this.}, which is an isomorphism from $G^*$ to $G^\bullet$ . Define the diagonal $m\times m$ matrix $\Upsilon$ with diagonal entries defined as
\begin{equation}\label{eq:matrix representation of bullet isomorphism}
\Upsilon(i,i)=
\begin{cases}
    1/N_i &\textrm{if ${G_i}= \Z_{N_i}$ for some $N_i$}, \\
   1  &\textrm{otherwise}.
\end{cases}
\end{equation}
It is easily verified that $\Upsilon$ satisfies the following property: for every $\mu\in G^*$ and $x\in \R^m$ satisfying $x\equiv \mu $ $(\mbox{mod } G^*)$, we have
\begin{equation}\label{eq:Upsilon is a Mat Rep of Bullet}
\mu^\bullet\equiv\Upsilon x \pmod{G^\bullet}.
\end{equation}
Note that, with the definition of $\Upsilon$, equation (\ref{eq:definition of bullet map}) implies
\begin{equation}\label{eq:definition of bullet map2}
\chi_{\mu}(g)=\exp{\left(2\pii\, \sum_{i=1}^m \mu^\bullet(i) g(i) \right)} = \exp{\left(2\pii\, \mu^T \Upsilon g\right)}.
\end{equation}
Looking at equation  (\ref{eq:Upsilon is a Mat Rep of Bullet}) coefficient-wise, we obtain a relationship  
$\mu^\bullet(i)\equiv\frac{x(i)}{N_i}\pmod{1}$ for each factor $G_i$ of the form $\Z_{N_i}$; other factors are left unaffected by the bullet map. From this expression it is easy to derive that  $\Upsilon^{-1}$ is a matrix representation of the inverse of the bullet map\footnote{We ought to highlight that the latter is by no means a general property of matrix representations. In fact, in many cases, the matrix-inverse $A^{-1}$ (if it exists) of a matrix representation $A$ of  a group isomorphism is not a valid matrix representation of a group homomorphism. (This happens, for instance, for all group automorphisms of the group $\Z_N$ that are different from the identity.)  In lemma \ref{lemma:Normal form of a matrix representation} we characterize which matrices are valid matrix representations. Also, in section \ref{sect:Computin Inverses} we discuss the problem of computing  matrix representations of  group automorphisms.}, i.e.\ the group isomorphism $\mu^\bullet\rightarrow\mu \pmod{G^*}$.

The next lemma (see appendix \ref{appendix:Supplement to section Homomorphisms} for a proof) summarizes some useful properties of matrix representations.
\begin{lemma}[\textbf{Properties of matrix representations}]\label{lemma:properties of matrix representations}

Let $G$, $H$, $J$ be elementary  Abelian groups, and $\alpha :G\rightarrow H$ and $\beta: H\rightarrow J$ be group homomorphisms with matrix representations $A,\, B$, respectively. Then it holds that
\begin{itemize*}
\item[(a)] $BA$ is a matrix representation of the composed homomorphism $\beta\circ\alpha$;
\item[(b)] The matrix $A^*:=\Upsilon_{G}^{- 1} A^\transpose \, \Upsilon_H $ is a matrix representation of the dual homomorphism  $\alpha^*$, where  $\Upsilon_X$ denotes the matrix representation of the bullet map $X^*\rightarrow X^\bullet$.
\end{itemize*}
\end{lemma}

As before, let $G=G_1\times\dots\times G_m$ be an elementary Abelian group with each $G_i$ of primitive type. Let \be e_i= (0, \dots, 0, 1, 0, \dots, 0)\ee denote the $i$-th canonical basis vector  of $\mathbb{R}^m$. If we regard $g\in G$ as an element of $\R^m$, we may write $g= \sum g(i) e_i$. Note however that $e_i$ may not belong to $G$ itself. In particular, if $G_i = \T$ then $e_i\notin \T$ (since $1\notin \T$ in the representation we use, i.e.\ $\T= [0, 1)$).

\begin{lemma}[\textbf{Existence of matrix representations}]\label{lemma:existence of matrix representations} Every group homomorphism $\alpha: G\rightarrow H$ has a matrix representation $A$. As a direct consequence, we have $\alpha(g)\equiv \sum_i g(i)Ae_i$ $(\mbox{mod }H)$, for every $g=\sum_i g(i)e_i\in G$.
\end{lemma}
The last property of lemma \ref{lemma:existence of matrix representations} is remarkable, since the coefficients $g(i)$ are real numbers when $G_i$ is of the types $\R$ and $\T$. We give a proof of the lemma in appendix \ref{appendix:Supplement to section Homomorphisms}.

We finish this section by giving a normal form for matrix representations and characterizing  which types of matrices constitute valid matrix representations as in definition \ref{def:Matrix representation}.
\begin{lemma}[\textbf{Normal form of a matrix representation}]\label{lemma:Normal form of a matrix representation}
Let $G = G_1\times\dots\times G_m$ and $H=H_1\times \dots\times H_n$ be elementary Abelian groups. Let $c_j, c_j^*, d_i$ and $d_i^*$  denote the characteristic of $G_j, G_j^*, H_i$ and $H_i^*$, respectively. Define $\mathbf{Rep}$ to be the subgroup of all $n\times m$   real matrices  that have integer coefficients in those rows $i$ for which $H_i$ has the form  $\Z$ or $\Z_{d_i}$. A real $n\times m$ matrix $A$ is a valid matrix representation of some group homomorphism $\alpha: G\rightarrow H$ iff  $A$ is an element of $\mathbf{Rep}$ fulfilling two (dual) sets of consistency conditions:
\begin{align}\label{eq:Consistency Conditions Homomorphism}
c_j A(i,j) = 0 \mod d_i,\qquad {d}_i^* A^*(i,j) = 0 \mod c^*_j, 
\end{align}
for every $i=1,\ldots,n$, $j=1,\ldots, m$, and being $A^*$  the $m\times n$ matrix defined in lemma \ref{lemma:properties of matrix representations}(b). 
Equivalently, $A$ must be of the form 
\begin{equation}\label{eq:block-structure of matrix representations}
A:= \begin{pmatrix}
      A_{\Z\Z} & 0 & 0 & 0 \\
      A_{\R\Z} & A_{\R\R} & 0 & 0 \\
      A_{F\Z} & 0 & A_{FF} & 0 \\
      A_{\T\Z} & A_{\T\R} & A_{\T F} & A_{\T\T}.
    \end{pmatrix}
\end{equation}
with the following restrictions:
  \begin{enumerate}
  \item $A_{\Z\Z}$ and $A_{\T\T}$ are arbitrary  integer matrices.
  \item $A_{\R\Z}$, $A_{\R\R}$ are arbitrary  real matrices.
  \item $A_{F\Z}$, $A_{FF}$ are integer matrices: the first can be arbitrary; the coefficients of the second must be of the form \begin{equation}\label{eq:coefficients of Matrix Rep for nonzero characteristic groups}
  A(i,j)= \alpha_{i,j}\, \frac{d_i}{\gcd{(d_i, c_j)}}
  \end{equation}
  where  $\alpha_{i,j}$ can be arbitrary integers\footnote{Since $A_{F\Z}$, $A_{FF}$ multiply integer tuples and output integer tuples modulo $F=\DProd{N}{c}$, for some $N_i$s, the coefficients of their $i$th rows can be chosen w.l.o.g. to lie in the range $[0,N_i)$ (by taking remainders).}.
  \item $A_{\T\Z}$, $A_{\T\R}$ and $A_{\T F}$ are real matrices: the first two are arbitrary; the coefficients of the third are of the form  $A(i,j)= \alpha_{i,j}/c_j$ where  $\alpha_{i,j}$ can be arbitrary integers\footnote{Due to the periodicity of the torus, the coefficients of $A_{\T\Z}$, $A_{\T F}$ can be chosen to lie in the range  $[0,1)$.}.
  \end{enumerate}
\end{lemma}
The result is proven in appendix \ref{appendix:Supplement to section Homomorphisms}.

\section{Systems of linear equations over Abelian groups}\label{sect:Systems of linear equations over groups}

Let $\alpha:G\rightarrow H$ be a  continuous group homomorphism between  elementary Abelian groups $G$, $H$ and let $A$ be a rational matrix representation of $\alpha$. We consider systems of equations of the form
\begin{equation}\label{eq:Systems of linear equations over groups}
\alpha(x)\equiv Ax\equiv b\pmod{H},\quad \textnormal{where } x\in G,
\end{equation}
which we dub \emph{systems of linear equations over (elementary) Abelian groups}.  In this section we develop algorithms to find solutions of such systems. 

Systems of linear equations over  Abelian groups form a large class of problems, containing, as particular instances, standard systems of linear equations over real vectors spaces,
\begin{equation}\label{eq:Linear System Ax=b over R}
\mathbf{A}\mathbf{x}=\mathbf{b},\quad \mathbf{A}\in\R^{n\times m},\, \mathbf{x}\in \R^m , \mathbf{b}\in \R^n,
\end{equation}
as well as systems of linear equations over other types of vector spaces, such as $\Z_2^n$, e.g.\
\begin{equation}\label{eq:Linear System Ax=b over Z_2}
\mathbf{B}\mathbf{y}=\mathbf{c},\quad \mathbf{B}\in\Z_2^{n\times m},\, \mathbf{y}\in \Z_2^m , \mathbf{c}\in \Z_2^n.
\end{equation}
In (\ref{eq:Linear System Ax=b over R}) the matrix $\mathbf{A}$ defines a linear map from $\R^m$ to $\R^n$, i.e.\ a map that fulfills $\mathbf{A}( a \mathbf{x}+b\mathbf{y})=\mathbf{A}( a \mathbf{x})+\mathbf{A}( b\mathbf{y}),$ for every $a,b\in\R,\,\mathbf{x,y}\in\R^n$ and is, hence, compatible with the vector space operations; analogously,  $\mathbf{B}$ in (\ref{eq:Linear System Ax=b over Z_2}) is a linear map between  $\Z_2$ vector spaces.

We dub systems (\ref{eq:Systems of linear equations over groups}) ``linear'' to highlight this resemblance. Yet the reader must beware that, in general, the groups $G$ and $H$ in problem (\ref{eq:Systems of linear equations over groups}) are \emph{not} vector spaces (primitive factors of the form $\Z$ or $\Z_d$, with non-prime $d$, are rings yet \emph{not} fields; the circle $\T$ is not even a ring, as it lacks a well-defined multiplication operation\footnote{Recall that $\T$ is a quotient group of $\R$ and that the addition in $\T$ is  well-defined  group operation between equivalence classes. It is, however, not possible to define a multiplication $a b$ for $a,b\in \T$ operation between equivalence classes: different choices of class representatives yield different results.}), and that the map $A$ is a group homomorphism between groups, but \emph{not} a linear map between vector spaces. 

Indeed, there are interesting classes of problems that fit in the class (\ref{eq:Systems of linear equations over groups}) and that are not systems of linear equations over vectors spaces. An example are the systems of linear equations over finite Abelian groups studied in \cite{BermejoVega_12_GKTheorem}. Another example are  systems of mixed real-integer linear equations \cite{BowmanBurget74_systems-Mixed-Integer_Linear_equations, HurtWaid70_Integral_Generalized_Inverse}, that we introduce later in this section (equation \ref{eq:System of Mixed Integer linear equations}).

\paragraph*{Input of the problem}

We only consider systems of the form (\ref{eq:Systems of linear equations over groups}) where the matrix $A$ is \emph{rational}. In other words,  we always assume that the group homomorphism $\alpha$ has a rational matrix representation $A$; the latter is  given to us in the input of our problem. Exact integer arithmetic will be used to store the rational coefficients of $A$; floating point arithmetic will never  
be needed in our work. 

Of course, not all group homomorphisms have rational matrix representations (cf.\ lemma \ref{lemma:Normal form of a matrix representation}). However, for the applications we are interested in this paper (cf.\  \cite{BermejoLinVdN13_BlackBox_Normalizers}) it is enough to study this subclass.

\paragraph*{General solutions of system (\ref{eq:System of Mixed Integer linear equations})}

Since $A$ is a homomorphism, it follows that the set $G_{\textnormal{sol}}$ of all solutions of (\ref{eq:Systems of linear equations over groups}) is either empty or a coset of the kernel of $A$:
\begin{equation}\label{eq:Systems of linear equations over groups: Solution Space}
G_{\textnormal{sol}}=x_0+\ker{A}
\end{equation}
The main purpose of this section is to devise efficient algorithms to solve system (\ref{eq:Systems of linear equations over groups}) when $A$, $b$ are given as  
input, in the following sense: we say that we have \emph{solved} system (\ref{eq:Systems of linear equations over groups}) if we manage to find a \emph{general solution} of (\ref{eq:Systems of linear equations over groups}) as defined next.
\begin{definition}[\textbf{General solution of system (\ref{eq:Systems of linear equations over groups})}]\label{def:General Solution of a system} A {\emph{general solution}} of a system of equations $Ax\equiv b\pmod{H}$ as in $(\ref{eq:Systems of linear equations over groups})$ is a duple $(x_0, \mathcal{E})$ where $x_0$ is a particular solution of the system and $\mathcal{E}$ is a continuous group homomorphism (given as a matrix representation) from an auxiliary group $\mathcal{X}:=\R^{\alpha} \times \Z^{\beta}$ into   $G$, whose image $\textnormal{im}\,\mathcal{E}$ is the kernel of ${A}$.
\end{definition}
Although it is not straightforward to prove, general solutions of systems of the form (\ref{eq:Systems of linear equations over groups}) \emph{always} exist. This is shown in appendix \ref{appendix:Closed subgroups of LCA groups / existence of general-solutions of linear systems}.

\subsection{Algorithm for finding a general solution of (\ref{eq:Systems of linear equations over groups})}

Observe that if we know a general-solution $(x_0, \mathcal{E})$ of system (\ref{eq:Systems of linear equations over groups}) then we can conveniently write the set of all solutions simply as $G_{\textnormal{sol}}=x_0+\textnormal{im}\,\mathcal{E}$. This expression suggests us a simple heuristic to sample random elements in  $G_{\textnormal{sol}}$--- which will be an important step in our proof of our main classical simulation result---based on the following approach:
\begin{enumerate}
\item[(1)] Choose a random element $v\in \mathcal{X}$ using some efficient classical procedure. This step should be feasible since this group has a simple structure: it is just the product of a conventional real Euclidean space $\R^a$ and an integer lattice $\Z^b$.
\item[(2)] Apply the map $v\rightarrow x_0 + \mathcal{E}(v)$, yielding a probability distribution on $G_{\textnormal{sol}}$.
\end{enumerate}

A  main contribution of our work is a deterministic classical algorithm that finds a general solution of any system of the form (\ref{eq:Systems of linear equations over groups}) in polynomial time.  
This is the content of the next theorem, which is one of our main technical results.
\begin{theorem}[\textbf{General solution of system (\ref{eq:Systems of linear equations over groups})}]\label{thm:General Solution of systems of linear equations over elementary LCA groups}
Let $A$, $b$  define a system of linear equations (over elementary Abelian groups) of form (\ref{eq:Systems of linear equations over groups}), with the group $G$ as solution space and image group $H$. Let $m$ and $n$ denote the number of direct-product factors of $G$ and $H$ respectively and let $c_i$, $d_j$ denote the characteristics of $G_i$ and $d_j$.   
Then there exist efficient, deterministic, exact classical algorithms to solve the following tasks in worst-case time complexity $O(\ppoly{m,n,\log\|A\|_{\mathbf{b}}, \log\|b\|_{\mathbf{b}}, \log{c_i},\log{d_j}}$:
\begin{enumerate}
\item Decide whether system (\ref{eq:Systems of linear equations over groups}) admits a solution.
\item Find a general solution $(x_0,\mathcal{E})$ of (\ref{eq:Systems of linear equations over groups}).
\end{enumerate}
\end{theorem}
A rigorous proof of this theorem is given in appendix \ref{appendix:Systems of Linear Equations over Groups}. The  main ideas behind it are discussed next. 

In short, we show that the problem of finding a general solution of a system of the form (\ref{eq:Systems of linear equations over groups}) reduces in polynomial time to the problem of finding a general solution of a so-called \emph{system of mixed real-integer linear equations} \cite{BowmanBurget74_systems-Mixed-Integer_Linear_equations}.
\begin{equation}\label{eq:System of Mixed Integer linear equations}
A'x'+B'y'=c,\quad  \textnormal{where } x'\in \Z^a, y' \in\R^b,
\end{equation}
 where $A'$ and $B'$  are rational matrices and $c$ is a rational vector. Denoting by $\R^{b'}$ the given space in which $c$ lives, we see that, in our notation, $\begin{pmatrix}
A & B
\end{pmatrix}w=c$, where $w\in \Z^a\times \R^b$ is a particular instance of a system of linear equations over elementary locally compact Abelian groups that are products of $\Z$ and $\R$. Systems (\ref{eq:System of Mixed Integer linear equations}) play an important role within the class of problems (\ref{eq:Systems of linear equations over groups}), since any efficient algorithm to solve the former  can be adapted to solve the latter in polynomial time.

The second main idea in the proof of theorem \ref{thm:General Solution of systems of linear equations over elementary LCA groups} is to apply an existing (deterministic) algorithm by Bowman and Burdet \cite{BowmanBurget74_systems-Mixed-Integer_Linear_equations} that computes a  general solution to a system of the form (\ref{eq:System of Mixed Integer linear equations}). Although Bowman and Burdet did not prove the efficiency of their algorithm in \cite{BowmanBurget74_systems-Mixed-Integer_Linear_equations}, we  show in appendix \ref{appendix:Efficiency of Bowman Burdet} that it can be implemented in polynomial-time, completing the proof of the theorem.

\subsection{Computing inverses of group automorphisms}\label{sect:Computin Inverses}
 
In section \ref{sect:Matrix Representations} we discussed that computing a matrix representation of the inverse $\alpha^{-1}$ of a group automorphism $\alpha$ cannot be done by simply inverting a (given) matrix representation $A$ of $\alpha$. However, the  algorithm given in theorem \ref{thm:General Solution of systems of linear equations over elementary LCA groups} can be adapted to  applied to solve this problem.
\begin{lemma}\label{lemma:Computing Inverses}
Let $\alpha:G\rightarrow G$ be a continuous group automorphism. Given any  matrix representation $A$ of $\alpha$, there exists efficient classical algorithms that compute a matrix representation $X$ of the inverse group automorphism $\alpha^{-1}$.
\end{lemma}
A proof (and an algorithm) is given in appendix \ref{appendix:Computing inverses}.

\section{Quadratic functions}\label{sect:quadratic_functions}

In this section we study the properties of quadratic functions over arbitrary elementary groups of the form $G=\R^a\times \T^a\times \Z^b \times \DProd{N}{c}$. Most importantly, we give normal forms for quadratic functions and bicharacters. We list results without proof, since all techniques used throughout the section are classical. Yet, we highlight that the normal form in theorem \ref{thm:Main Result} should be of quantum interest, since it can be used to give a normal form for stabilizer states over elementary groups.

All results in this section are proven in \textbf{appendix \ref{app:Supplement Quadratic Functions}}. 

\subsection{Definitions}

Let $G$ be an elementary Abelian group. Recall from section \ref{sect_normalizer_circuits} that a bicharacter of $G$ is a continuous complex function $B:G\times G\to U(1)$ such that the restriction of $B$ to either one of its arguments is a character of $G$.  Recall that a quadratic function $\xi:G\to U(1)$ is a continuous function for which there exists a bicharacter  $B$ such that  
\begin{equation} 
\xi(g+h) = \xi(g)\xi(h)B(g, h) \quad \mbox{ for all } g,h \in G.
\end{equation}
We say that $\xi$ is a $B$-representation.

A bicharacter $B$ is said to be symmetric if $B(g, h)= B(h, g)$ for all $g, h\in G$. Symmetric bicharacters are natural objects to consider in the context of quadratic functions: if $\xi$ is a $B$-representation then $B$ is symmetric since 
\begin{equation} B(g, h)= \xi(g+h)\overline{\xi(g)}\overline{\xi(h)}= \xi(h+g)\overline{\xi(h)}\overline{\xi(g)}=B(h,g).
\end{equation}

\subsection{Normal form of  bicharacters}

The next lemmas characterize bicharacter functions.
\begin{lemma}[\textbf{Normal form of a bicharacter}]\label{lemma:Normal form of a bicharacter 1}
Given an elementary Abelian group $G$, then a function $B:G\times G\rightarrow U(1)$ is a bi-character iff it  can be written in the normal form
\begin{equation}\label{eq:first normal fomr of a bicharacter}
B(g,h)=\chi_{\beta(g)}(h)
\end{equation}
where $\beta$ is some group homomorphism from $G$ into $G^*$.
\end{lemma}
 
This result generalizes lemma 5(a) in \cite{VDNest_12_QFTs}.The next lemma gives a explicit characterization of symmetric bicharacter functions.
\begin{lemma}[\textbf{Normal form of a symmetric bicharacter}]\label{lemma symmetric matrix representation of the bicharacter homomorphism}
Let $B$ be a symmetric bicharacter of $G$ in the form (\ref{eq:first normal fomr of a bicharacter}) and let $A$ be a matrix representation of the homomorphism $\beta$. Let $\Upsilon$ denote the default matrix representation of the bullet map $G^*\rightarrow G^\bullet$ as in (\ref{eq:matrix representation of bullet isomorphism}), and $M=\Upsilon A$. Then
\begin{enumerate*}
\item[(a)] $B(g,h)=\exp{\left(2\pii \, g^\transpose M h\right)}$ for all $g,h\in G$.
\item[(b)] $M$ is a matrix representation of the homomorphism $G\stackrel{\beta}{\rightarrow} G^*\stackrel{\bullet}{\rightarrow}G^\bullet$.
\item[(c)] If $x, y\in\R^m$ and $g, h\in G$ are such that $x\equiv g$ (mod $G$) and $y\equiv h$ (mod $G$), then
\be
B(g, h) = \exp{\left(2\pii \, x^\transpose M y\right)}.
\ee

\item[(d)] The matrix $M$ is symmetric modulo integer factors, i.e.\ $M = M^{T} \bmod{\Z}$.
\item[(e)] The matrix $M$ can be efficiently symmetrized: i.e.\ one can compute in classical polynomial time a symmetric matrix $M'=M^{'T}$ that also fulfills (a)-(b)-(c).
\end{enumerate*}
\end{lemma}

\subsection{Normal form of quadratic functions}

Our final goal is to characterize all quadratic functions. This is achieved in theorem \ref{thm:Normal form of a quadratic function}. To show this result a few lemmas are needed.
\begin{lemma}\label{lemma quadratic B-representations differ by a character}
Two quadratic functions $\xi_1$,$\xi_2$ that are $B$-representations of the same bicharacter  $B$ must be equal up to multiplication by a character of $G$, i.e. there exists $\mu\in G^*$ such that
\begin{equation}
\xi_1(g)=\chi_{\mu}(g)\xi_2(g), \quad \text{for every $g\in G$.}
\end{equation}
\end{lemma}
\begin{proof}
We prove that the function $f(g):=\xi_1(g)/\xi_2(g)$ is a character, implying that there exists $\mu\in G^*$ such that $\chi_{\mu}=f$:
\begin{equation}\label{eq proof of quadratic B-representations differ by character}
f(g+h):=\frac{\xi_1(g)}{\xi_2(g)}\frac{\xi_1(h)}{\xi_2(h)}\frac{B(g,h)}{B(g,h)}=f(g)f(h).
\end{equation}
\end{proof}
We highlight that a much more general version of lemma \ref{lemma quadratic B-representations differ by a character} was proven in  \cite{BackBrad71_projective_representations_of_Abelian_groups}, using projective representation theory\footnote{Precisely, the authors show (see theorem 1 in \cite{BackBrad71_projective_representations_of_Abelian_groups}) that if  $D_1$, $D_2$ are finite-dimensional irreducible unitary projective representations of a locally compact Abelian group $G$, possessing the same factor system $\omega$, then there exists a unitary transformation $U$ and a character $\chi_h$ such that $U^{-1}D_1(g)U=\chi_h(g)D_2(g)$. In our set-up quadratic functions are one-dimensional projective irreps of $G$ and bicharacters are particular examples of factor systems.}.

Our approach now will be to find a method to construct a quadratic function that is a $B$-representation for any given bicharacter $B$. Given one $B$-representation,  lemma \ref{lemma quadratic B-representations differ by a character} tells us how all other $B$-representation look like. We can exploit this to characterize all possible quadratic functions, since we know how symmetric bicharacters look (lemma\ref{lemma symmetric matrix representation of the bicharacter homomorphism}).

The next lemma shows how to construct $B$-representations canonically.
\begin{lemma}\label{lemma:B-representations can always be constructed}
Let be a bicharacter $B$ of $G$. Consider  a symmetric real matrix $M$ such that $B(g,h)=\exp{\left(2\pii \, g^{\transpose} M h\right)}$. Then the following function is quadratic and  a $B$-representation:
\begin{equation}
Q(g):=\euler^{\pii \, \left(  g^{\transpose} M g + C^\transpose g \right)},
\end{equation}
where $C$ is an integer vector dependent on $M$, defined component-wise as $C(i)=M(i,i)c_i$, where $c_i$ denotes the characteristic of the group $G_i$.
\end{lemma}
Finally, we arrive at the main result of this section.
\begin{theorem}[\textbf{Normal form of a quadratic function}]\label{thm:Normal form of a quadratic function}
Let $G$ be an elementary Abelian group. Then a function $\xi: G\rightarrow U(1)$ is quadratic if and only if
\begin{equation}\label{eq:Normal Form Quadratic Function}
\xi(g)=\euler^{\pii \,\left(g^{\transpose} M g \: +  \: C^{\transpose} g  \: +  \: 2v^\transpose g\right)}
\end{equation}
where $C$, $v$, $M$ are, respectively, two vectors and a matrix that satisfy the following:
\begin{itemize}
\item $v$ is an element of the bullet group $G^\bullet$;
\item $M$ is the matrix representation of a group homomorphism from $G$ to $G^\bullet$; and
\item $C$ is an integer vector dependent on $M$, defined component-wise as $C(i)=M(i,i)c_i$, where $c_i$ is the characteristic of the group $G_i$.
\end{itemize}
\end{theorem}
As discussed in the introduction, theorem \ref{thm:Normal form of a quadratic function} may be used to extend part of the so-called  ``discrete Hudson theorem'' of Gross   \cite{Gross06_discrete_Hudson_theorem}, which states that the phases of odd-qudit stabilizer  states are quadratic.

The normal form in theorem \ref{thm:Normal form of a quadratic function} can be very useful to perform certain calculations within the space of quadratic functions, as illustrated by the following lemma.
\begin{lemma}\label{lemma:Quadratic Function composed with Automorphism}
Let $\xi_{M,v}$ be the quadratic function (\ref{eq:Normal Form Quadratic Function}) over $G$. Let $A$ be the matrix representation of a continuous group homomorphism $\alpha:G\rightarrow G$. Then the composed function $\xi_{M,v}\circ \alpha$ is also quadratic and can be written in the normal form
(\ref{eq:Normal Form Quadratic Function})  as $\xi_{M',v'}$, with
\begin{equation}
M':=A^\transpose M A, \qquad v':= A^\transpose v + v_{A,M},\qquad v_{A,M}:= A^\transpose C_M - C_{A^\transpose M A},
\end{equation}
 
 where $C_M$ is the vector $C$ associated with $M$ in (\ref{eq:Normal Form Quadratic Function}).
\end{lemma}
This lemma will be used to prove theorem \ref{thm:Main Result} (in the proof of lemma \ref{lemma:Evolution of Pauli Phases}.)

\section{Pauli operators over Abelian groups}\label{sect:Pauli operators over Abelian groups}

In this section we introduce Pauli operators over groups of the form $G= \Z^a\times \T^b\times F$ (note that we no longer include factors of $\R^d$ because these groups are not related to the Hilbert spaces that we study in this paper), discuss some of their basic properties and finally show that normalizer gates map any Pauli operator to another Pauli operator. The latter property is a generalization of a well known property for qubit systems, namely that Clifford operations map the Pauli group to itself.

\textbf{Note on terminology.} Throughout the rest of the paper, sometimes we use the symbol $\mathcal{H}_\T$ as a second name for the Hilbert space $\mathcal{H}_\Z$. Whenever this notation is used, we make implicit that we are working on the Fourier basis of  $\mathcal{H}_\Z$, which is labeled by the circle group $\T$. Sometimes, this basis will be called the $\T$ standard basis or just $\T$ basis. 

\subsection{Definition and basic properties}\label{sect:pauli_def}

Consider an Abelian group of the form $G= \Z^a\times \T^b\times F$ and the associated Hilbert space ${\cal H}_G$ with the associated group-element basis $\{|g\rangle: g\in G\}$ as defined in section \ref{sect:Quantum states over infinite Abelian groups} . We define two types of unitary gates acting on $\mathcal{H}_G$, which we call the \emph{Pauli operators of} $G$. The first type of Pauli operators are the X-type operators $X_G(g)$ (often called \emph{shift operators} in generalized harmonic analysis):
\be\label{eq:Pauli operator type X, over G, shift definition}
X_G(g)\psi(h) := \psi(h-g), \quad \text{for every $g, h\in G$},
\ee
where the $\psi(h)$ are the coefficients of some quantum state $|\psi\rangle$ in ${\cal H}_G$. These operators can also be written via their action on the standard basis, which yields a more familiar definition:
\begin{equation}\label{eq:Pauli operator type X, over G, eigenket definition}
X_{{G}}(g)\ket{h} = \ket{g+h}, \quad  \text{for every $g,h\in G$}.
\end{equation}
In representation theory, the map $g\rightarrow X_G(g)$ is called the \emph{regular representation} of the group $G$. The second type of Pauli operators are the Z-type operators $Z_G(\mu)$:
\begin{equation}\label{eq:Pauli operator type Z}
Z_{G}(\mu)\ket{g} :=\chi_{\mu}(g)\ket{g}, \quad \text{for every $g\in G,\, \mu\in G^*$}.
\end{equation}
We define a \emph{ generalized Pauli operator} of $G$ to be any unitary operator of the form
\begin{equation}\label{eq definition of Pauli operator}
 \sigma:=\gamma Z_G(\mu) X_G(g)
\end{equation}
where $\gamma$ is a complex number with unit modulus. We will call the duple $(\mu, g)$ and the complex number $\gamma$, respectively, the \emph{label} and the \emph{phase} of the Pauli operator $\sigma$. Furthermore we will regard the label  $(\mu, g)$ as an element of the Abelian group $G^*\times G$. The above definition of Pauli operators is a generalization of the notion of Pauli operators over finite Abelian groups as considered in \cite{VDNest_12_QFTs,BermejoVega_12_GKTheorem}, which was in turn a generalization of the standard notion of Pauli operators for qubit systems. An important distinction between Pauli operators for finite Abelian groups and the current setting is that the $Z_G(\mu)$ are labeled by $\mu\in G^*$. For finite Abelian groups, we have  $G^* = G$ and consequently the Z-type operators are also labeled by elements of $G$.

Using the definition of Pauli operators, it is straightforward to verify the following commutation
relations, which hold for all $g\in G$ and $\mu\in G^*$:
\begin{align}
X_{G}(g)X_{G}(h) &=  X_G(g+h)=X_G(h)X_G(g)\nonumber \\
Z_{G}(\mu)Z_{G}(\nu) &=  Z_{G}(\mu+\nu)=Z_{G}(\nu)Z_{G}(\mu)\label{eq:pauli_commutation}\\
Z_{G}(\mu)X_{G}(g) &=  \chi_{\mu}(g)X_{G}(g)Z_{G}(\mu)\nonumber
\end{align}
It follows that the set of generalized Pauli operators of $G$
form a group, which we shall call the \emph{Pauli group} of $G$.

\subsection{Evolution of Pauli operators} \label{sect:Evolution Pauli Operators}

The connection between normalizer gates and the Pauli group is that the former ``preserve'' the latter  under conjugation, as we will show in this section. This property will be a generalization of the well known fact that the Pauli group for $n$ qubits is mapped to itself under the conjugation map $\sigma\to U\sigma U^{\dagger}$, where $U$ is either a Hadamard gate, CNOT gate or $(\pi/2)$-phase gate\cite{Gottesman_PhD_Thesis,Gottesman99_HeisenbergRepresentation_of_Q_Computers}---these gates being normalizer gates for the group $G=\Z_2\times\dots\Z_2$ \cite{VDNest_12_QFTs,BermejoVega_12_GKTheorem}. More generally, it was shown in \cite{VDNest_12_QFTs} that normalizer gates over any finite Abelian group $G$ map the corresponding Pauli group over $G$ to itself under conjugation. In generalizing the latter result to Abelian groups of the form $G= \Z^a\times \T^b\times F$, we will however note an important distinction. Namely, normalizer gates over $G$ will map Pauli operators over $G$ to Pauli operators over a group $G'$ which is, in general, \emph{different} from the initial group $G$. This feature is a consequence of the fact that the groups $\Z^a$ and $\T^b$ are no longer autodual (whereas all finite Abelian groups are). Consequently, as we have seen in section \ref{sect Normalizer gates}, the QFT over $G$ (or any partial QFT) will change the  group that labels the designated basis of ${\cal H}$ from $G$ to $G'$. We will therefore find that the QFT maps Pauli operators over $G$ to Pauli operators over $G'$. In contrast, such a situation does not occur for automorphism gates and quadratic phase gates, which do not change the group $G$ that labels the designated basis.

Before describing the action of normalizer gates on Pauli operators (theorem \ref{thm:Normalizer gates are Clifford}), we provide two properties of QFTs.
\begin{lemma}[\textbf{Fourier transforms diagonalize shift operators}]\label{lemma:Fourier transform diagonalizes the Regular Representation}
Consider a group of the form $G= \Z^a\times \T^b\times F$. Then the X-type Pauli operators of $G$ and the Z-type operator of $G^*$ are related via the quantum Fourier transform $\mathcal{F}_G$ over $G$:
\be
Z_{G^*}(g)=\mathcal{F}_G X_G(g) \mathcal{F}_G^\dagger.
\ee
\end{lemma}
\begin{proof}
 
We show this by direct evaluation of the operator $X_G(h)$  on the Fourier basis states (section \ref{sect:Quantum states over infinite Abelian groups}  definition). Using the definitions introduced in section \ref{sect:characters} we can write the vectors in the Fourier basis of $G$ in terms of character functions (definition \ref{def:Characters}): letting $\ket{\mu}$ be the state
\begin{equation}\label{eq:Fourier Basis in terms of Characters}
\ket{\mu}= \int_{G} \mathrm{d} h \, \overline{\chi_\mu(h)}\ket{h} = \int_{G} \mathrm{d} h \, \chi_{-\mu}(h)\ket{h},
\end{equation}
then the Fourier basis of $G$ is just the set $\{\ket{\mu}, \,\mu \in G^*\}$. Now it is easy to derive
\begin{align}
X_G(g)\ket{\mu} &= X_G(g) \left(\int_{G} \mathrm{d} h \overline{\chi_\mu(h)}\ket{h}\right) =  \int_{G} \mathrm{d} h \overline{\chi_\mu(h)}\ket{g+h} =  \int_{G} \mathrm{d} h' \overline{\chi_\mu(h'-g)}\ket{h'}\notag\\
&= \overline{\chi_\mu(-g)}\left(\int_{G} \mathrm{d} h' \chi_\mu(h')\ket{h'} \right)= \chi_{g}(\mu)\ket{\mu} = Z_{G^*}(g) \ket{\mu}.
\end{align}
In the derivation we use  lemmas \ref{lemma:Pontryagin duality for characters}, \ref{lemma:Character Multiplication} and equation (\ref{eq:Pauli operator type Z}) applied to the group $G^*$.
\end{proof}
The next theorem shows that normalizer gates are generalized Clifford operations, i.e.\ they transform Pauli operators into Pauli operators under conjugation and, therefore, they \emph{normalize}  the group of all Pauli operators within the group of all unitary gates\footnote{It is usual in quantum information theory to call the normalizer group of the $n$-qubit Pauli group ``the Clifford group'' because of a ``tenuous relationship'' \cite[Gottesman]{Gottesman09Intro_QEC_and_FTQC}  to Clifford algebras.}. 
\begin{theorem}[\textbf{Normalizer gates are Clifford}]\label{thm:Normalizer gates are Clifford}
Consider a group of the form $G= \Z^a\times \T^b\times F$. Let $U$ be a normalizer gate of the group $G$. Then $U$ corresponds to an isometry from ${\cal H}_G$ to ${\cal H}_{G'}$ for some suitable group $G'$, as discussed in section  
\ref{sect_normalizer_circuits}.  Then the conjugation
map $\sigma\rightarrow U\sigma U^{\dagger}$ sends Pauli operators
of $G$ to Pauli operators of $G'$. We say that $U$ is a Clifford operator.
\end{theorem}
\begin{proof}
We provide an explicit proof for Pauli operators of type $X_{G}(g)$  and  $Z_{G}(\mu)$. This is enough to prove the lemma due to  (\ref{eq:pauli_commutation}). As before,  $G=G_{1}\times\cdots\times G_{m}$ where the $G_{i}$ are groups of primitive type.

We break the proof into three cases.
\begin{itemize}
\item If $U$ is an automorphism gate $U_{\alpha}:\:|h\rangle\rightarrow|\alpha(h)\rangle$ then
\begin{equation}\label{eq:Automorphism Gate on X}
U_{\alpha}X_{G}(g)U_{\alpha}^{\dagger}|h\rangle=|\alpha(\alpha^{-1}(h)+g)\rangle=|h+\alpha(g)\rangle=X_{G}(\alpha(g))|h\rangle,
\end{equation}
\begin{equation}\label{eq:Automorphism Gate on Z}
U_{\alpha}Z_{G}(\mu)U_{\alpha}^{\dagger}|h\rangle=\chi_{\mu}(\alpha^{-1}(h))|h\rangle=\chi_{\alpha^{*^{\inverse}}(\mu)}(h)|h\rangle=Z_{G}\left(\alpha^{*^{\inverse}}(\mu)\right)|h\rangle.
\end{equation}
\item If $U$ is a quadratic phase gate $D_{\xi}$ associated  
with a quadratic function $\xi$  then
\begin{align}\label{eq:QuadraticPhase Gate on X}
D_{\xi}X_{G}(g)D_{\xi}^{\dagger}|h\rangle &= \xi(g+h)\overline{\xi(h)}|g+h\rangle=\xi(g)B(g,h)|g+h\rangle \notag\\
 & =  \xi(g)\chi_{\beta(g)}(h)|g+h\rangle=\xi(g)X(g)Z(\beta(g))|h\rangle,
\end{align}
where, in the second line, we use lemma \ref{lemma:Normal form of a bicharacter 1}. Moreover $D_{\xi}Z_{G}(\mu)D_{\xi}^{\dagger}= Z_{G}(\mu)$ since diagonal gates commute.
\item If $U$ is the  Fourier transform $F_G$ on the  $\mathcal{H}_{G}$ then
\begin{equation}\label{eq:Fourier transform on X and Z}
F_G X_{G}(g)F_G^{\dagger}=Z_{G^*}(g),\qquad\quad\quad F_G Z_{G}(\mu)F_G^{\dagger}=X_{G^*}(-\mu).
\end{equation}
The first identity is the content of lemma \ref{lemma:Fourier transform diagonalizes the Regular Representation}. The second is proved in a similar way:
\begin{align}
Z_G(\mu)\ket{\nu} &= Z_G(\mu) \left(\int_{G} \mathrm{d} h \chi_{-\nu}(h)\ket{h}\right) =  \int_{G} \mathrm{d} h \chi_{-\nu}(h)\chi_\mu(h)\ket{h} =  \int_{G} \mathrm{d} h \chi_{-(\nu-\mu)}(h)\ket{h}\notag\\
&= \ket{\nu-\mu} = X_{G^*}(-\mu) \ket{\nu},
\end{align}
where we apply (\ref{eq:Fourier Basis in terms of Characters}), lemma \ref{lemma:Character Multiplication} and (\ref{eq:Pauli operator type X, over G, eigenket definition},\ref{eq:Pauli operator type Z}). These formula also apply to partial Fourier transforms $\mathcal{F}_{G_i}$, since Pauli operators decompose as tensor products.\qedhere
\end{itemize}
\end{proof}

\section{Stabilizer states}\label{sect:Stabilizer States}

In this section we develop a stabilizer framework to simulate normalizer circuits over infinite Abelian groups of the form $G=\Z^a\times\T^b\times \DProd{N}{c}$. As explained in section \ref{sect:Previous Work}, our techniques generalize methods given in \cite{VDNest_12_QFTs, BermejoVega_12_GKTheorem} (which apply to groups of the form $F=\DProd{N}{c}$) and is closely related to the (more general) monomial stabilizer formalism \cite{nest_MMS}.

 \subsection{Definition and basic properties}

A \emph{stabilizer group $\mathcal{S}$}  over  $G$ is any group of \emph{commuting} Pauli operators of $G$ with a nontrivial +1 common eigenspace.  
Here we are interested in stabilizer groups where the +1 common eigenspace is one-dimensional, i.e. there exists a state $|\psi\rangle$ such that $\sigma|\psi\rangle = |\psi\rangle$ for all $\sigma\in {\cal S}$, and moreover $|\psi\rangle$ is the unique state (up to normalization) with this property. Such states are called stabilizer states (over $G$). This terminology is an extension of the already established stabilizer formalism for finite-dimensional systems \cite{Gottesman_PhD_Thesis,Gottesman99_HeisenbergRepresentation_of_Q_Computers,Knill96non-binaryunitary,Gottesman98Fault_Tolerant_QC_HigherDimensions,VDNest_12_QFTs,BermejoVega_12_GKTheorem}.

We stress here  that stabilizer states $|\psi\rangle$ are allowed to be unnormalizable states;  
in other words, we do not require $|\psi\rangle$ to belong to the physical Hilbert space ${\cal H}_G$. In a more precise language, stabilizer states may be tempered distributions in the Schwartz-Bruhat space $\mathcal{S}_G^\times$ \cite{Bruhat61_Schwatz-Bruhat-functions,Osborne75_Schwartz_Bruhat}. This issue arises only when considering infinite groups, i.e. groups containing  $\Z$ or $\T$. An example of a non-physical stabilizer state is the  Fourier basis state $|p\rangle$ (\ref{eq:Fourier basis state over Z})  (we argue below that this is indeed a stabilizer state). Note that not all stabilizer states for $G=\T$ must be unphysical; an example of a physical stabilizer state within ${\cal H}_G$ is 
\begin{equation}
\int_{\T} dp |p\rangle.
\end{equation} 
The stabilizer group of this state is $\{X_\T(p): p\in \T\}$, which can be alternatively written as $\{Z_\Z(p): p\in \T\}$ (lemma \ref{lemma:Fourier transform diagonalizes the Regular Representation}). Similar examples of stabilizer states within and outside ${\cal H}_G$ can be given for $G=\Z$. Note, however, that in this case the standard basis states $|x\rangle$ with $x\in \Z$ (which are again stabilizer states) \emph{do} belong to ${\cal H}_{\Z}$.

Next we show that all standard basis states are stabilizer states. 
\begin{lemma}\label{lemma:Stabilizer Group for Standard Basis State}
Consider $G= \Z^a\times \T^b\times F$ with associated Hilbert space ${\cal H}_G$ and standard basis states $\{|g\rangle:g\in G\}$. Then every standard basis state $|g\rangle$ is a stabilizer state. Its stabilizer group is \be\label{stabilizer_g}
\{\overline{\chi_{\mu}}(g)Z_G(\mu): \mu\in G^*\}.\ee
\end{lemma}
The lemma implies that the Fourier basis states and, in general,
any of the allowed group-element basis states (11) are stabilizer states.
 
\begin{proof} Let us first prove the theorem for $g=0$, and show that $\ket{0}$ is the unique state that is stabilized by $\mathcal{S}=\{Z_G(\mu): \mu\in G^*\}$. It is easy to check that a standard-basis state $\ket{h}$ with $h\in G$ is a common +1-eigenstate of $\mathcal{S}$ if and only if $\chi_\mu(h)=1$ for all $\mu\in G^*$ or, equivalently, iff $h$ belongs to $G^\perp$, the annihilator of $G$. It is known that $G^\perp$ coincides with the trivial subgroup $\{0\}$ of $G^*$ \cite[corollary 20.22]{Stroppel06_Locally_Compact_Groups}, and therefore $\ket{0}$ is the unique standard-basis state that is also a common +1 eigenstate of $\mathcal{S}$. Since all unitary operators of $\mathcal{S}$ are diagonal in the standard basis, $\ket{0}$ is the unique common +1 eigenstate of $\mathcal{S}$.

For arbitrary $\ket{g}=X_G(g)\ket{0}$, the stabilizer group of $\ket{g}$ is $X_G(g)\mathcal{S}X_G(g)^\dagger$, which equals $\{\overline{\chi_{\mu}}(g)Z_G(\mu): \mu\in G^*\}$ (see equation  (\ref{eq:pauli_commutation})).

\end{proof}
Let $|\psi\rangle$ be a stabilizer state with stabilizer group ${\cal S}$. We define the following sets, all of which are easily  
verified to be Abelian groups:
\begin{align}\label{label_groups}
\mathbb{L}&:=\{(\mu,g)\in G^*\times G\ : \ \mathcal{S} \textnormal{ contains a Pauli operator of the form }\gamma Z(\mu)X(g)\};\nonumber\\
\mathbb{H}&:=\{g\in G\ : \ \mathcal{S} \textnormal{ contains a Pauli operator of the form }\gamma Z(\mu)X(g)\};\nonumber\\
\mathbb{D}&:=\{\mu\in G^*\ : \ \mathcal{S} \textnormal{ contains a Pauli operator of the form }\gamma Z(\mu)\}
\end{align}
The groups $\mathbb{L}$, $\mathbb{D}$ and $\mathbb{H}$ contain information about the labels of the operators in ${\cal S}$. We highlight that, although $\mathbb{D}$ and $\mathbb{H}$ are subsets of very different groups (namely $G$ and $G^*$, respectively),  they are actually closely related to each other by the relation
\begin{equation}
\mathbb{H}\subseteq \mathbb{D}^\perp \quad (\mbox{or, equivalently, } \mathbb{D}\subseteq \mathbb{H}^\perp),
\end{equation}
which follows from the commutativity of the elements in $\mathcal{S}$ and the definition of orthogonal complement (recall section \ref{sect:pauli_def}). 

Finally, let $\mathcal{D}$ be the subgroup of all diagonal Pauli operators of $\mathcal{S}$. It is easy to see that, by definition, $\mathcal{D}$ and $\mathbb{D}$ are isomorphic to each other.

\subsection{Support of a stabilizer state}

We show that the support of a stabilizer state $\ket{\psi}$ (the manifold of points where the wavefunction $\psi(x)$ is not zero) can be fully characterized in terms of the  label groups $\mathbb{H}$, $\mathbb{D}$.

Our next result characterizes the the structure of this wavefunction.
\begin{lemma}
Every stabilizer state \label{lemma:StabStates_are_Uniform_Superpositions}$\ket{\psi}$ over $G$ is a uniform quantum superposition over some subset of the form $s+\mathbb{H}$, where $\mathbb{H}$ is a subgroup of $G$. Equivalently, any stabilizer state  $\ket{\psi}$ can be writen in the the form
\begin{equation}\label{eq:Stabilizer_States_r_Uniform_Superpositions}
\ket{\psi}=\int_{\mathbb{H}} \mathrm{d}h \, \psi(h) \ket{s+h}
\end{equation}
where all amplitudes have equal absolute value $|\psi(h)|=1$. We call the $s+\mathbb{H}$ the \emph{support} of the state.
\end{lemma}
 
This lemma generalizes corollary 1 in \cite{VDNest_12_QFTs}.
\begin{proof}
 Let $\ket{\psi}$ be an arbitrary quantum state $\ket{\psi}=\int_X\mathrm{d}g\, \psi(g)\ket{g}$. The action of an arbitrary Pauli operator $U=\gamma Z_G(\mu)X_G(h)\in \mathcal{S}$ on the state is
\begin{equation}\label{eq:inproof_Support_Stabilizer_State 0}
U\ket{\psi}= \gamma \int_X\mathrm{d}g\, \chi_{\mu}(g+h)\psi(g)\ket{g+h} = \int_X\mathrm{d}g\, \psi(g)\ket{g} =\ket{\psi}.
\end{equation}
Recall the definition of $\mathbb{H}$ in (\ref{label_groups}). Comparing the two  integrals in (\ref{eq:inproof_Support_Stabilizer_State 0}), and knowing that $|\chi_\mu(x)|=1$ for every $x\in G$, we find that the absolute value of $\psi$ cannot change if we shift this function  by an element of $\mathbb{H}$; in other words,
\begin{equation}\label{eq:inproof_Support_Stabilizer_States 1}
\textit{for every $g\in X$ it holds $|\psi(g)|=|\psi(g+h)|$ for every $h\in\mathbb{H}$.}
\end{equation}
Now let $Y\subset X$  denote the subset of points $y\in X$ for which $\psi(y)\neq0$. Eq.\ (\ref{eq:inproof_Support_Stabilizer_States 1}) implies that $Y$ is a disjoint union of cosets of $\mathbb{H}$, i.e.\
\begin{equation}
Y=\bigcup_{\iota\in I} s_\iota + \mathbb{H},
\end{equation}
where $I$ is a (potentially uncountable) index set, and that $\ket{\psi}$ is of the form
\begin{equation}\label{eq:inproof_Support_Stabilizer_States 2}
\ket{\psi} = \int_Y\mathrm{d}y\, \psi(y)\ket{y}  = \int_I \mathrm{d} \iota\, \alpha(\iota) \ket{\phi_\iota},
\end{equation}
where the states $\ket{\phi_\iota}$ are \emph{non-zero} linearly-independent uniform superpositions over  the cosets $x_\iota+\mathbb{H}$:
\begin{equation}\label{eq:inproof_Support_Stabilizer_States_3}
\ket{\phi_\iota}=\int_{\mathbb{H}}\mathrm{d}h\, \phi_\iota(h)\ket{s_\iota + h}
\end{equation}
and $|\phi_\iota(h)|=1$ for every $h$. Putting together (\ref{eq:inproof_Support_Stabilizer_States 2}) and (\ref{eq:inproof_Support_Stabilizer_States_3}) we conclude that, for any $U\in \mathcal{S}$, the condition $U\ket{\psi}=\ket{\psi}$ is fulfilled  if and only if $U\ket{\phi_\iota}=\ket{\phi_\iota}$ for every $\ket{\phi_\iota}$:  this holds because $U$ leaves invariant the mutually-orthogonal vector spaces  $\mathcal{V}_\iota:=\mathrm{span}\{\ket{s_\iota + h}:\: h\in \mathbb{H}\}$. Consequently,  every state $\ket{\phi_\iota}$ is a (non-zero) common +1 eigenstate of all operators in $\mathcal{S}$.  
Finally, since we know that  $\ket{\psi}$ is the \emph{unique} +1 common eigenstate of $\mathcal{S}$,  it follows from  (\ref{eq:inproof_Support_Stabilizer_States 2}, \ref{eq:inproof_Support_Stabilizer_States_3}) that  $I$ has exactly one element and $Y=s+\mathbb{H}$; as a result,  $\ket{\psi}$ is a uniform superposition of the form (\ref{eq:inproof_Support_Stabilizer_States_3}). This proves the lemma.
\end{proof}
\begin{lemma}\label{lemma:Support of StabStates}
An element $x\in G$  belongs to the support $s+\mathbb{H}$ of a stabilizer state $\ket{\psi}$  if and only if
\begin{equation}\label{eq:Support of Stabilizer State 1}
D\ket{x}=\ket{x} \quad \textrm{for all $D\in \mathcal{D}$}.
\end{equation}
Equivalently, using that every $D$ is of the form $D=\gamma_\mu Z_{G^*}(\mu)$ for some $\mu\in\mathbb{D}$ and that $\gamma_\mu$ is determined given $\mu$,
\begin{equation}\label{eq:Support of Stabilizer State 2}
\mathrm{supp}(\ket{\psi}) =\{ x\in G\: : \: \chi_\mu(x)=\overline{\gamma_\mu}\textrm{ for all $\mu\in \mathbb{D}$}\}.
\end{equation}
\end{lemma}
Lemma \ref{lemma:Support of StabStates} was proven for finite groups in \cite{VDNest_12_QFTs}, partially exploiting the monomial stabilizer formalism (MSF) developed in \cite{nest_MMS}. Since the MSF framework has not been generalized to finite dimensional Hilbert spaces, the techniques in \cite{nest_MMS,VDNest_12_QFTs} can no longer be applied in our setting\footnote{The authors believe that the MSF formalism in \cite{nest_MMS} should be easy to extend to infinite dimensional systems if one looks at monomial stabilizer groups with normalizable eigenstates. However, dealing with monomial operators with unnormalizable eigenstates---which can be the case for (\ref{eq:Pauli operator type X, over G, shift definition},\ref{eq:Pauli operator type Z})---seems to be notoriously harder.}. Our proof works in infinite dimensions and even in the case when the Pauli operators (\ref{eq:Pauli operator type X, over G, shift definition},\ref{eq:Pauli operator type Z}) have unnormalizable eigenstates.
\begin{proof}
Write  $\ket{\psi}$ as in (\ref{eq:Stabilizer_States_r_Uniform_Superpositions}) integrating over  $X:=s+\mathbb{H}$. Then, the ``if'' condition follows easily by evaluating the action of an arbitrary diagonal stabilizer operator $D=\gamma_\mu Z_{G^*}(\mu)$  on a the stabilizer state $\ket{\psi}$: indeed, the condition
\begin{equation}
D\ket{\psi}=\ket{\psi} \quad \iff\quad  \int_{X}\mathrm{d}x\, \left(\gamma_\mu \chi_{\mu}(x)\right)\psi(x)\ket{x} = \int_X\mathrm{d}x\, \psi(x)\ket{x},
\end{equation}
holds only if $\gamma_\mu \chi_\mu (x) =1$, which is equivalent to $D\ket{x}=\ket{x}$ (here, we use implicitly that $\psi(x)\neq0$ for all integration points).

Now we prove the reverse implication. Take $x\in G$ such that  $D\ket{x}=\ket{x}$ for all $D\in \mathcal{D}$. We want to show that  $\ket{x}$ belongs to the set $s+\mathbb{H}$. We argue by contradiction, showing that $x\notin s+\mathbb{H}$ implies that there exists a nonzero  common +1 eigenstate $\ket{\phi}$ of all $\mathcal{S}$ that is not proportional to $\ket{\psi}$, which cannot happen.

We now show how to  construct such a $\ket{\phi}$.

Let   $Y:=\{\xi(\mu, g)Z_{G^*}(\mu)X_G(g)\}$ be  a system of representatives of the factor group $\mathcal{S}/\mathcal{D}$. For every $h\in \mathbb{H}$, we use the notation $V_h$ to denote a Pauli operator of the form $\xi(\nu_h, h)Z_{G^*}(\nu_h)X_G(h)$. It is easy to see that the set of all such $V_h$ forms an equivalence class in $\mathcal{S/D}$, so that there is a one-to-one correspondence between $\mathbb{H}$ and $\mathcal{S/D}$. Therefore, if to every $h\in\mathbb{H}$ we associate a $U_h\in Y$ (in a unique way), written as $U_h:=\xi(\nu_h,h) Z_{G*}(\nu_h)X(h)$, then we have that:
\begin{itemize}
\item[(a)] any Pauli operator  $V\in\mathcal{S}$ can be written as $V=U_x D$ for some $U_x\in Y$ and $D\in\mathcal{D}$;
\item[(b)]   $U_g U_h =U_{g+h} D_{g,h}$ for every $U_g$, $U_h\in Y$ and some $D_{g,h}\in \mathcal{D}$.
\end{itemize}
With this conventions, we take $\phi$ to be the state
\begin{equation}\label{eq:inproof_Support of Stab States 1}
\ket{\phi}:=\left( \int_{Y} \mathrm{d}U \, U\right ) \ket{x} = \left( \int_{\mathbb{H}} \mathrm{d}h \, U_h\right ) \ket{x} = \int_{\mbb{H}} \mathrm{d}h \,  \xi(\nu_h, h) \chi_{\nu_h}(x+h)  \ket{x+h} \textrm{d}{h}.
\end{equation}
The last equality in (\ref{eq:inproof_Support of Stab States 1}) shows that $\ket{\phi}$ is a uniform superposition over $x+\mathbb{H}$. As a result, $\ket{\phi}$ is non-zero. Moreover, $\ket{\phi}$ linearly independent from $\ket{\psi}$ if we assume $x\notin \mathrm{supp}(\psi)$, since this  implies that $\mathrm{supp}(\phi) = x+\mathbb{H}$ and $\mathrm{supp}(\psi)= s+\mathbb{H}$  are disjoint.   
Lastly, we prove that $\ket{\phi}$ is  stabilized by all Pauli operators in $\mathcal{S}$. First, for any diagonal stabilizer $D$ we get
\begin{equation}
D\ket{\phi}=D\left( \int_{Y} \mathrm{d}U \, U\right ) \ket{x} = \left( \int_{Y} \mathrm{d}U \, U\right ) D \ket{x} = \left( \int_{Y} \mathrm{d}U \, U\right ) \ket{x},
\end{equation}
due to commutativity and the promise that $D\ket{x}=\ket{x}$. Also, any stabilizer of the form $U_x$ from the set of representatives $Y$ fulfills
\begin{align}
U_x \ket{\phi}&=U_x\left( \int_{\mathbb{H}} \mathrm{d}h\, U_h\right)\ket{x} =\left( \int_{\mathbb{H}} \mathrm{d}h\, U_x U_h\right ) \ket{x} = \left( \int_{\mathbb{H}} \mathrm{d}h\, U_{x+h} \right) D_{x,h} \ket{x}\\
&= \left( \int_{\mathbb{H}} \mathrm{d}h'\, U_{h'} \right) \ket{x} = \ket{\phi}
\end{align}
Hence, using property (a) above, it follows that any arbitrary stabilizer $V$ stabilizes $\ket{\phi}$ as well.
\end{proof}

\begin{corollary}\label{corollary:StabStates have Closed Support}
The sets $\mathbb{H}$ and $\mathrm{supp}(\ket{\psi})=s+\mathbb{H}$ are closed.
\end{corollary}
\begin{proof}
It follows from $(\ref{eq:Support of Stabilizer State 2})$ that supp($\ket{\psi}$) is of the form $x_0+\mathbb{D}^\perp$.  Putting this together with (\ref{eq:Stabilizer_States_r_Uniform_Superpositions}) in lemma \ref{lemma:StabStates_are_Uniform_Superpositions} it follows that $\mathbb{H}=\mathbb{D}^\perp$. Since any annihilator is closed (lemma \ref{lemma:Annihilator properties}),  $\mathbb{H}$ is closed. Since the group operation of $G$ is a continuous map\footnote{This is a fundamental property of topological groups. Consult e.g.\ \cite{Stroppel06_Locally_Compact_Groups,Dikranjan11_IntroTopologicalGroups} for details.},  $s+\mathbb{H}$ is closed too.
\end{proof}

\section{Proof of theorem \ref{thm:Main Result} }\label{sect:Proof of theorem 1}
 
In this section we prove our main result (theorem \ref{thm:Main Result}). As anticipated, we divide the proof in three parts. In section \ref{sect:Tracking normalizer evolutions with stabilizer groups}, we show that the evolution of the quantum state during a normalizer computation can be tracked efficiently using \textbf{stabilizer groups} (which we introduced in the previous section). In section \ref{sect:Computing the support of the final state} we show how to compute the support of the final quantum state by reducing the problem to solving systems of \textbf{{linear equations over an Abelian group}}, which can be reduced to systems of mixed real-integer linear equations \cite{BowmanBurget74_systems-Mixed-Integer_Linear_equations} and solved with the  classical algorithms presented in section \ref{sect:Systems of linear equations over groups}. Finally, in section \ref{sect:Sampling the support of a state}, we show how to simulate the final measurement of a normalizer computation by developing \textbf{net techniques} (based, again, on the algorithms of section \ref{sect:Systems of linear equations over groups}) to sample the support of the final state.

\subsection{Tracking normalizer evolutions with stabilizer groups}\label{sect:Tracking normalizer evolutions with stabilizer groups}

As in the celebrated Gottesman-Knill theorem \cite{Gottesman_PhD_Thesis,Gottesman99_HeisenbergRepresentation_of_Q_Computers}  and  its existing generalizations \cite{Gottesman98Fault_Tolerant_QC_HigherDimensions, dehaene_demoor_hostens,deBeaudrap12_linearised_stabiliser_formalism,VDNest_12_QFTs,BermejoVega_12_GKTheorem}, our approach will be to track the evolution of the system in a stabilizer picture. Since we know that the initial state $|0\rangle$ is a stabilizer state (lemma \ref{lemma:Stabilizer Group for Standard Basis State}) and that normalizer gates are Clifford operations (lemma \ref{lemma:Evolution of Pauli labels}), it follows that the quantum state at every time step of a normalizer computation is  a stabilizer state. It is thus tempting to use stabilizer groups of Abelian-group Pauli operators to classically describe the evolution of the system during the computation; this approach was used in  \cite{VDNest_12_QFTs,BermejoVega_12_GKTheorem} to simulate normalizer circuits over finite Abelian groups. (We remind the reader at this point that  qubit and qudit Clifford circuits are particular instances of normalizer circuits over finite Abelian groups \cite{BermejoVega_12_GKTheorem}.)

However, complications arise compared to all previous cases where normalizer circuits are associated to a finite group $G$. We discuss these issues next.\\

\textbf{Stabilizer groups are infinitely generated.} A common ingredient in all previously known methods to simulate Clifford circuits and normalizer circuits over finite Abelian groups can no longer be used in our setting: traditionally\footnote{As discussed in section ``Relationship to previous work'', there are a few simulation methods \cite{VdNest10_Classical_Simulation_GKT_SlightlyBeyond,Veitch12_Negative_QuasiProbability_Resource_QC,MariEisert12_Positive_Wigner_Functions_Quantum_Computation} for Clifford circuits that are not based on stabilizer-groups, but they are more limited than stabilizer-group methods: the Schrödinger-picture simulation in \cite{VdNest10_Classical_Simulation_GKT_SlightlyBeyond} is  for non-adaptive qubit Clifford circuits; the Wigner-function simulation in \cite{Veitch12_Negative_QuasiProbability_Resource_QC,MariEisert12_Positive_Wigner_Functions_Quantum_Computation} is for odd-dimensional qudit Clifford circuits (cf.\ also section \ref{sect:Previous Work}).}, simulation algorithms based on stabilizer groups keep track of a list of (polynomially  many) generators of a stabilizer group, which can be updated to reflect the action of Clifford/normalizer gates. 
In our set-up,  this is a \emph{futile approach} because \emph{stabilizer groups over infinite Abelian groups can have an \textbf{infinite} number of generators}. Consider for example the state $|0\rangle$ with $G=\Z$, which has a continuous stabilizer group  $\{Z_{G}(p)|p\in\mathbb{T}\}$ (lemma \ref{lemma:Stabilizer Group for Standard Basis State}); the group  that describes the labels of the Pauli operators is the circle group $\T$, which cannot be generated by a finite number of elements (since it is uncountable).\\

\textbf{Fourier transforms change the group $G$}. In previous works \cite{VDNest_12_QFTs, BermejoVega_12_GKTheorem}, the group $G$ associated to a normalizer circuits is a parameter that does not change during the computation. In section \ref{sect:Evolution Pauli Operators} we discussed that our setting is different, as Fourier transforms can change the group that labels the designated basis (theorem \ref{thm:Normalizer gates are Clifford}, eq.\ \ref{eq:Fourier transform on X and Z}; this reflects that groups (\ref{group_hilbert_space},\ref{group_labels_basis}) are not autodual. \\

In this section we will develop new methods to track the evolution of stabilizer groups, that deal with the issues mentioned above. 

From now on, unless stated otherwise, we consider a normalizer circuit ${\cal C}$ comprising $T$ gates. The input is the $\ket{0}$ state of a group $G$, which we denote by $G(0)$ to indicate that this group occurs at time $t=0$.  The stabilizer group of $|0\rangle $ is $\{ Z_G(\mu):  \mu \in\ G(0)^*\}$. The quantum state at any time $t$ during the computation will have the form $\ket{\psi(t)}=\mathcal{C}_t\ket{0}$ where ${\cal C}_t$ is the normalizer circuit containing the first $t$ gates of ${\cal C}$. This state is a stabilizer state over a group $G(t)$. The stabilizer group of $\ket{\psi(t)}$ is  ${\cal S}(t):=\{\mathcal{C}_t Z_{G^*}(\mu) \mathcal{C}_t^\dagger,\:  \mu\in G(0)^* \}$.

Throughout this section, we always assume that normalizer gates are given in the standard encodings defined in section \ref{sect:Main Result}.

\subsubsection*{Tracking the change of  group $G$}

 First, we show how to keep track of how the group $G$ that labels the designated basis changes along the computation. Let $G=G_1\times\dots\times G_m$ with each $G_i$ of primitive type. Define now the larger group $\Gamma:=G^*\times G$. Note that the labels $(\mu, g)$ of a Pauli operator $\gamma Z_G(\mu)X_G(g)$ can be regarded as an element of $\Gamma$, so that the transformations of these labels  in theorem \ref{thm:Normalizer gates are Clifford} can be understood as transformations of this group. We show next that the transformations induces on this group by normalizer gates are \emph{continuous group isomorphisms}, that can be stored in terms of matrix representations. This will give us a method to keep track of $G$ and $G^*$ at the same time. Studying the transformation of $\Gamma$ as a whole (instead of just $G$)  will be useful in the next section, where we consider the evolution of Pauli operators.

First, note that both automorphism gates and quadratic phase gates leave $G$ (and thus $\Gamma$) unchanged (theorem \ref{thm:Normalizer gates are Clifford}). We can keep track of this effect by storing the $2m\times 2m$ identiy matrix $I_{2m}$ (the matrix clearly defines a group automorphism of $\Gamma$). Moreover, (\ref{eq:Fourier transform on X and Z}) shows that Fourier transforms just induce a signed-swap operation on the factors of $\Gamma$. We can associate a $2m\times 2m$ matrix $S_i$ to this operation, defined as follows: $S_i$ acts  non-trivially (under multiplication) only on the factors $G_i^*$ and $G_i$; in the  subgroup  $G_i^*\times G_i$ formed by these factors $S_i$ acts as 
 \begin{equation}\label{eq:Signed Swap}
 (\mu(i), g(i)) \in G_i^* \times G_i \quad \longrightarrow \quad (g(i), -\mu(i)) \in G_i \times G_i^*.
 \end{equation}
By construction, $\Gamma'=S_i\Gamma$. Manifestly, $S_i$ defines a group isomorphism $S_i:\Gamma\rightarrow\Gamma'$.

Lastly, let $G(t)$ denote the underlying group at time step $t$ of the computation. Define $\Gamma(t):=G^*(t)\times G(t)$ and let $V_1,\ldots,V_t$ be the matrices associated to the first $t$ gates describing the transformations of $\Gamma$. Then, we have $\Gamma(t)=V_t V_{t-1}\cdots V_1 \Gamma(0)$, so that it is enough to store the matrix $V_t V_{t-1}\cdots V_1 $ to keep track of the group $\Gamma(t)$. 

\subsubsection*{Tracking Pauli operators}

We deal next with the fact that we can no longer store the ``generators'' of a stabilizer group.  
We will exploit a crucial mathematical property of our stabilizer groups:  for any stabilizer group ${\cal S}$ arising along the course of a normalizer circuit, we will show that there always exists a classical description for $\mathcal{S}$ consisting of a triple $(\Lambda, M, v)$ where $\Lambda$ and $M$ are real matrices and $v$ is a real vector. If we have $G=\T^a\times \Z^b \times \DProd{N}{c}$ with $m=a+b+c$, then all elements of the triple  $(\Lambda,M,v)$ will  have $O(\poly{m})$ entries. As a result, we can use these triples to describe the stabilizer state $\ket{\psi}$ associated to ${\cal S}$ efficiently classically. Moreover, we shall show (lemmas \ref{lemma:Evolution of Pauli labels}, \ref{lemma:Evolution of Pauli Phases}) that the description  $(\Lambda, M, v)$ can be \emph{efficiently transformed} to track the evolution of $\ket{\psi}$ under the action of a normalizer circuit.

Let $\Gamma(t)$ be the group $G^*(t)\times G(t)$. Recalling the definition of the group $\mathbb{L}$ in (\ref{label_groups}), we denote by $\mathbb{L}(t)\subseteq \Gamma(t)$ this group at time $t$. We want to keep track of this group in a way that does not involve storing an infinite number of generators. As a first step, we consider the initial standard basis state  $|0\rangle$, where \be \mathbb{L}(0)= \{(\mu,0):\mu\in G(0)^*\}.\ee  
A key observation is that this group can be written as the image of a continuous group homomorphism
\be \Lambda_0: (\mu, g)\in \Gamma(0)\rightarrow (\mu, 0)\in \Gamma(0);\ee it is easy to verify  $\mathbb{L}(0) = \mathrm{im}\,\Lambda_0$. Therefore, in order to keep track of the (potentially uncountable) set $\mathbb{L}(0)$ it is enough to store a $2m\times 2m$ matrix representation of $\Lambda_0$ (which we  denote by the same symbol):
 \begin{equation}\label{Lambda_0}
\Lambda_0  =\begin{pmatrix}
I & 0 \\ 0 & 0
\end{pmatrix}
 \end{equation}
Motivated by this property,  we will track the evolution of the group $\mathbb{L}(t)$ of Pauli-operator labels by means of a matrix representation of a group homomorphism: $\Lambda_t: \Gamma(0)\rightarrow\Gamma(t)$ whose image is precisely $\mathbb{L}(t)$. The following lemma states that this approach works.
\begin{lemma}[\textbf{Evolution of Pauli labels}]\label{lemma:Evolution of Pauli labels}
There exists a group homomorphism $\Lambda_t$
from $\Gamma(0)$ to $\Gamma(t)$ satisfying
\begin{equation}
\mathbb{L}(t)=\mathrm{im}\, \Lambda_t.
\end{equation}
Moreover,  a  matrix representation of $\Lambda_t$ can be computed in classical polynomial time,  using $O(\ppoly{m,t})$ basic arithmetic operations.
\end{lemma}
\begin{proof}
We show this by induction. As discussed above, at $t=0$
we choose  $\Lambda_0$ as in (\ref{Lambda_0}). Now, given the homomorphism $\Lambda_t$ at time $t$, we show how to compute $\Lambda_{t+1}$ for every type of normalizer gate. The proof relies heavily on the identities in the proof of theorem \ref{thm:Normalizer gates are Clifford}. We also note that the equations below are for groups of commuting Pauli operators but they can be readily applied to any single Pauli operator just by considering the stabilizer group it generates.
\begin{itemize}
\item Automorphism gate $U_{\alpha}$: Let $A$ be a matrix representation of $\alpha$; then equations (\ref{eq:Automorphism Gate on X})-
(\ref{eq:Automorphism Gate on Z}) imply

\begin{equation}
\Lambda_{t+1}=\left(\begin{array}{cc}
A^{*^{\inverse}} & 0\\
0 & A
\end{array}\right)\Lambda_t.
\end{equation}
The matrix $A^{*^{\inverse}}$ can be computed efficiently due to lemmas  \ref{lemma:Computing Inverses} and  \ref{lemma:properties of matrix representations}.(b).
\item Quadratic phase gate $D_{\xi}$: suppose that $\xi$ is a $B$-representation for some bicharacter $B$ (recall section \ref{sect:quadratic_functions}). Let $M$ be a matrix representation of the homomorphism $\beta$ that appears in lemma \ref{lemma:Normal form of a bicharacter 1}.  Then  (\ref{eq:QuadraticPhase Gate on X}) implies
\begin{equation}
\Lambda_{t+1}=\left(\begin{array}{cc}
I & M \\
0 & I
\end{array}\right)\Lambda_t.
\end{equation}
\item Partial Fourier transform $\mathcal{F}_{G_i}$: recalling (\ref{eq:Fourier transform on X and Z}), we  simply have
\begin{equation}
\Lambda(t+1)=S_i\Lambda(t),
\end{equation}
with
\begin{equation}
S_i = 
\begin{pmatrix}
\begin{matrix}
\textbf{ \Large 1} & & \\
& 0 & \\
& & \textbf{ \Large 1}
\end{matrix}  &\vline &
\begin{matrix}
\textbf{ \Large 0} & & \\
& 1 & \\
& & \textbf{ \Large 0}
\end{matrix} \\
\hline
\begin{matrix}
\textbf{ \Large 0} & & \\
& -1 & \\
& & \textbf{ \Large 0}
\end{matrix} &\vline & \begin{matrix}
\textbf{ \Large 1} & & \\
& 0 & \\
& & \textbf{ \Large 1}
\end{matrix}
\end{pmatrix}
\end{equation}
where the $ \begin{pmatrix}
0 & 1 \\ -1 & 0
\end{pmatrix}$ subblock in $S_i$ corresponds to the $i$-th entries of $G^*$ and $G$. \qedhere
\end{itemize}
\end{proof}

We now show how the phases of the Pauli operators in ${\cal S}(t)$ can be tracked.  

Suppose that there exists $(\mu, g)\in \mathbb{L}$ and complex phases $\gamma$ and $\beta$ such that both \be \sigma:=\gamma Z(\mu)X(g)\quad \mbox{ and } \tau:=\beta Z(\mu)X(g)\ee belong to ${\cal S}$.   
Then $\sigma^{\dagger}\tau$ must also belong to ${\cal S}$, where $\sigma^{\dagger}\tau= \overline{\gamma}\beta I $ with $I$ the identity operator. But this implies that $\overline{\gamma}\beta|\psi\rangle = |\psi\rangle$, so that $\gamma = \beta$. This shows that the phase of $\sigma$ is uniquely determined by the couple $(\mu, g)\in\mathbb{L}$.
We may thus define a function $\gamma: \mathbb{L}\to U(1)$ such that \be {\cal S}= \{\gamma(\mu, g)Z(\mu)X(g):\ (\mu, g)\in \mathbb{L}\}.\ee
\begin{lemma}\label{thm_gamma_quadratic}
The function $\gamma$ is a quadratic function on $\mathbb{L}$.
\end{lemma}
\begin{proof}
By comparing the phases of two stabilizer operators $\sigma_1=\gamma(\mu_1,g_1) Z_{G}(\mu_1)X_{G}(g_1)$ and $\sigma_2=\gamma(\mu_2,g_2)\allowbreak Z_{G}(\mu_2)X_{G}(g_2)$ to the phase $ \gamma((g_1,h_1)+(g_2,h_2))$ of their product operator $\sigma_2\sigma_1$, we obtain
 \begin{equation}
 \gamma((\mu_1,g_1)+(\mu_2,g_2))=\gamma(\mu_1,g_1)\gamma(\mu_2,g_2)\overline{\chi_{\mu_2}(g_1)},
 \end{equation}
 which implies that $\gamma$ is quadratic.
 \end{proof}
Although it does not follow from lemma \ref{thm_gamma_quadratic}, in our setting, the quadratic function $\gamma$ will always be \emph{continuous}. As a result, we can apply the normal form given in theorem \ref{thm:Normal form of a quadratic function} to describe  the phases of the Pauli operators of a stabilizer group. Intuitively, $\gamma$ must be continuous in our setting, since  this is the case for the allowed family of input states (lemma \ref{lemma:Stabilizer Group for Standard Basis State}) and normalizer gates continuously transform Pauli operators under conjugation; this is rigorously shown using induction in the proof of  theorem \ref{thm_evolution_phases}.

We will use that these phases of Pauli operators are described by quadratic functions on $\mathbb{L}(t)$ (recall lemma \ref{thm_gamma_quadratic}).  In particular, theorem \ref{thm:Normal form of a quadratic function} shows that every quadratic function can be described by means of an $m\times m$  matrix $M$ and a $m$-dimensional vector $v$. For the initial state $\ket{0}$, we  simply set both $M$, $v$ to be zero. The next lemma shows that $M$, $v$ can be efficiently updated during any normalizer computation.

\begin{lemma}[\textbf{Evolution of Pauli phases}]\label{lemma:Evolution of Pauli Phases} \label{thm_evolution_phases}
At every time step $t$ of a normalizer circuit, there exists a $2m\times 2m$ rational matrix $M_t$ and a $m$-dimensional rational vector $v_t$ such that the quadratic function describing the phases of the Pauli operators in $\mathcal{S}(t)$ is $\xi_{M_t,v_t}$ (as in theorem \ref{thm:Normal form of a quadratic function}). Moreover, $M_t$ and $v_t$ can be efficiently computed classically with $O(\ppoly{m,n})$ basic arithmetic operations.
\end{lemma}
\begin{proof} The proof is similar to the proof of lemma \ref{lemma:Evolution of Pauli labels}. We act by induction. At $t=0$ we just take $M_0$ to be the zero matrix and $v_0$ to be the zero vector. Then, given $M_t$ and $v_t$ at time $t$, we show how to compute $M_{t+1}$, $v_{t+1}$. In the following, we denote by $\mathbf{A}$ the matrix that fulfills $\Lambda_{t+1}=\mathbf{A}\Lambda_{t}$ in each case of lemma \ref{lemma:Evolution of Pauli labels} and write $(\mu',g')=\mathbf{A}(\mu,g)$ for every $(\mu, g)\in \Gamma_{t}$. Finally, let $\xi_t$ and $\xi_{t+1}$ denote the quadratic phase functions for ${\cal S}(t)$ and ${\cal S}(t+1)$, respectively.
\begin{itemize}
\item \textbf{Automorphism gate $U_\alpha$.} Let $A$, $A^{*^{\inverse}}$ be matrix representations of $\alpha$, $\alpha^{*^{\inverse}}$. Using (\ref{eq:Automorphism Gate on X}, \ref{eq:Automorphism Gate on Z}) we have
\begin{align}
\xi_t(\mu,g)Z_G(\mu)X_{G}(g)\:\xrightarrow{U_\alpha} \:\xi_t(\mu,g)Z_G(\mu')X_{G}(g')
\end{align}
with $(\mu',g')=\mathbf{A}(\mu, g)$ and $\mathbf{A}= \begin{pmatrix}
A^{*^{\inverse}} & 0 \\ 0 & A
\end{pmatrix}$. The matrix $A^{*^{\inverse}}$ can be computed using lemmas \ref{lemma:Computing Inverses} and \ref{lemma:properties of matrix representations}.(b). The phase $\xi_t(\mu,g)$ of the Pauli operator can be written now as a function  $\xi_{t+1}$ of $(\mu',g')$ defined as
\begin{equation}
\xi_{t+1}(\mu',g'):=\xi_{t}(\mathbf{A}^{-1}(\mu', g')) = \xi_{t}(\mu,g).
\end{equation}
The function is manifestly quadratic. By applying lemma \ref{lemma:Quadratic Function composed with Automorphism} we obtain
\begin{equation}
 M_{t+1}= \mathbf{A}^{-\transpose} M_t \mathbf{A}^{-1}, \qquad v_{t+1} =\mathbf{A}^{-\transpose} v_t + v_{\mathbf{A}^{-1},M_t},
\end{equation}
where $v_{\mathbf{A}^{-1},M_t}$ is defined as $v_{A,M}$ in lemma \ref{lemma:Quadratic Function composed with Automorphism}.

\item \textbf{Partial Fourier transform $\mathcal{F}_{G_i}$.}  The proof is analogous using that  $\mathbf{A}= S_i$. Since the Fourier transform at the register $i$th exchanges the order of the X and Z Pauli operators acting on the subsystem $\mathcal{H}_{G_i}$ (\ref{eq:Fourier transform on X and Z}), we locally exchange the operators locally using (\ref{eq:pauli_commutation}), gaining an extra phase. Assume for simplicity that $i=1$ and re-write $G=G_1\times \cdots\times G_m$ as $G= A \times B$; let $g=(a,b)$ and $\mu=(\alpha, \beta)$. Then  $\mathcal{F}_{G_1}$ acts trivially on $\mathcal{H}_{G'}$ and we get
\begin{align}\notag
\xi_t(\mu,g)Z_{G_1}(\alpha)X_{G_1}(a)\otimes U \:\xrightarrow{\mathcal{F}_{G_1}+\mathrm{reorder}} \:\left(\xi_{t}(\mu,g)\chi_{\left(\alpha,0\right)}(a,0)\right) Z_{G_1^*}(a)X_{G_1^*}(-\alpha)\otimes U.
\end{align}
In general, for arbitrary $i$, we gain a phase factor $\overline{\chi_{\left( 0,\ldots,\mu(i),\ldots,0\right)}{\left( (0,\ldots, g(i),\ldots,0)\right)}}$. Using the change of variables $(\mu',g')=\mathbf{A}(\mu, g)=S_i(\mu, g)$, we define $\xi_{t+1}$ to be function that carries on the accumulated phase of the operator. For arbitrary $i$ we obtain
\begin{align}
\xi_{t+1}(\mu',g')&:=\xi_{t}(\mu,g)\,\chi_{\left( 0,\ldots,\mu(i),\ldots,0\right)}{\left( (0,\ldots, g(i),\ldots,0)\right)}.
\end{align}
The character $\chi_{\left( 0,\ldots,\mu(i),\ldots,0\right)}{\left( (0,\ldots, g(i),\ldots,0)\right)}$ can be written as a quadratic function $\xi_{M_F,v_F}(\mu, g)$ with $v_F=0$ and 
\begin{equation}
M_{F}:=\begin{pmatrix}
\begin{matrix}
 & & \\
& \textbf{ {\Huge 0}} & \\
&&
\end{matrix}  &\vline &
\begin{matrix}
\textbf{ \Large 0} & & \\
& \Upsilon_G(i,i) & \\
& & \textbf{ \Large 0}
\end{matrix} \\
\hline
\begin{matrix}
\textbf{ \Large 0} & & \\
& \Upsilon_G(i,i) & \\
& & \textbf{ \Large 0}
\end{matrix} &\vline & \begin{matrix}
& & \\
& \textbf{ {\Huge 0}} & \\
& &
\end{matrix}
\end{pmatrix},
\end{equation}
where $\Upsilon_G(i,i)$ is the $i$th diagonal element of  $\Upsilon_G$ (\ref{eq:definition of bullet map}). Applying lemma \ref{lemma:Quadratic Function composed with Automorphism}  we obtain
\begin{equation}
 M_{t+1}= \mathbf{A}^{-\transpose}\left( M_t +M_{F}\right) \mathbf{A}^{-1}, \qquad v_{t+1} =\mathbf{A}^{-\transpose} v_t +  v_{\mathbf{A}^{-1},M_t+M_F}.
\end{equation}
\item \textbf{Quadratic phase gate $D_\xi$.} Let $\xi=\xi_{M_Q,v_Q}$ be the quadratic function implemented by the gate and $M_\beta$ be the matrix representation of $\beta$ as in (\ref{lemma:Normal form of a bicharacter 1}). We know from lemma \ref{lemma symmetric matrix representation of the bicharacter homomorphism} that $M_Q=\Upsilon_G M_\beta$. Using (\ref{eq:QuadraticPhase Gate on X}) and reordering Pauli gates (similarly to the previous case) we get
\begin{align}\notag
\xi_t(\mu,g)Z_{G}(\mu)X_{G}(g) \:\xrightarrow{D_\xi+\mathrm{reorder}} \:\left( \xi_{t}(\mu,g)\xi_{M_Q,v_Q}(g)\overline{\chi_{\beta(g)}(g)}\right) Z_{G}(\mu+\beta(g))X_{G}(g)
\end{align}
The accumulated phase can be written as a quadratic function $\xi_{M',v'}$ with  \begin{equation}
M':=  M_t + \begin{pmatrix}
 0 & 0 \\ 0 & M_Q
 \end{pmatrix} -\begin{pmatrix}
 0 & 0 \\ 0 & 2M_Q
 \end{pmatrix}, \qquad v' := v + \begin{pmatrix}
 0\\ v_Q
 \end{pmatrix}
 \end{equation}
 Using lemma \ref{lemma:Quadratic Function composed with Automorphism} and  $\mathbf{A}= \begin{pmatrix}
I & M_\beta \\ 0 & I
\end{pmatrix}$  (from the proof of lemma \ref{lemma:Evolution of Pauli labels}) we arrive at:
\begin{align}
 M_{t+1}&=\mathbf{A}^{-\transpose} M' \mathbf{A}^{-1}, \qquad
v' = \mathbf{A}^{-\transpose} v' +  v'_{\mathbf{A}^{-1}, M'}
 \end{align}\qedhere
\end{itemize}
\end{proof}
Combining lemmas \ref{lemma:Evolution of Pauli labels} and \ref{thm_evolution_phases}, we find that 
the triple $(\Lambda_t, M_t, v_t)$, which  constitutes a classical description of the stabilizer state $|\psi(t)\rangle$, can be efficiently computed for all $t$. This yields a poly-time algorithm to compute the description $(\Lambda_T, M_T, v_T)$ of the output state $|\psi_T\rangle$ of the circuit.  Henceforth we continue to work with this final state and drop the reference to $T$ throughout. That is, the final state is denoted by $|\psi\rangle$, which is a stabilizer state over $G$ with stabilizer is ${\cal S}$. The latter is described by the triple $(\Lambda, M, v)$, the map from $\Gamma(0)$ to $\Gamma$ is described by $\Lambda$, etc.

\subsection{Computing the support of the final state}\label{sect:Computing the support of the final state}

Given the triple $(\Lambda, M, v)$ describing the final state $|\psi\rangle$ of the computation,  
we now consider the problem of determining the support of $\ket{\psi}$. Recall that the latter has the form $x+\mathbb{H}$ where the label group $\mathbb{H}$ was defined in (\ref{label_groups}) and $x\in G$ is any element satisfying conditions (\ref{eq:Support of Stabilizer State 1}). Since $\mathbb{L}= \Lambda \Gamma(0)$ and $\Lambda$ is given, a description of $\mathbb{H}$ is readily obtained:  the $m\times 2m$ matrix $P = (0 \ I)$
is a matrix representation of the homomorphism  $(\mu, g)\in \Gamma\to g\in G$. It easily follows that $\mathbb{H}= P\Lambda\Gamma(0)$. Thus the matrix $P\Lambda$ yields an efficient description for $\mathbb{H}$. To compute an $x$ in the support of $|\psi\rangle$, we need to solve the equations (\ref{eq:Support of Stabilizer State 1}).  In the case of finite groups $G$, treated in previous works \cite{VDNest_12_QFTs,BermejoVega_12_GKTheorem}, the approach consisted of first computing a (finite) set of generators $\{D_1, \dots, D_r\}$ of $\mathcal{D}$. Note that $x\in G$ satisfies (\ref{eq:Support of Stabilizer State 1}) if and only if $D_i|x\rangle = |x\rangle$ for all $i$. This gives rise to a finite number of equations. In \cite{VDNest_12_QFTs,BermejoVega_12_GKTheorem} it was subsequently showed how such equations can be solved efficiently.  In contrast with such a finite group setting, here the group $G$, and hence also the group $\mathcal{D}$, can be continuous, so that $\mathcal{D}$ can in general not be described by a finite list of generators. Consequently, the approach followed for finite groups does no longer work. Next we provide an alternative approach to compute an $x$ in the support of $|\psi\rangle$ in polynomial time.

\subsubsection{Computing $\mathcal{D}$}

We want to solve the system of equations (\ref{eq:Support of Stabilizer State 2}). Our approach will be to reduce  
this problem to a system of linear equations over a group of the form (\ref{eq:Systems of linear equations over groups}) and apply the algorithm in theorem \ref{thm:General Solution of systems of linear equations over elementary LCA groups} to solve it. To compute $\mathcal{D}$ it is enough to find a compact way to represent $\mathbb{D}$, since we can compute the phases of the diagonal operators using the classical description $(\Lambda, M,v )$ of the stabilizer group. To compute $\mathbb{D}$ we argue as follows. An arbitrary element of $\mathbb{L}$ has the form $\Lambda u$ with $u\in \Gamma(0)$. Write $\Lambda$ in a block form
\be
\Lambda =
\begin{pmatrix}
\Lambda_1\\ \Lambda_2
\end{pmatrix}
\ee
so that $\Lambda u = (\Lambda_1 u, \Lambda_2 u)$ with $\Lambda_1 u\in G^*$ and $\Lambda_2 u\in G$. Then
\begin{equation}\notag
\mathbb{D} =\left\{\Lambda_1 u:  u \mbox{ satisfies } \Lambda_2 u \equiv 0 \mbox{ mod } G.
\right\}
\end{equation}
The equation $\Lambda_2 u \equiv 0 \mbox{ mod } G$ is  of the form (\ref{eq:Systems of linear equations over groups}). This means  we can compute in polynomial time a description for $\mathbb{D}$  of the form \be \mathbb{D}=\{\mathcal{E}_\mathbb{D}w:\ w \in \R^a\times \Z^b\},\ee where $\mathcal{E}_\mathbb{D}$ is a group homomorphism $\mathcal{E}_\mathbb{D}:\R^a\times \Z^b\rightarrow G^*$ whose image is precisely $\mathbb{D}$.

\subsubsection{Computing the support $x_0+\mathbb{H}$}

Recalling the support equations (\ref{eq:Support of Stabilizer State 2}) and the fact that $|\psi\rangle$ is described by the triple $(\Lambda, M, v)$, we find that $x_0$ belongs to the support of $|\psi\rangle$ if and only if
\begin{equation}\label{eq_support_final}
\xi_{M,v}(\mu,0)\chi_{\mu}(x_0)=1,\quad\textrm{for all }\mu\in \mathbb{D}.
\end{equation} We will now write the elements $\mu\in \mathbb{D}$ in the form $\mu=\mathcal{E}_{\mathbb{D}}w$ where $w$ is an arbitrary element in $\R^a\times \Z^b$. We further denote $\mathcal{E}:=\begin{pmatrix}
\mathcal{E}_\mathbb{D} \\ 0
\end{pmatrix}$. We now realize that
\begin{itemize}
\item $\xi_{M,v}(\mathcal{E}_\mathbb{D}w,0)$, as a function of $w$ only, is a quadratic function of $\R^a\times\Z^{b}$, since $\xi_{M,v}$ is quadratic and $\mathcal{E}_\mathbb{D}$ is a homomorphism. Furthermore \be \xi_{M,v}(\mathcal{E}_\mathbb{D}w,0)= \xi_{M',v'}(w)\qquad \mbox{ with } M':=\mathcal{E}^\transpose M\mathcal{E},\ v':=\mathcal{E}^\transpose v.
    \ee
\item $\chi_{\mathcal{E}_{\mathbb{D} }w } (x_0) $, as a function of $w$ only, is a character function of $\R^a\times\Z^{b}$ which can be written as $\chi_\varpi$ with $\varpi:={\mathcal{E}_\mathbb{D}}^*(x_0)$.
\end{itemize}
It follows that $x_0$ satisfies (\ref{eq_support_final}) if and only if the quadratic function $\xi_{M',v'}$ is a character and coincides with $\chi_\varpi$. Using lemma \ref{lemma symmetric matrix representation of the bicharacter homomorphism} and theorem \ref{thm:Normal form of a quadratic function}, we can write these two conditions equivalently as:
\begin{align}
 w_1^\transpose M' w_2 &= 0 \pmod{\Z}, \quad \textnormal{for all $w_1$, $w_2\in  \R^a\times \Z^b$}\\
\mathcal{E}_\mathbb{D}^*(x_0)&=\mathcal{E}^\transpose v \pmod{\R^a\times \T^b}.
\end{align}
The first equation does not depend on $x_0$ and  it must hold by promise:  we are guaranteed that the support is not empty, so that the above equations must admit a solution. The second equation is a system of linear equations over groups of the form given in section \ref{sect:Systems of linear equations over groups}, and it can be solved with the techniques given in that section.

\subsection{Sampling the support of a state}\label{sect:Sampling the support of a state}

Back in section \ref{sect:Systems of linear equations over groups} we formulated a fairly simple heuristic to sample the solution space of a linear system of equations over elementary Abelian groups (\ref{eq:Systems of linear equations over groups},\ref{eq:Systems of linear equations over groups: Solution Space}) that exploited our ability to compute general solutions of such systems (theorem \ref{thm:General Solution of systems of linear equations over elementary LCA groups}). Unfortunately, this straightforward method does not yield an efficient algorithm  to sample such solution spaces, which  would allow us to efficiently simulate classically quantum normalizer circuits.   
In the first place, the heuristic neglects two delicate mathematical properties of the groups under consideration, namely, that they are \emph{continuous} and \emph{unbounded}. Moreover, the second step of the heuristic involves the transformation of a given probability distribution  on a space $\mathcal{X}$ by the application of a non-injective map $\mathcal{E}:\mathcal{X}\rightarrow G$; this step is prone to create a wild number of \emph{collisions}, about which the heuristic gives no information.

In this section we will present an \emph{efficient classical algorithms} to sample solution space of systems of linear equations over groups. Our algorithm applies appropriate techniques to tackle the previously mentioned issues. The algorithm relies on a subroutine to construct and sample from a certain type of epsilon net that allows the \emph{collision-free} sampling from a subgroup of an elementary group, where the subgroup is given as the image of a homomorphism. This algorithm is sufficient for our purposes, since the solution-space (\ref{eq:Systems of linear equations over groups: Solution Space}) of (\ref{eq:Systems of linear equations over groups}) is exactly $x_0 + \mathrm{im}\, \mathcal{E}$ for some homomorphism $\mathcal{E}$ and group element $x_0$.

\subsubsection*{Input of the problem and assumptions}

Again let $G$ be of the form
\begin{equation}\label{eq:Elementary TZF Groups 1}
G=\T^a \times \Z^b \times \DProd{N}{c}
\end{equation}
with $m=a+b+c$.  We are given a matrix representation $\mathcal{E}$ of a group homomorphism from  $\R^\alpha\times\Z^\beta$ to $G$ such that $H$ is the image of $\mathcal{E}$. The matrix $\mathcal{E}$ and the numbers $\alpha$, $\beta$ provide a description of the subgroup $H$. 

Throughout the entire section, $H$ is assumed to be \textbf{closed} (in the topological sense). The word ``subgroup'' will be used as a synonym of ``closed subgroup''. This is enough for our purposes, since the subgroup $\mathbb{H}$ that defines the support of a stabilizer state (and we aim to sample) is always closed (corollary \ref{corollary:StabStates have Closed Support}).

\subsubsection*{Norms}

 There exists a natural notion of 2-norm  for every group of the form $G:=\Z^a\times \R^b \times \DProd{N}{c}\times \T^d$ analogous to the standard 2-norm $\|\cdot\|_2$ of a real Euclidean space (we denote the group 2-norm simply by $\|\cdot \|_{G}$):  given  $g=(g_\Z, g_\R, g_F, g_\T)\in G$, 
\begin{equation}\label{eq:Norms}
\|g\|_G:= \left\| \left( g_\Z, g_\R,  g_F^{\circlearrowleft}, g_\T^{\circlearrowleft} \right) 	\right\|_2
\end{equation}
where  $g_F^{\circlearrowleft}$ (resp.\ $g_\T^{\circlearrowleft}$) stands for any integer  tuple $x\in \Z^a$  (resp.\ real tuple $y\in \R^d$)  that is congruent to $g_F$ (resp.\ $g_\T$) and has minimal two norm $\|\cdot\|_2$. The reader should note that, although $g_F^{\circlearrowleft}$, $g_\T^{\circlearrowleft}$ may not be uniquely defined, the value of $\|g\|_G$ is always unique.

The following relationship between norms will later be useful:
\begin{equation}\label{eq:Relationship between norms}
\textnormal{if}\quad\|g\|_{2}\leq \tfrac{1}{2}\quad \textnormal{then}\quad \|g\|_{G}=\|g\|_{2}\leq \tfrac{1}{2};
\end{equation}
or, in other words, if an element $g\in G$ has small $\|\cdot\|_2$ norm as a tuple of real numbers, then its norm $\|g\|_{G}$ as a group element of $G$ is also small and equal to $\|g\|_2$.

\subsubsection*{Net techniques}

Groups of the form (\ref{eq:Elementary TZF Groups 1}) contain subgroups that are \emph{continuous} and/or \emph{unbounded} as sets. These properties must  be taken into account in the design of algorithms to sample subgroups. 
 
We briefly discuss the technical issues---absent from the case of finite $G$ as in \cite{VDNest_12_QFTs,BermejoVega_12_GKTheorem} ---that arise, and present net techniques to tackle them.
 
The first issue to confront, related to continuity, is the presence of discretization errors due to finite precision limitations, for no realistic algorithm can sample a continuous subgroup $H$ exactly. Instead, we will  sample some distinguished discrete subset $\mathcal{N}_\varepsilon$  of $H$  that, informally, ``discretizes''  $H$ and that can be efficiently represented in a computer. More precisely, we choose $\mathcal{N}_\epsilon$ to be a certain type of $\varepsilon$-net:
 \begin{definition}[$\boldsymbol{\varepsilon}$\textbf{-net}\footnote{Our definition of  $\varepsilon$-net is based on the ones used in \cite{HaydenLeungShorWinter04Randomizing_Quantum_States,YaoyunXiaodi@12_E_Net_Method_Optimizations_Separable,Ni13Commuting_Circuits,Deza13Encyclopedia_of_Distances}. We adopt an additional non-standard convention, that $\mathcal{N}$ must be a subgroup, because it is convenient for our purposes.}]\label{def:Epsilon Net} An $\varepsilon$-net $\mathcal{N}$ of a subgroup $H$ is a finitely generated subgroup of $H$ such that for every $h\in H$ there exists  $\mathpzc{n}\in \mathcal{N}$ with $\| h -  \mathpzc{n} \|_{G}\leq \varepsilon$.
 \end{definition}
The second issue in our setting is the unboundedness of certain subgroups of $G$ \emph{by itself}. We must carefully define  a notion of sampling for such sets that suits our needs, dealing with the fact that  uniform distributions over unbounded sets (like $\R$ or $\Z$) cannot be interpreted as well-defined probability distributions; as a consequence, one cannot simply ``sample'' $\mathcal{N}$ or $H$ uniformly. However, in order to simulate the distribution of measurement outcomes of a \emph{physical} normalizer quantum computation (where the initial states $\ket{g}$ can only be prepared approximately) it is enough to sample uniformly some bounded compact region of $H$ with finite volume $V$. We can approach the infinite-precision limit by choosing $V$ to be larger and larger, and in the $V\rightarrow \infty$ limit we will approach an exact quantum normalizer computation. 

We will slightly modify the definition of $\epsilon$-net so that we can sample $H$ in the sense described above. For this, we need to review some structural properties of the subgroups of groups of the form (\ref{eq:Elementary TZF Groups 1})

It is known that any arbitrary closed subgroup $H$ of an elementary group $G$ of the form (\ref{eq:Elementary TZF Groups 1}) is isomorphic to an elementary group also of the form  (\ref{eq:Elementary TZF Groups 1})  (see \cite{Stroppel06_Locally_Compact_Groups} theorem 21.19 and proposition 21.13). As a result, any subgroup $H$ is of the form $H=H_{\mathrm{comp}}\oplus H_{\mathrm{free}}$ where $H_{\mathrm{comp}}$ is a \emph{compact} Abelian subgroup of $H$ and $H_{\mathrm{free}}$ is either the trivial subgroup or an \emph{unbounded} subgroup that does not contain non-zero finite-order elements (it is \emph{torsion-free}, in group theoretical jargon). By the same argument, any $\varepsilon$-net $\mathcal{N}_\varepsilon$ of $H$ decomposes in the same way
 \begin{equation}\label{eq:eNet_compact-free_decomposition}
 \mathcal{N}_\varepsilon:=\mathcal{N}_{\varepsilon,\mathrm{comp}}\oplus \mathcal{N}_{\varepsilon,\mathrm{free}}.
 \end{equation}
 where $\mathcal{N}_{\varepsilon,\mathrm{comp}}$ is a finite subgroup of $H_{\mathrm{comp}}$ and $\mathcal{N}_{\varepsilon,\mathrm{free}}$ is a finitely generated torsion-free subgroup of $H_{\mathrm{free}}$. The fundamental theorem of finitely generated Abelian groups tells us that $\mathcal{N}_{\varepsilon,\mathrm{free}}$ is  isomorphic to a group of the form $\Z^\mathbf{r}$ (a \emph{lattice} of rank $\mathbf{r}$) and, therefore, it has a \emph{$\Z$-basis} \cite{Mollin09Advanced_Number_Theory_with_Applications}: i.e.\ a set  $\{\mathpzc{b}_1, \ldots, 	\mathpzc{b}_\mathbf{r}\}$ of  elements  such that every $\mathpzc{n}\in \mathcal{N}_{\varepsilon,\mathrm{free}}$ can be written in one and only one way as a linear combination of basis elements with \emph{integer} coefficients:
  \begin{equation}
   \mathcal{N}_{\varepsilon,\mathrm{free}}=\left\{\mathpzc{n} = \sum_{i=1}^{\mathbf{r}} \mathpzc{n}_i  \mathpzc{b}_i, \,\textnormal{ for some } \mathpzc{n}_i \in \Z\right\}.
   \end{equation}
    
In view of equation (\ref{eq:Elementary TZF Groups 1}) we introduce a more general notion of nets that is adequate for sampling this type of set.
  \begin{definition}\label{def:Delta Epsilon Net} Let $\mathcal{N}_\varepsilon$ be an $\varepsilon$-net of $H$ and let $\{\mathpzc{b}_1, \ldots, 	\mathpzc{b}_\mathbf{r}\}$ be a prescribed  basis of $\mathcal{N}_{\varepsilon,\mathrm{free}}$. Then, we call a $\boldsymbol{(\Delta, \varepsilon )}$\textbf{\em-net} any  finite subset $\mathcal{N}_{\Delta,\varepsilon}$ of $\mathcal{N}_\varepsilon$ of the form
  \begin{equation}
   \mathcal{N}_{\Delta,\varepsilon}= \mathcal{N}_{\varepsilon,\mathrm{comp}}\oplus \mathcal{P}_\Delta,
   \end{equation}
  where  $\mathcal{P}_\Delta$ denotes the  parallelotope contained in $\mathcal{N}_{\varepsilon,\mathrm{free}}$ with vertices $\pm\Delta_1 \mathpzc{b}_1, \ldots, 	\pm\Delta_\mathbf{r} \mathpzc{b}_\mathbf{r}$,
    \begin{equation}\label{eq:Parallelotope DEFINITION}
    \mathcal{P}_\Delta := \left\{\mathpzc{n} = \sum_{i=1}^{\mathbf{r}} \mathpzc{n}_i  \mathpzc{b}_i, \,\textnormal{ where } \mathpzc{n}_i \in \{0,\pm 1, \pm 2,\ldots , \pm \Delta_i\}\right\}.
    \end{equation}
    The index of $\mathcal{P}_\Delta$ is a  tuple of positive integers $\Delta:=(\Delta_1,\ldots,\Delta_\mathbf{r})$ that specifies the lengths of the edges of $\mathcal{P}_\Delta$.
 \end{definition}
 Notice that $\mathcal{N}_{\Delta,\varepsilon}\rightarrow \mathcal{N}_\varepsilon$ in the limit where the edges  $\Delta_i$ of $\mathcal{P}_\Delta$ become infinitely long and that the volume covered by $\mathcal{N}_{\Delta,\varepsilon}$ increases monotonically as a function of the edge-lengths. Hence, any algorithm to construct and sample $(\Delta, \varepsilon)$-nets of $H$ can be used to sample  $H$ in the sense we want. Moreover, the next theorem (a  main contribution of this paper) states that there exist  classical algorithms to sample the subgroup $H$ through $(\Delta,\varepsilon)$-nets \emph{efficiently}.
 \begin{theorem}\label{thm:Algorithm to sample subgroups}
Let  $H$ be an arbitrary closed subgroup of an elementary group $G=\T^a\times\Z^b\times \DProd{N}{c}$. Assume we are given a matrix-representation $\mathcal{E}$ of a group homomorphism $\mathcal{E}:\R^\alpha\times \Z^\beta\rightarrow G$ such that $H$ is the image of $\mathcal{E}$. Then, there exist  classical algorithms to sample $H$ through $(\Delta, \varepsilon)$-nets using $O(\ppoly{m, \alpha,\beta, \log{N_i},\log\|\mathcal{E}\|_\mathbf{b}, \log\frac{1}{\varepsilon}, \log\Delta_i})$ time and bits of memory.
 \end{theorem}
 Again, $\log\|\mathcal{E}\|_\mathbf{b}$ denotes the maximal number of bits needed to store a coefficient of $\mathcal{E}$ as a fraction. The proof is the content of the next section, where we devise a classical algorithm with the advertised properties.

\subsubsection*{Proof of theorem \ref{thm:Algorithm to sample subgroups}: an algorithm to sample subgroups}

We  denote by $\mathcal{E}_{\T\R}$ the block of $\mathcal{E}$ with image contained in $\T^a$ and with domain $\mathbb{R}^\alpha$. Define  a new set  $\mathcal{L}:= (\varepsilon_1 \Z) ^\alpha \times \Z^\beta$, which is a subgroup of $\R^a\times \Z^b$, and let   $\mathcal{N}:= \mathcal{E}(\mathcal{L})$ be the image of $\mathcal{L}$ under the action of the homomorphism $\mathcal{E}$.  

In first place, we show that  by setting $\varepsilon_1$ to be  smaller than $2\varepsilon/(\alpha \sqrt{a}|\mathcal{E}|)$, we can ensure that $\mathcal{N}$ is an $\varepsilon$-net of $H$ for any $\varepsilon$ of our choice. We will use that $\mathcal{L}$ is, by definition,  a $(\frac{\varepsilon_1\sqrt{\alpha}}{2})$-net of $\R^\alpha\times \Z^\beta$. (This follows from the fact that, for every   $x\in\R^\alpha$ there exists $x'\in (\varepsilon_1 \Z)^\alpha$ such that $|x(i)-x'(i)|\leq \varepsilon_1/2$, so that $\|x-x'\|_2\leq \varepsilon_1\sqrt{\alpha}/2$).  Of course, we must have that   $\mathcal{N}$ must be an $\varepsilon$-net of $H$ for some value of $\varepsilon$. To bound this $\varepsilon$ we will use the following bound for the operator norm of the matrix $\mathcal{E}_{\T\R}$:
\begin{equation}\label{eq:Error propagation in e-nets}
\|\mathcal{E}_{\T\R}\|_{\mathrm{op}}^2\leq \alpha a  |\mathcal{E}_{\T\R}|^2 \leq \alpha a |\mathcal{E}|^2 .
\end{equation}
The first inequality in (\ref{eq:Error propagation in e-nets})  follows from Schur's bound on the maximal singular value of a real matrix. This bound implies that, if two elements $\mathpzc{x}:=(x,z)\in\mathcal{X}$ and $\mathpzc{x}':=(x',z)\in\mathcal{L}$ are $\varepsilon_1\sqrt{\alpha}/2$-close to each other, then
\begin{equation}\label{eq:Error propagation in e-nets 2}
 \|\mathcal{E}\mathpzc{x}-\mathcal{E}\mathpzc{x}'\|_2\leq \|\mathcal{E}_{\T\R}\|_{\mathrm{op}} \|x-x'\|_2 \leq \tfrac{1}{2}\alpha \sqrt{a} |\mathcal{E}|\varepsilon_1
\end{equation}
 (In the first inequality, we apply the normal form in lemma \ref{lemma:Normal form of a matrix representation}.) Finally, by   imposing $\frac{ \alpha \sqrt{a}}{2}|\mathcal{E}|\varepsilon_1\leq \varepsilon\leq \tfrac{1}{2}$, we get  that $\|\mathcal{E}(\mathpzc{x}-\mathpzc{x'})\|_G\leq \varepsilon$ due to property (\ref{eq:Relationship between norms}); it follows that $\mathcal{N}$ is an $\varepsilon$-net  if $\varepsilon_1\leq 2\varepsilon/(\alpha \sqrt{a}|\mathcal{E}|)$ for every $\varepsilon\leq \tfrac{1}{2}$.

Assuming that $\varepsilon_1$ is chosen so that $\mathcal{N}$ is an $\varepsilon$-net, our next step will be to devise an algorithm to construct and sample an $(\Delta,\varepsilon)$-net $\mathcal{N}_\Delta\subset\mathcal{N}$. The key step of our algorithm will be a subroutine that computes a nicely-behaved classical representation of the quotient group $Q=\mathcal{L}/\ker{\mathcal{E}}$ and a matrix representation  of the group isomorphism $\mathcal{E}_{\mathrm{iso}}:Q\rightarrow \mathcal{N}$  (we know that these groups are isomorphic due to the first isomorphism theorem). We will use the computed representation of $Q$ to construct  a $(\Delta,\varepsilon)$-net $Q_\Delta\subset Q$ and sample elements form it; then,  by applying the map $\mathcal{E}_{\mathrm{iso}}$ to the sampled elements, we will effectively sample a  $(\Delta,\varepsilon)$-net $\mathcal{N}_\Delta\subset \mathcal{N}$; and, moreover, in a clean \emph{collision free} fashion.

To simplify conceptually certain calculations to come, we will change some notation. Note  that $\mathcal{L}$ is isomorphic to the group $\mathcal{L}':=\Z^{\alpha+\beta}$ and that $\varepsilon_1 I_\alpha\oplus I_\beta$ is a matrix representation of this isomorphism. We will work with the group $\mathcal{L'}$ instead of  $\mathcal{L}$; accordingly, we will substitute $\mathcal{E}$ with the map $\mathcal{E'}:= \mathcal{E}(\varepsilon_1 I_\alpha\oplus I_\beta)$ and $Q$ with $Q':=\mathcal{L}'/\ker{\mathcal{E}'}$.

Our subroutine to compute a representation of $Q'$ begins by applying algorithm in theorem \ref{thm:General Solution of systems of linear equations over elementary LCA groups} to obtain a $(\alpha+\beta)\times\gamma$ matrix representation $A$ of a group homomorphism $A:\Z^{\gamma} \rightarrow \mathcal{L}'$ such that  $\mathrm{im}\,A = \ker{\mathcal{E}'}$ (where $\gamma=\alpha+\beta+m$). (Lemma \ref{lemma:Normal form of a matrix representation} ensures that real factors do not appear in the  domain of  $A$ because there are no non-trivial continuous group homomorphisms from products of $\R$ into products of $\Z$.) We can represent these maps in a diagram:
\begin{equation}
\begin{tikzcd}
\Z^{\gamma}\arrow{r}{A}
&\Z^{\alpha+\beta}\arrow{r}{\mathcal{E}'}
&\mathcal{N}
\end{tikzcd}
\end{equation}
 The worst-case time complexity needed to compute $A$ with the algorithm in theorem  \ref{thm:General Solution of systems of linear equations over elementary LCA groups} is polynomial in the variables  $m$, $\alpha$, $\beta$, $ \log{N_i}$, $\log\|\mathcal{E}\|_{\mathbf{b}}$, and $\log{\frac{1}{\varepsilon}}$. 

Next, we compute two integer unimodular matrices $U$, $V$ such that $A=USV$ and $S$ is in Smith normal form (SNF). This can be done, again, in  $\ppoly{m,\alpha,\beta,\log{N_i}, \mathrm{size}(\mathcal{E}), \log{\frac{1}{\varepsilon}}}$ time with existing  algorithms to compute the SNF of an integer matrix (see e.g.\ \cite{Storjohann10_Phd_Thesis} for a review). Each matrix $V$, $S$, $U$ is the matrix of representation of some new group homomorphism, as illustrated in the following diagram.
\begin{equation}\label{eq:Commutative diagram}
\begin{tikzcd}
\Z^{\gamma}\arrow{r}{A}\arrow{d}{V}
&\Z^{\alpha+\beta}\arrow{r}{\mathcal{E}'}
&\mathcal{N}\\
\Z^\gamma \arrow{r}{S}&\Z^{\alpha+\beta}\arrow{u}{U}
\end{tikzcd}
\end{equation}
Since $V$, $U$ are invertible integer matrices the maps $V:\Z^\alpha\rightarrow\Z^\alpha$ and $U:\Z^{\alpha+\beta}\rightarrow\Z^{\alpha+\beta}$ are continuous group isomorphisms and, hence, have trivial kernels. As a result, $\mathrm{im}\, S = \mathrm{im}\,{U^{-1}AV^{-1}} = \mathrm{im}\,U^{-1}A = U^{-1}(\mathrm{im}\,A)=U^{-1}(\ker{\mathcal{E}'})$, which shows that $\ker{\mathcal{E}'}$ is isomorphic to $\mathrm{im}\, S$ via the isomorphism $U^{-1}$. These facts together with  lemma \ref{lemma:properties of matrix representations}.(a) show that   $\mathcal{E}_{\mathrm{iso}}:=\mathcal{E}'U$ is a matrix representation of a \emph{group isomorphism} from the group  $Q':=\mathcal{L}'/\mathrm{im}\, S$ into  $\mathcal{N}$. 

Finally, we show that $Q'$ can be written explicitly as a direct product of primitive groups of type $\Z$ and $\Z_{d}$, thereby computing a finite set of generators of $Q'$ that  we can immediately use to construct  $(\Delta, \varepsilon)$-nets $\mathcal{N}_\Delta$. We make crucial use of the fact that $S$ is Smith normal form, i.e.\
\begin{equation}\label{eq:Smith normal form}
S=\left(\begin{matrix}
s_1 &     &         &     &\vline   \\
    & s_2 &         &     &\vline  \\
    &     & \ddots  &     &\vline  \\
    &     &         & s_{(\alpha+\beta)} &\vline \\
\end{matrix}\begin{array}{c}
\quad\mbox{\huge 0}\end{array}\:\right)
=\left(\begin{matrix}
I_{\mathbf{a}} &     &         &     \vline   \\
    &  
    \begin{matrix}
    \sigma_1 & &\\
             &\ddots &\\
             &       &\sigma_{\textnormal{\textbf{b}}}
    \end{matrix} &     &   \vline\\
    &   &  \mbox{\Large 0}&\vline  
\end{matrix}\begin{array}{c}
\quad\mbox{\Huge 0}
\end{array}\:\right),
\end{equation}
where the coefficients $\sigma_i$ are strictly positive. It follows readily that $\mathrm{im}\,S = \Z^\mathbf{a}\times \sigma_1\Z\times\cdots\sigma_\mathbf{b}\Z\times \{0\}^{\mathbf{c}}$, and therefore
\begin{equation}\label{inproof:Sampling Algorithm Decomposition of a Quotient}
Q'=\Z^{a+b}/\mathrm{im}\,S = \{0\}^\mathbf{a}\times \DProd{\sigma}{\mathbf{b}}\times \Z^{\mathbf{c}}.
\end{equation}
As $Q'$ and $\mathcal{N}$ are \emph{ isomorphic} the columns of the matrix $\mathcal{E}_{\mathrm{iso}}=\mathcal{E'}U$ form a generating set of $\mathcal{N}$. Moreover, since $\mathcal{E}_{\mathrm{iso}}$ acts isomorphically on (\ref{inproof:Sampling Algorithm Decomposition of a Quotient}), the subgroup $\mathcal{N}$ must be a direct sum of cyclic subgroups generated by the  columns of  $\mathcal{E}_{\mathrm{iso}}$:
\begin{equation}
\mathcal{N}=  \langle \mathpzc{f}_1 \rangle\oplus\cdots \oplus \langle  \mathpzc{f}_\mathbf{b} \rangle\oplus \langle \mathpzc{b}_1 \rangle \oplus \cdots \oplus \langle \mathpzc{b}_\mathbf{c}\rangle,
\end{equation} 
where $\mathpzc{f}_i$, $\mathpzc{b}_j$ stand for the  $(\mathbf{a}+i)$th and the $(\mathbf{b}+j)$th column of $\mathcal{E}_{\mathrm{iso}}$. Equation  (\ref{inproof:Sampling Algorithm Decomposition of a Quotient}) also tells us that the $\mathpzc{f}_i$s must generate the compact subgroup $\mathcal{N}_{\mathrm{comp}}$ and that the  $\mathpzc{b}_j$ form a $\Z$-basis of $\mathcal{N}_{\mathrm{free}}$.

The last two observations yield an efficient straightforward method to construct and sample  $(\Delta, \varepsilon)$-nets within $\mathcal{N}$. First, set  $\{\mathpzc{f}_i\}$ (resp.\  $\{\mathpzc{b}_j\}$) to be the default generating-set (resp.\ default basis) of $\mathcal{N}_{\mathrm{comp}}$ and $\mathcal{N}_{\mathrm{free}}$; then, select a parallelotope $\mathcal{P}_\Delta$ of the form (\ref{eq:Parallelotope DEFINITION}) with some desired $\Delta=(\Delta_1,\ldots,\Delta_\mathbb{c})$. This procedures specifies a net $\mathcal{N}_\Delta=\mathcal{N}_{\mathrm{comp}}\oplus\mathcal{P}_\Delta$ that can be efficiently represented with $O(\ppoly{m,\alpha,\beta,\log{N_i}, \log\|\mathcal{E}\|_{\mathbf{b}}, \log{\frac{1}{\varepsilon}} \log{\Delta_i}})$ bits of memory (by keeping track of the generating-sets of $\mathcal{N}$ and the numbers $\Delta_i$). Moreover, we can efficiently sample  $\mathcal{N}_\Delta$ uniformly and collision-freely by generating random strings of the form
\begin{equation}
\sum_{i=1}^{\mathbf{b}} \mathpzc{x}_i \mathpzc{f}_i + \sum_{j=1}^{\mathbf{c}} \mathpzc{y}_i \mathpzc{b}_j,
\end{equation}
where $\mathpzc{x}_i\in\Z_{\sigma_i}$ and $\mathpzc{y}_j\in\{0,\pm1,\ldots,\pm\Delta_j\}$.

\section{Acknowledgments}

We are grateful to Mykhaylo (Mischa) Panchenko, Uri Vool, Pawel Wocjan,  Raúl García-Patrón, Geza Giedke, Mari Carmen Bañuls, Liang Jiang, Steven M. Girvin, Barbara M.\ Terhal, Hussain Anwar and Tobias J.\ Osborne  for helpful and encouraging discussions; and to Geza Giedke, Hussain Anwar, Samuel L.\ Braunstein and Robert Raussendorf for  comments on the manuscript. JBV also thanks the MIT Center for Theoretical Physics, where part of this work was conducted, for their warm and generous hospitality in Spring 2013.

JBV acknowledges financial support from the Elite Network of Bavaria QCCC Program. CYL acknowledges support from the ARO cooperative research agreement W911NF-12-0486 (Quantum Information Group), and the Natural Sciences and Engineering Research Council of Canada. 

This paper is preprint MIT-CTP/4583.

\phantomsection

\small
\addcontentsline{toc}{section}{References}
\bibliographystyle{utphys}
\bibliography{database}
\normalsize

\appendix

\appendix

\section{Supplementary material for section \ref{sect:Homomorphisms and matrix representations}}\label{appendix:Supplement to section Homomorphisms}

\subsection*{Proof of lemma \ref{lemma:properties of matrix representations}}

First we prove (a). Note that it follows from the assumptions that $\alpha(g+h)=A(g+h)=Ag+Ah\pmod{H}$, $\beta(x+y)=B(x+y)= Bx+By \pmod{J}$, for every $g,h\in G$, $x,y\in H$. Hence, $\beta\circ\alpha(g+h)\allowbreak=\beta(Ag+Ah+\textnormal{zero}_H)=BAg+BAh+ \textnormal{zero}_J\pmod{J}$, where zero$_X$ denotes some string congruent to the neutral element $0$ of the group $X$. As in the last equation zero$_J$ vanishes modulo $J$, $BA$ is a matrix representation of $\beta\circ\alpha$.

We prove (b). From the definitions of character, bullet group and bullet map it follows that
\begin{equation}\label{inproof:Properties Matrix Reps 1}
\chi_{\mu}(\alpha(g))=\exp\left(2\pii\sum_{ij} \mu^\bullet(i)A(i,j)g(j)\right)=\exp\left(2\pii(A^\transpose \mu^\bullet)\cdot g\right) \:\textnormal{for every $g\in G$.}
\end{equation}

Let $f$ be the function $f(g):=\exp\left(2\pii(A^\transpose \mu^\bullet)\cdot g\right)$. Then  it follows from (\ref{inproof:Properties Matrix Reps 1})  that $f$ is continuous and that $f(g+h)=f(g)f(h)$, since the function $\chi_\mu\circ \alpha$ has these properties. As a result,  $f$ is a continuous character \ $f=\chi_\nu$, where $\nu\in G^*$ satisfies $\nu^\bullet=\Upsilon_G \nu = A^\transpose \mu^\bullet\pmod{G^\bullet}$. Moreover, since $f=\chi_\mu\circ\alpha = \chi_{\alpha^*(\mu)}$ it follows that  $\alpha^*(\mu)=\nu\pmod{G^*}$ and, consequently, 
\begin{equation}
 \alpha^*(\mu)= \Upsilon_G^{-1} (A^\transpose \mu^\bullet)\pmod{G^*}= \Upsilon_G^{-1} A^\transpose \Upsilon_H\mu\pmod{G^*}. 
\end{equation}
Finally, since $\chi_\mu(\alpha(g))=\chi_x (\alpha(g))$ for any $x\in\R^n$ congruent to $\mu$, we  get that $\alpha^*(\mu)=\Upsilon_G^{-1} A^\transpose \Upsilon_Hx\pmod{G^*}$ for any such $x$, which proves the second part of the lemma.\qed

\subsection*{Proof of lemma \ref{lemma:existence of matrix representations}}\label{proof:lemma existence of matrix representations}

We will show that each of the homomorphisms $\alpha_{XY}$ as considered in lemma \ref{corollary:block-structure of Group Homomorphisms} has a matrix representation, say $A_{XY}$. Then it will follow from (\ref{eq:block decomposition of a homomorphism}) in lemma \ref{corollary:block-structure of Group Homomorphisms} that 
\begin{equation}
A:= \begin{pmatrix}
      A_{\Z\Z} & 0 & 0 & 0 \\
      A_{\R\Z} & A_{\R\R} & 0 & 0 \\
      A_{F\Z} & 0 & A_{FF} & 0 \\
      A_{\T\Z} & A_{\T\R} & A_{\T F} & A_{\T\T}.
    \end{pmatrix},
\end{equation}
as in (\ref{eq:block-structure of matrix representations}), is a matrix representation of $\alpha$.

First, note that if the group $Y$ is finitely generated, then the tuples  $e_i$ form a generating set of $Y$. It is then  easy to find a matrix representation $A_{XY}$ of $\alpha_{XY}$:  just choose the $j$th column of $A_{XY}$ to be  the element $\alpha(e_j)$ of $X$. Expanding  $g=\sum_i g(i)e_i$ (where the coefficients $g(i)$ are integral), it easily follows that $A_{XY}$ satisfies the requirements for being a proper matrix representation as given in definition \ref{def:Matrix representation}. Thus, all homomorphisms $\alpha_{XY}$ with $Y$ of the types $\Z^a$ or $F$ have matrix representations; by duality and lemma \ref{lemma:properties of matrix representations}(b), all homomorphisms $\alpha_{XY}$ with $X$ of type $\T^a$ or $F$ have matrix representations too.

The only non-trivial  $\alpha_{XY}$ left to consider is $\alpha_{\R\R}$. Recall that the latter is a continuous map from $\R^m$ to $\R^n$ satisfying $\alpha_{\R\R}(x+y)= \alpha_{\R\R}(x) + \alpha_{\R\R}(y)$ for all $x, y\in \R$. We claim that every such map must be linear, i.e. in addition we have \be \label{linearity}\alpha_{\R\R}(rx)= r\alpha_{\R\R}(x)\ee for all $r\in \R$. To see this, first note that $d \alpha_{\R\R}(kx/d )=k \alpha_{\R\R}(x)$, where $k/d$ is any  fraction ($k$, $d$ are integers). Thus  (\ref{linearity}) holds for all rational numbers $r=k/d$.  Using that $\alpha_{\R\R}$ is continuous and that the rationals are dense in the reals then implies that (\ref{linearity}) holds for all $r\in \R$. This shows that $\alpha_{\R\R}$ is a linear map; the existence of a matrix representation readily follows. \qed

\subsection*{Proof of lemma \ref{lemma:Normal form of a matrix representation}}\label{proof:lemma Normal form of a matrix representation}

It suffices to show (a), that any matrix representation $A$ of $\alpha$ must be an element of $\mathrm{Rep}$ and fulfill the consistency conditions (\ref{eq:Consistency Conditions Homomorphism}); (b), that  these consistency conditions imply that $A$ is of the form (\ref{eq:block-structure of matrix representations}) and fulfills propositions 1-4; and (c), that  every such matrix defines a group homomorphism.

We will first prove (a). Let $H_i$ is of the form $\Z$ or $\Z_{d_i}$. Then, for every $j=1,\ldots,m$, the definition of matrix representation \ref{def:Matrix representation} requires that  $(Ae_j)(i)=A(i,j)\pmod{H_i}$ must be an element of $H_i$. This shows that the $i$th row of $A$ must be integral and, thus, $A$ belongs to $\mathrm{Rep}$. Moreover, since $x:=c_je_j\equiv 0\pmod{G}$ and $y:=d^*_i e_i \equiv 0 \pmod{H^*}$, (due to the definition of characteristic) it follows that $Ax=0\pmod{H}$ and $Ay=0\pmod{G^*}$, leading to the consistency conditions (\ref{eq:Consistency Conditions Homomorphism}). 

Next, we will now prove (b).

First, the block form (\ref{eq:block-structure of matrix representations}) almost follows from (\ref{eq:block decomposition of a homomorphism}) in lemma \ref{corollary:block-structure of Group Homomorphisms}: we only have to show, in addition, that the zero matrix is the only valid matrix representation for any trivial group homomorphism $\alpha_{XY}=\mathpzc{0}$ in (\ref{eq:block decomposition of a homomorphism}). It is, however, easy to check case-by-case that, if $A_{XY}$ is a matrix representation of $\alpha_{XY}$ with $A_{XY}\neq 0$, then $\alpha_{XY}$ cannot be trivial.

Second, we prove propositions 1-4. In proposition 1, $A_{\Z\Z}$ must be integral since  $A_{\Z\Z}e_j(i)\in\Z$  (where, with abuse of notation, $i,\,j$ index the rows and columns of $A_{\Z\Z}$). By duality the same holds for $A_{\T\T}$ (it can be shown using  lemma \ref{lemma:properties of matrix representations}(b)). In proposition 2 the consistency conditions (including dual ones) are vacuously fulfilled and tell us nothing about  $A_{\R\Z}$, $A_{\R\R}$. In proposition 3, both matrices have to be integral to fulfill that $A_{XY}(e_i)\pmod{Y}$ is an element of $X$, which is of type $\Z^a$ or $F$; moreover, for $Y=F$, the  consistency conditions directly impose that the coefficients must be of the form (\ref{eq:coefficients of Matrix Rep for nonzero characteristic groups}). (The derivation is similar to that of  lemma 11 in \cite{BermejoVega_12_GKTheorem}. Lastly, in proposition 4, all consistency conditions associated to $A_{\T \Z}$ and $A_{\T\R}$ are, again, vacuous and tell us nothing about the matrix; however, the first consistency condition tells us that $A_{\T F}$ must  be fractional with coefficients of the form $\alpha_{i,j}/c_j$.

Finally we will show (c). First, it is manifest that if $A$ fulfills 1-4 then $A\in \mathrm{Rep}$. Second, to show that $A$ is a matrix representation of a group homomorphism it is enough to prove that every $A_{XY}$ fulfilling 1-4 is the matrix representation of a group homomorphism from $Y$ to $X$.  This can be checked straightforwardly for the cases  $A_{\Z\Z}$,  $A_{\R\Z}$, $A_{\R\R}$,  $A_{F\Z}$, $A_{\T\Z}$,  $A_{\T\R}$ applying properties 1-4 of $A$ and using that, in all cases, there are no non-zero real vectors congruent to the zero element of $Y$. Obviously, for the cases where $A_{XY}$ must be zero the proof is trivial. It remains to consider the cases $A_{FF}$, $A_{\T F}$, $A_{\T\T}$. In all of this cases, it holds due to properties 1,3,4 that the  first consistency condition in (\ref{eq:Consistency Conditions Homomorphism}) is fulfilled. Then, for the case $A_{FF}$, lemma 2 in \cite{BermejoVega_12_GKTheorem} can be applied to show that $A_{FF}$ is a group homomorphism; moreover, for any $g\in F$ and $x$ congruent to $g$, the first consistency condition implies that $Ax = Ag \pmod{G}$, which is what we wanted to show. For the cases $A_{\T F}$, $A_{\T \T}$, lemma 2 in \cite{BermejoVega_12_GKTheorem}  can still be generalized (the proof given in \cite{BermejoVega_12_GKTheorem} for the second statement of lemma 2 can be directly adapted) and the same argument applies. \qed

\section{Existence of general-solutions of systems of the form (\ref{eq:Systems of linear equations over groups})}\label{appendix:Closed subgroups of LCA groups / existence of general-solutions of linear systems}

In this section we  show that general-solutions of systems of linear equations over elementary Abelian groups always exist (given that the systems admit at least one solution). 

We start by recalling an important property of elementary  Abelian groups.
\begin{lemma}[{See theorem 21.19 in \cite{Stroppel06_Locally_Compact_Groups} or section 7.3.3 in \cite{Dikranjan11_IntroTopologicalGroups}}]\label{lemma:Elementary_LCA_group property extends to subgroups, quotients, products}
The class of elementary Abelian groups\footnote{Beware that in  \cite{Stroppel06_Locally_Compact_Groups} the class of elementary groups is referred as ``the category CGAL'', which stands for Compactly Generated Abelian Lie groups.} is closed with respect to forming closed subgroups, quotients by these, and finite products.\footnote{In fact, as mentioned in \cite{Stroppel06_Locally_Compact_Groups}, corollary 21.20 elementary LCA groups constitute the smallest subclass of LCA containing $\R$ and fulfilling all these properties.}
\end{lemma}
In our setting, the kernel of a continuous group homomorphism $A:G\rightarrow H$ as in (\ref{eq:Systems of linear equations over groups}) is always closed: this follows from the fact that the singleton $\{0\}\subset H$ is closed (because elementary Abelian groups are Hausdorff  \cite{Stroppel06_Locally_Compact_Groups}), which implies that $\ker{A}=A^{-1}(\{0\})$ is closed (due to continuity of $A$). Hence, it follows from lemma \ref{lemma:Elementary_LCA_group property extends to subgroups, quotients, products} that $\ker{A}$ is topologically isomorphic to some elementary Abelian group  $H':=\R^a\times\T^b \times \Z^c \times \DProd{N}{c}$; consequently, there exists a continuous group isomorphism  $\varphi$ from $H'$ to $H$.

Next, we write the group $H'$ as a quotient group $X/K$ of the group $X:=\R^{a+b}\times \Z^{c+d}$ by the subgroup $K$ generated by the elements of the form $\mathrm{char}(X_i)e_i$. The quotient group $X/K$ is the image of the quotient map $q:X\rightarrow X/K$ and the latter is a continuous group homomorphism \cite{Stroppel06_Locally_Compact_Groups}. By composing  $\varphi$ and $q$ we obtain a continuous group homomorphism $\mathcal{E}:=\varphi\circ  q$ from $X$ onto $H$. The map $\mathcal{E}$ together with any particular solution $x_0$ of (\ref{eq:Systems of linear equations over groups}) constitutes a general solution of (\ref{eq:Systems of linear equations over groups}), proving the statement.

\section{Proof of theorem \ref{thm:General Solution of systems of linear equations over elementary LCA groups}}\label{appendix:Systems of Linear Equations over Groups}

In this section we show that systems of linear equations over groups (\ref{eq:Systems of linear equations over groups}) can be reduced to systems of mixed real-integer linear equations (\ref{eq:System of Mixed Integer linear equations}).

Start with two elementary groups of general form $G$, $H$. First, notice that we can write $G$ and $H$ as $G=G_1\times\cdots\times G_m$, $H=H_1\times \cdots \times H_n$ where each factor $G_i$, $H_j$ is of the form $G_i=\textbf{X}_i/c_i\Z$, $H_j=\textbf{Y}_i/d_i\Z$ with $\textbf{X}_i$, $\textbf{Y}_i\in \{\Z, \R\}$; the numbers $c_i$, $d_j$ are the characteristics of the primitive factors. We assume w.l.o.g.\ that the primitive factors of $G$, $H$ are \emph{ordered} such that both groups are of the form $\Z^a\times F\times \T^b$: in other words, the finitely generated factors come first.

We now define a new group $\textbf{X}:=\textbf{X}_1\times\cdots\times \textbf{X}_m$ (recall that with the ordering adopted $X$ is of the form $\Z^a\times \R^b$) which will play the role of an \emph{enlarged solution space}, in the following sense. Let $\textbf{V}$ be the subgroup of $\mbf{X}$ generated by the elements $c_1e_1,\ldots, c_me_m$. Observe that the group $G$---the solution space in system (\ref{eq:Systems of linear equations over groups})---is precisely the quotient group $\mbf{X}/\mbf{V}$, and thus can be \emph{embedded} inside the larger group $\textbf{X}$ via the quotient group homomorphism $\textbf{q}:\textbf{X}\rightarrow G= \mbf{X}/\mbf{V}$:
\begin{equation}
\textbf{q}(\textbf{x}):=(\textbf{x}(1)\bmod{c_1}, \ldots, \textbf{x}(m)\bmod{c_m})=\textbf{x}\pmod{G};
\end{equation}
remember also that $\ker{\textbf{q}}=\textbf{V}$. Now let $\alpha:\textbf{X}\rightarrow H$ be the group homomorphism defined as $\alpha:=A\circ \textbf{q}$. Then it follows from the definition that $\alpha(\mathbf{x})=A\mathbf{x}\pmod{H}$, and $A$ is a  matrix representation of $\alpha$. (This is also a consequence of the composition property of matrix representations (lemma \ref{lemma:properties of matrix representations}.(a), since the $m\times m$ identity matrix $I_m$ is a matrix representation of $\textbf{q}$.) We now consider the relaxed\footnote{Notice that the new system is less constrained, as we look for solutions in a larger space than beforehand.} system of equations
\begin{equation}\label{eq:Enlarged system of equations 1}
\alpha(\textbf{x})=b\pmod{H},\quad  \textnormal{where } \textbf{x}\in \textbf{X}=\Z^a\times \R^b.
\end{equation}
Note that the problem of solving (\ref{eq:Systems of linear equations over groups}) reduces to solving (\ref{eq:Enlarged system of equations 1}), which  looks closer to a system of mixed real-integer linear equations. Indeed, let   $\textbf{X}_{\textnormal{sol}}$ denote the set  of all solutions  of system  (\ref{eq:Enlarged system of equations 1}); then\footnote{It is easy to prove $G_{\textnormal  {sol}}=\mathbf{q}(\mathbf{X}_{\textnormal{sol}})$ by showing $G_{\textnormal  {sol}}\supset \mathbf{q}(\mathbf{X}_{\textnormal{sol}})$ and the reversed containment for the preimage $\textbf{q}^{-1}(G_{\textnormal  {sol}})\subset\mathbf{X}_{\textnormal{sol}}$; then surjectivity of $\textbf{q}$ implies  $G_{\textnormal{sol}}=\textbf{q}(\textbf{q}^{-1}(G_{\textnormal{sol}}))\subset \textbf{q}(\textbf{X}_\textnormal{sol})$.}
\begin{equation}\label{eq:Relationship between solutions of original and enlarged system}
G_{\textnormal  {sol}}=\mathbf{q}(\mathbf{X}_{\textnormal{sol}})\quad\Longrightarrow\quad  G_{\textnormal  {sol}}=\textbf{q}(\textbf{x}_0)+\textbf{q}(\ker{\alpha})= \textbf{x}_0+\ker{\alpha}\pmod{G},
\end{equation}
Hence, our original system (\ref{eq:Systems of linear equations over groups}) admits solutions iff (\ref{eq:Enlarged system of equations 1}) also does, and the former can be obtained from the latter via the homomorphism $\mathbf{q}$. We further show next that (\ref{eq:Enlarged system of equations 1}) is equivalent to a system of  form (\ref{eq:System of Mixed Integer linear equations}). First, note that the matrix $A$ has a block form $A=\begin{pmatrix}
A_\Z & A_\R
\end{pmatrix}$ where $A_\Z$, $A_\R$ act, respectively, in integer and real variables. Since the constraint (mod ${H}$) is equivalent to the modular constraints mod ${ d_1},\ldots,$ mod ${ d_n}$, it follows that $\textbf{x}=\begin{pmatrix}
\mathbf{x_\Z} & \mathbf{x_\R}\end{pmatrix}\in\textbf{X}_{\textnormal{sol}}$ if and only if
\begin{equation}\label{eq:Enlarged system of equations 2}
A_\Z \mathbf{x_\Z} + A_\R\mathbf{x_\R} + D\mathbf{y}= c, \quad \textnormal{where } D=\textnormal{diag}(d_1,\ldots,d_n),\: \mathbf{y}\in \Z^n.
\end{equation}
Clearly, if we rename $A':=\begin{pmatrix}
A_\Z & D
\end{pmatrix}$, $\mathbf{x}':=\begin{pmatrix}
\mathbf{x_\Z} &  \mathbf{y}
\end{pmatrix}$, $B=A_\R$ and $\mathbf{y}':=\mathbf{x_\R}$, system (\ref{eq:Enlarged system of equations 2}) is a system of mixed-integer linear equations as in (\ref{eq:System of Mixed Integer linear equations}). Also,  system (\ref{eq:System of Mixed Integer linear equations}) can be seen as a system  of linear equations over Abelian groups: note that in the last step the solution space $\mathbf{X}$ is increased by introducing new extra integer variables $\mathbf{y}\in\Z^n$. If we let $\mathbf{G}$ denote the group $\mbf{X}\times \Z^n$ that describes this new space of solutions, then  (\ref{eq:Enlarged system of equations 2}) can be rewritten as
\begin{equation}\label{eq:Enlarged system of equations 3}
\mathbf{A}\mathbf{g}:=
\begin{pmatrix}
A & D 
\end{pmatrix}\mathbf{g}=c,\quad \textnormal{where } \mathbf{g} \in \mathbf{G}
\end{equation}
and $c$ represents an element of $\mathbf{Y}$. 

Mind that (\ref{eq:Enlarged system of equations 2}) (or equivalently (\ref{eq:Enlarged system of equations 3})) admits solutions if and only if both of (\ref{eq:Enlarged system of equations 1}) and (\ref{eq:Systems of linear equations over groups}) admit solutions. Indeed, the solutions of (\ref{eq:Enlarged system of equations 2}) and (\ref{eq:Enlarged system of equations 1}) are---again---related via a surjective group homomorphism  $\pi:\mathbf{X}\times\Z^n\rightarrow \mbf{X}:(\mathbf{x}, \mathbf{y})\rightarrow \mathbf{x}$. It follows from the derivation of (\ref{eq:Enlarged system of equations 2}) that $\pi(\mathbf{G}_{\textnormal{sol}})=\mathbf{X}_{\textnormal{sol}}$ and, consequently, $\mathbf{q}\circ \pi(\mathbf{G}_{\textnormal{sol}})=G_{\textnormal{sol}}$; these relationships show that either all systems admit solutions or none of them do.\\

Finally, we use existing algorithms to find a general solution $(\mathbf{g}_0, \mathbf{P})$ of system (\ref{eq:Enlarged system of equations 3}) and show how to use this information to compute a general solution of our original problem (\ref{eq:Systems of linear equations over groups}).

First,  we recall that  algorithms presented in \cite{BowmanBurget74_systems-Mixed-Integer_Linear_equations} can be used to: (a) \emph{check} whether a system of the form (\ref{eq:System of Mixed Integer linear equations},\ref{eq:Enlarged system of equations 3}) admits a solution\footnote{Mind that this step is actually not essential for our purposes, since in the applications we are interested on all such systems admits solutions by promise.}; (b)  \emph{find} a particular solution $\mathbf{g}_0$ (if there is any) and a matrix $\mathbf{P}$ that defines a group endomorphism of $\textbf{G}=\mbf{X}\times \Z^n$ whose image $\textnormal{im}\,\mathbf{P}$ is precisely the kernel\footnote{In fact, the matrix $P$ is also idempotent and defines a projection map  on $\textbf{G}$ and $\ker{\mathbf{A}}$ is the image of a projection map: subgroups satisfying this property are called \emph{retracts}. Though the authors never mention the fact that $\textbf{P}$ is a projection, this follows immediately from their equations (10a,10b).} of  $\mathbf{A}=
\begin{pmatrix}
A & D 
\end{pmatrix}$ (for details see theorem 1 in \cite{BowmanBurget74_systems-Mixed-Integer_Linear_equations}). 

Assume now that (\ref{eq:Enlarged system of equations 2}) admits solutions and that we have already found a general solution $(\mbf{g}_0=(\mathbf{x}_0,\mathbf{y}_0), \mathbf{P})$. We show next how a general solution $(x_0,P)$ of (\ref{eq:Systems of linear equations over groups}) can be computed by making use of the map $\mathbf{q}\circ \pi$. We also discuss the overall worst-time running time we need to compute $(x_0,P)$, as a function of the sizes of the matrix $A$ and the tuple $b$ given as an input in our original problem (\ref{eq:Systems of linear equations over groups}) (the bit-size  or simply \emph{size} of an array of real numbers---tuple, vector or matrix---is defined as the minimum number of bits needed to store it with infinite precision), size($G$) and size($H$):
\begin{itemize}
\item First, note that $(\mbf{g}_0=(\mathbf{x}_0,\mathbf{y}_0), \mathbf{P})$ can be computed in polynomial-time in  $\textnormal{size}(A)$, $\textnormal{size}(b)$, size($G$) and size($H$), since there is only a polynomial number of additional variables and constrains in (\ref{eq:Enlarged system of equations 2}) and the worst-time scaling of the algorithms in \cite{BowmanBurget74_systems-Mixed-Integer_Linear_equations} is also polynomial in the mentioned variables. (We discussed the complexity of these methods in section\ref{sect:Systems of linear equations over groups}.)
\item Second, a particular solution $x_0$ of (\ref{eq:Systems of linear equations over groups}) can be easily computed just by taking $x_0:=\mathbf{q}\circ\pi((\mathbf{x}_0,\mathbf{y}_0))=\pi(\mathbf{x}_0)\pmod{G}$: this computation is clearly efficient in size($\textbf{x}_0$) and size($G$). 
\item Third, note that the composed map $P:=\mathbf{q}\circ\pi\circ \mathbf{P}$ defines a group homomorphism  $P:\mathbf{G}\rightarrow G$ whose image is precisely the subgroup $\ker A$; a matrix representation of $P$ (that we denote with the same symbol) can be efficiently computed, since
\begin{equation}
\textnormal{if}\quad \mathbf{P}=\begin{pmatrix}
\mathbf{P}_{\mathbf{X} \mathbf{X}} & \mathbf{P}_{\mathbf{X} \Z}\\
\mathbf{P}_{\Z \mathbf{X}} & \mathbf{P}_{\Z \Z}
\end{pmatrix}
\quad
\textnormal{then} \quad P:=\begin{pmatrix}
\mathbf{P}_{\mathbf{X} \mathbf{X}} & \mathbf{P}_{\mathbf{X} \Z}
\end{pmatrix}
\end{equation}
is a matrix representation of $\mathbf{q}\circ\pi\circ \mathbf{P}$ that we can take without further effort. 
\end{itemize}
The combination of all steps above yields a deterministic polynomial-time algorithm to  compute a general solution $(x_0, P;\: \mathbf{G})$ of system (\ref{eq:Systems of linear equations over groups}), with  worst-time scaling as a polynomial in the variables $m$, $n$, $\log\|A\|_{\mathbf{b}}$, $\log\|b\|_{\mathbf{b}}$, $\log{c_i}$, $\log{d_j}$. This proves theorem \ref{thm:General Solution of systems of linear equations over elementary LCA groups}.

\section{Efficiency of Bowman-Burdet's algorithm}\label{appendix:Efficiency of Bowman Burdet}

In this appendix we briefly discuss the time performance of Bowman-Burdet's algorithm \cite{BowmanBurget74_systems-Mixed-Integer_Linear_equations} and argue that, using current algorithms to compute certain matrix normal forms (namely, Smith normal forms) as subroutines), their algorithm can be implemented in worst-time polynomial time.

An instance of the problem $Ax+By = C$, of the form (\ref{eq:System of Mixed Integer linear equations}), is specified by the rational matrices $A$, $B$ and the rational vector $C$. Let $A$, $B$, $C$ have $c\times a$, $c\times b$ and $c$ entries. Bowman-Burdet's algorithm (explained in \cite{BowmanBurget74_systems-Mixed-Integer_Linear_equations}, section 3) involves different types of steps, of which the most time consuming are (see equations 8-10 in \cite{BowmanBurget74_systems-Mixed-Integer_Linear_equations}):
\begin{enumerate}
\item the calculation of a constant number of certain types of generalized inverses introduced by Hurt and Waid \cite{HurtWaid70_Integral_Generalized_Inverse};
\item a constant number of matrix multiplications.
\end{enumerate}
A Hurt-Waid generalized inverse $M^\#$ of a rational matrix $M$ can be computed with an algorithm given in \cite{HurtWaid70_Integral_Generalized_Inverse}, equations 2.3-2.4. The  worst-running time of this procedure is dominated by the computation of a Smith Normal form  $S=UMV$ of $M$ with pre- and post- multipliers $U$, $V$. This subroutine becomes the bottleneck of the entire algorithm, since existing algorithms for this problem are slightly slower than those for multiplying matrices (cf.\ \cite{Storjohann10_Phd_Thesis} for a slightly outdated review). Furthermore, $S$, $U$ and $V$ can be computed in polynomial time (we refer the reader to \cite{Storjohann10_Phd_Thesis} again). 

The analysis above shows that Bowman-Burdet's algorithm runs in worst-time polynomial in the variables $\log\|A\|_{\mathbf{b}}$, $\log\|B\|_{\mathbf{b}}$, $\log\|C\|_{\mathbf{b}}$, $a$, $b$, $c$, which is enough for our purposes.

\section{Proof of lemma \ref{lemma:Computing Inverses}}\label{appendix:Computing inverses}

As a preliminary, recall  that group homomorphism form an Abelian group with the point-wise addition operation. Clearly, matrix representations inherit this group operation and form a group too. This follows from the following formula,
\begin{equation}\label{eq:Homomorphisms form a group}
(\alpha+\beta)(g)=\alpha(g)+\beta(g)=Ag + Bg = (A+B)g\pmod{G},
\end{equation}
which also states that the sum $(A+B)$ of the  matrix representations $A$, $B$  of two homomorphisms $\alpha$, $\beta$ is a matrix representation of the homomorphism $\alpha+\beta$. The group structure of the matrices is, in turn, inherited by their  \emph{columns}, a fact that will be exploited in the rest of the proof; we will denote by  $X_j$ the Abelian group formed by the $j$th columns of all matrix representations with addition rule inherited from the matrix addition operation. 

A consequence of lemma \ref{lemma:Normal form of a matrix representation} is that the group $X_j$ is always an elementary Abelian group, namely
\begin{align}
G_j&=\Z  &\Rightarrow\quad\qquad & X_j =  G= \Z^{a}\times \R^{b} \times \DProd{N}{c} \times \T^{d};\notag\\
G_j&=\R   &\Rightarrow\qquad\quad & X_j = \{0\}^{a}\times \R^{b} \times \{0\}^{c} \times \R^{d};\notag\\
G_j&=\Z_{N_j}  &\Rightarrow\qquad\quad & X_j =  \{0\}^{a}\times \{0\}^{b} \times \left(\eta_{1,j} \Z \times\cdots \times \eta_{c,j}\Z\right)\times \left(\tfrac{1}{N_j}\Z\right)^{d};\notag\\
G_j&=\T  &\Rightarrow\qquad \quad& X_j = \{0\}^{m_z}\times \{0\}^{m_r} \times \{0\}^{m_f} \times  \Z^{m_t};\label{eq:Columns of M representations form a group}
\end{align}
where   $\eta_{i,j}:=N_i/\gcd{(N_i, N_j)}$.

We will now prove the statement of the lemma.

First, we reduce the problem of computing a valid matrix representation $X$ of $\alpha^{-1}$ to that of solving   the equation $\alpha\circ\beta=\mathrm{id}$ ($\alpha$) is now the given automorphism) where $\beta$ stands for any continuous group homomorphism $\beta:G\rightarrow G$. It is easy to show that this equation admits $\beta=\alpha^{-1}$ as unique solution, since
\begin{equation}
\alpha\circ\beta=\mathrm{id}\quad\Longrightarrow\quad\beta=(\alpha^{-1}\circ\alpha)\circ\beta = \alpha^{-1}\circ(\alpha\circ\beta) = \alpha^{-1}.
\end{equation}
Hence, our task will be to find a matrix $X$ such that $g\rightarrow Xg\pmod{G}$ is a continuous group \emph{homomorphism} and such that $AX$ is a matrix representation of the identity automorphism. The latter condition reads  $AX g = g\pmod{G}$ for every $g\in G$ and is equivalent to
\begin{align}\label{inproof:Computing inverses 0}
AX\left( \sum_j g(j)e_j\right) = \sum_j g(j) Ax_j= \sum_j g(j)e_j \pmod{G} ,\,  \textnormal{ for every $g \in G,$}
\end{align}
where $x_j$ denotes the $j$th column of $X$. Since (\ref{inproof:Computing inverses 0}) holds, in particular, when all but  one number $g(j)$ are  zero, it can be re-expressed as an equivalent system of equations:
\begin{equation}\label{inproof:Computing Inverses 1}
  g(j)Ax_j=  g(j)e_j \pmod{G} ,\,  \textnormal{ for any $g(j) \in G_j$, for $j=1,\ldots,m$}.
\end{equation}
Finally, we will reduce each individual equation in (\ref{inproof:Computing Inverses 1}) to  a linear system of equations of the form (\ref{eq:Systems of linear equations over groups}). This will let us apply the algorithm in theorem \ref{thm:General Solution of systems of linear equations over elementary LCA groups} to compute every individual column $x_j$ of $X$. 

We begin by finding some simpler equivalent form for (\ref{inproof:Computing Inverses 1}) for the different types of primitive factors:
\begin{itemize}
\item[(a)] If $G_j=\Z$ or  $G_j=\Z_{N_j}$ the coefficient $g(j)$ is integral and can take the value 1. Hence, equation (\ref{inproof:Computing Inverses 1}) holds iff $Ax_j=  e_j \pmod{G}$.
\item[(b)] If $G_j=\R$ or  $G_j=\T$ we show that (\ref{inproof:Computing Inverses 1}) is equivalent to $Ax_j = e_j\pmod{X_j}$. Clearly, (\ref{inproof:Computing Inverses 1}) implies $g(j)Ax_j = g(j)e_j + \mathit{zero}$ where $\mathit{zero}= 0\pmod{G}$ and where  we fix a value of $g(j)\in G_j$. Since $G_j$ is divisible,  $g(j)'=g(j)/d$ is also an element of $G_j$ for any positive integer $d$. For this value we get $\tfrac{g(j)}{d} Ax_j   = \tfrac{g(j)}{d}e_j  + \mathit{zero}'$. These two equations combined show that  $\mathit{zero}= d\,\mathit{zero}'$ must hold for every positive integer $d\in \Z$. Since both $\mathit{zero}$ and $\mathit{zero}'$ are integral, it follows that the entries of $\mathit{zero}$ are \emph{divisible} by all positive integers; this can only happen if $\mathit{zero}=0$ and, consequently, (\ref{inproof:Computing Inverses 1}) is equivalent to $Ax_j = e_j$. Since both  $Ax_j$ and $e_j$  are $j$th columns of matrix representations, the latter equation can be  written as $Ax_j = e_j\pmod{X_j}$ with $X_j$  as in (\ref{eq:Columns of M representations form a group}).  
\end{itemize}
Finally, we argue that the final systems (a) $Ax_j=  e_j \pmod{G}$ and (b) $Ax_j=  e_j \pmod{X_j}$ are linear systems of the form (\ref{eq:Systems of linear equations over groups}).  First notice that for any two homomorphisms $\beta$, $\beta'$ with matrix representations $X$, $Y$, it follows from (\ref{eq:Homomorphisms form a group}) and lemma \ref{lemma:properties of matrix representations}.(a)  that $A(X+Y)=AX+AY$ is a matrix representation of the homomorphism $\alpha\circ(\beta+\beta')=\alpha\circ\beta+\alpha\circ\beta'$. Consequently,
\begin{equation}\label{inproof:Computing Inverses 3}
A(X+Y)g = (AX + AY) g\pmod{G}, \textnormal{for every }g\in G.
\end{equation}
The argument we used to reduce $AXg=g\pmod{G}$ to the cases (a) and (b) can be applied again to find a simpler form for (\ref{inproof:Computing Inverses 3}). Applying the same procedure step-by-step (the derivation is omitted), we obtain that, if $G_j=\Z$ or  $G_j=\Z_{N_j}$, then  (\ref{inproof:Computing Inverses 3}) is equivalent to $A(x_j+y_j)=Ax_j + Ay_j \pmod{G}$; if $G_j=\R$ and  $G_j=\T$, we get $A(x_j+y_j)=Ax_j + Ay_j \pmod{X_j}$ instead. It follows that the map $x_j\rightarrow Ax_j$ is a group homomorphism from $X_j$ to $G$ in case (a) and from $X_j$ to $X_j$ in case (b). This shows that systems  (a) and  (b) are  of the form (\ref{eq:Systems of linear equations over groups}).

\section{Supplementary material for section \ref{sect:quadratic_functions}}\label{app:Supplement Quadratic Functions}

\subsection*{Proof of lemma \ref{lemma:Normal form of a bicharacter 1}}

The lemma is a particular case of proposition 1.1 in \cite{Kleppner65Multipliers_on_Abelian_Groups}. We reproduce a shortened proof of the result in \cite{Kleppner65Multipliers_on_Abelian_Groups} (modified to suit our notation) here.

If $\beta$ is an continuous homomorphism from $G$ into $G^*$ then $ B(g, h) = \chi_{\beta(g)}(h)$ is continuous, since composition preserves continuity. Also, it follows using the linearity of this map and  of the character functions that  $B(g, h)$ is bilinear, and hence a bicharacter.  Conversely, consider an arbitrary bicharacter $B$. The condition that $B$ is a character on  the second argument says that for every $g$ the function $f_g: h\to B(g, h)$ is a character. Consequently $f_g(h)=B(g,h)=\chi_{\mu_g}(h)$ for all $h\in G$ and some $\mu_g\in G^*$ that is determined by $g$. We denote by $\beta$ be the map which sends $g$ to $\mu_g$. Using that $g\to B(g, h)$ is also a character  it follows that $\chi_{\beta{(g+g')}}(h)=\chi_{\beta{(g)}}(h)\chi_{\beta{(g')}}(h)$ for all $h\in G$, so that $\beta:G\rightarrow G^*$ is a group homomorphism. It remains to show that $\beta$ is continuous; for this we refer to the proof in \cite{Kleppner65Multipliers_on_Abelian_Groups}, where the author analyzes how  neighborhoods are transformed under this map.

\subsection*{Proof of lemma \ref{lemma symmetric matrix representation of the bicharacter homomorphism}}

We obtain (a) by combining (\ref{eq:definition of bullet map2}) with the normal form (\ref{eq:first normal fomr of a bicharacter}): the matrix $M$ is  of the form $\Upsilon X$ where $X$ is a matrix representation of $\beta$;  (b) follows from this construction. (c) follows from the normal form in lemma \ref{lemma:Normal form of a bicharacter 1}, property (a) and  lemma \ref{thm_extended_characters}.

To prove (d) we bring together (a) and the relationship $B(h,g)=B(g,h)$, and derive
\begin{equation}\label{eq:symmetry of the starred-matrix representation of a bicharacter MODULO Z}
 g^\transpose M h =  g^\transpose M^\transpose  h\mod{\Z},\quad\text{for every $g$, $h\in G$.}
\end{equation}
Write $G=G_1\times\dots\times G_m$ with $G_i$ of primitive type. If $G_i$ is either finite or equal to $\Z$ or $\R$ then the canonical basis vector $e_i$ belongs to $G$. If $G_i=\T$ then $te_i\in G$ for all $t\in[0, 1)$. If neither $G_i$ nor $G_j$ is equal to $\T$, taking  $g= e_i$, $h=e_j$ in equation (\ref{eq:symmetry of the starred-matrix representation of a bicharacter MODULO Z}) yields $M(i, j)\equiv M(j, i)$ mod $\Z$. If $G_i$ and $G_j$ are equal to $\T$, setting $g= te_i$ and $h=se_j$ yields $stM(i, j)\equiv stM(j, i)$ mod $\Z$ for all $s, t\in [0, 1)$, which implies that \be st (M(i, j)- M(j, i))\in \Z  \ee for all $s, t\in [0, 1)$. This can only happen if $M(i, j)= M(j, i)$. The other cases are treated similarly. In conclusion, we find that $M$ is symmetric modulo $\Z$. This proves (d). 

Lastly, we prove (e). Note that we have just shown that $M(i,j)=M(j,i)$ if  $G_i=G_j=\T$; the same argument can be repeated (with minor modifications) to show $M(i,j)=M(j,i)$ if either one of $G_i$ or $G_j$ is  of the form $\R$ or $\T$. 
Hence, $M(i,j)\neq M(j,i)$ can only happen if $G_i$, $G_j$ are of the form   $\Z$ or $\Z_d$. In this case, we denote by $\Delta_{ij}$ the number such that $M(j,i)=M(i,j)+\Delta_{i,j}$. (d) tells us that $\Delta_{ij}$ is an integer. Moreover, by choosing  $g=g(i)e_i$, $h=h(j)e_j$ in (\ref{eq:symmetry of the starred-matrix representation of a bicharacter MODULO Z}) it follows that
\begin{equation}\label{inproof: Delta ij}
M(j,i)g(i)h(j)=M(i,j)g(i)h(j) + \Delta_{i,j}g(i)h(j) \mod{\Z},
\end{equation}
As $g(i)$ and $h(j)$ are integers the factor $\Delta_{i,j}g(i)h(j)$  gets cancelled modulo $\Z$ and produces no effect. Finally, we  define  a new symmetric matrix  $M'$ as $M'(i,j)=M(i,j)$ if $i\geq j$, and $M'(i,j)=M(j,i)$ if $i< j$. It follows from our discussion that  $g^\transpose M'h=g^\transpose Mh\bmod{\Z}$ for every $g,h\in G$, so that $M'$ manifestly fulfills (a).

It remains to show that $M'$ fulfills (b)-(c). Keep in mind that $h\rightarrow Mh\pmod{G^\bullet}$ defines a group homomorphism into $G^\bullet$. From our last equations, it follows that either $M(i,j)h(j)=M'(i,j)h(j)$ or $M(i,j))h(j)=M'(i,j)h(j)\bmod{\Z}$ if both $G_i$ and $G_j$ are discrete groups. From the definition of bullet group (\ref{eq:Bullet Group}), it is now easy to derive that $Mh=M'h\pmod{G^\bullet}$ for every $h$, and to extend this equation to all tuples $x$ congruent to $h$  (this reduces to analyzing all possible combinations of primitive factors). As a result, $M'$ is a matrix representation that defines the same map as $M$, which implies (b). The fact that $M'$ satisfies (c) follows using the same argument we used for $M$.

\subsection*{Proof of lemma \ref{lemma quadratic B-representations differ by a character}}

We prove that the function $f(g):=\xi_1(g)/\xi_2(g)$ is a character, implying that there exists $\mu\in G^*$ such that $\chi_{\mu}=f$:
\begin{equation}\label{eq proof of quadratic B-representations differ by character}
f(g+h):=\frac{\xi_1(g)}{\xi_2(g)}\frac{\xi_1(h)}{\xi_2(h)}\frac{B(g,h)}{B(g,h)}=f(g)f(h).
\end{equation}

\subsection*{Proof of lemma \ref{lemma:B-representations can always be constructed}}

Define the function $q:G\rightarrow \R$ as
\begin{equation}
q(g):= g^{\transpose} M g + C^\transpose g.
\end{equation}
We prove that $q(g)$ is a quadratic form modulo $2\Z$ with associated bilinear form $b_q(g,h):=2 g^\transpose M h$; or, in other words, that the following equality holds for every $g,h\in G$:
\begin{equation}\label{eqinproof:quadratic form over 2Z}
q(g+h)=q(g)+q(h)+ 2 g^\transpose M h \pmod{2\Z}.
\end{equation}
Assuming that (\ref{eqinproof:quadratic form over 2Z}) is correct, it follows readily that the function $Q(g)=\exp{\left(\pii q(g)\right)}$ is quadratic and also a $B$-representation, which is what we wanted:
\begin{equation}
Q(g+h)=Q(g)Q(h)\exp{\left(2\pii\, g^\transpose M h\right)}
\end{equation}
We prove  (\ref{eqinproof:quadratic form over 2Z}) by direct evaluation of the statement. First we define $q_M(g):=g^{\transpose} M g$ and $q_C(g):=C^{\transpose} g$, so that $q(g)=q_M(g)+q_C(g)$. We will also (temporarily, i.e. only within the scope of this proof) use the notation $g\oplus h$ to denote the group operation in $G$ and reserve $g+h$ for the case when we sum over the reals. Also, denoting $G=G_1\times\dots G_m$ with $G_i$ primitive, we define $c:=(c_1,\ldots,c_m)$ to be a tuple containing all the characteristics $c:=\charac{G_i}$. With these conventions we  have $g\oplus h = g+h + \lambda\circ c$, where $\lambda$ is a vector of integers and $\circ$ denotes the entrywise product: $\lambda\circ c=(\lambda_1 c_1, \dots, \lambda_m c_m)$.  Note that $\lambda\circ c$  is the most general form of any string of real numbers that is congruent to $0\in G$ (the neutral element of the group).
We then have (using that $M=M^T$):
\begin{align}\label{eqinproof quadratic form derivation 1}
q_M(g\oplus h) = \: & q_M (g) + q_M (h) + 2 g^\transpose M h  \nonumber\\ &+ 2g^\transpose M (\lambda\circ c) + 2h^\transpose M (\lambda\circ c) + (\lambda\circ c)^\transpose M (\lambda\circ c),\\
q_C(g\oplus h) =\: & q_C(g) + q_C(h) +  \sum_i M(i,i) \lambda(i)c_i^2.
\end{align}
Consider an $x\in \R^m$ for which there exists  $g\in G$ such that $x\equiv g$ mod $G$. Then $x^\transpose M (\lambda\circ c)$ with $x\in G $ must be an integer. Indeed, we have \be \label{inproof:B(g,0)} 1 = B(g,0)=\exp{\left(2\pii x^\transpose M (\lambda\circ c) \right)},\ee where in the second identity we used lemma \ref{lemma symmetric matrix representation of the bicharacter homomorphism} together with the property $\lambda\circ c\equiv 0$ mod $G$.   This shows that $x^\transpose M (\lambda\circ c)$ is an integer. It follows that the fourth and fifth terms on the right hand side of eq.\ (\ref{eqinproof quadratic form derivation 1}) must be equal to an even integer and thus cancel modulo $2\Z$. Combining results we end up with the expression
\begin{equation}
q(g\oplus h) = q (g) + q (h) + 2 g^\transpose M h  + \Delta \pmod{2\Z},
\end{equation}
where
\begin{equation}
\Delta := (\lambda\circ c)^\transpose M (\lambda\circ c) + \sum_i M(i,i)\lambda(i) c_i^2.
\end{equation} We finish our proof by showing that $\Delta$ is an even integer too, which proves (\ref{eqinproof:quadratic form over 2Z}).

First, we note that, due to the symmetry of $M$, we can expand $(\lambda\circ c)^\transpose M (\lambda\circ c)$  as
\begin{equation}\label{Delta}
(\lambda\circ  c)^\transpose M (\lambda\circ c) = \sum_{i,j\: :\: i<j} 2  M(i,j) \lambda(i)\lambda(j) c_i c_j + \sum_i M(i,i)\lambda(i)^2 c_i^2.
\end{equation}
Revisiting (\ref{inproof:B(g,0)}) and  choosing $x=e_i$ and $\lambda = e_j$ for all different values of $i$, $j$, we obtain the following consistency equation for $M$
\begin{equation}\label{eqinproof:consistency conditions of the BICHARACTER matrix}
c_j M(i,j) = c_i M(i,j) = 0 \pmod{\Z}
\end{equation}
It follows that all terms of the form $2M(i,j) \lambda(i)\lambda(j) c_i c_j$  are \emph{even} integers. We can thus remove these terms from (\ref{Delta}) by taking modulo $2\Z$, yielding
 \begin{align}
 \Delta &=  \sum_i M(i,i)\lambda(i)^2 c_i^2 + \sum_i M(i,i) \lambda(i) c_i^2 \pmod{2\Z}\\
 &=  \sum_i M(i,i) c_i^2 \lambda(i)(\lambda(i)+1) =0  \pmod{2\Z},
 \end{align}
where in the last equality we used the fact that $\lambda(i)(\lambda(i)+1)$ is  necessarily even.

\subsection*{Proof of lemma \ref{lemma:Quadratic Function composed with Automorphism}}

The fact that $\xi_{M,v}\circ \alpha$ is quadratic follows immediately from the fact that $\xi_{M,v}$ is quadratic and that $\alpha$ is a homomorphism. Composed continuous functions lead to continuous functions. As a result, theorem \ref{thm:Normal form of a quadratic function} applies and we know $\xi_{M',v'}=\xi_{M,v}\circ \alpha$ for some choice of $M'$, $v'$. 

Let $B_M(g,h)=\exp\left(2\pii g^\transpose M h\right)$ be the bicharacter associated with $\xi_{M, v}$. One can show by direct evaluation (and using lemma \ref{lemma:properties of matrix representations}(a) and lemma  \ref{lemma symmetric matrix representation of the bicharacter homomorphism}) that $B_{M'}$ with $M':=A^\transpose M A$ is the bicharacter associated to $\xi_{M,v}\circ \alpha$. Let $Q_{M'}(g):=\exp(\pii \, (  g^{\transpose} M' g + C_{M'}^\transpose g )),$ be the quadratic function in lemma \ref{lemma quadratic B-representations differ by a character}. By construction,  both $\xi_{M,v}\circ \alpha$ and $Q_{M'}$ are $B_{M'}$-representations of  the bicharacter $B_{M'}$. As a result, lemma \ref{lemma quadratic B-representations differ by a character} tells us that the function $f(g):=\xi_{M,v}\circ \alpha(g)/ Q_{M'}(g)$ is a character of $G$, so that there exists $v'\in G^*$ such that $\chi_{v'}(g)=f(g)$. We can compute $v'$ by direct evaluation of this expression:
\begin{equation}
\chi_{v'}(g)= \exp\left( \pii \left( A^\transpose C_{M} - C_{A^\transpose M A}\right)g \right) \exp\left( 2\pii \left( A^\transpose v \right)\cdot  g \right).
\end{equation}
It can be checked that the function $\exp \left( 2\pii \left( A^\transpose v \right) g \right) $ is a character, using that  it is the composition of a character $\exp(2\pi v^\transpose g)$  (theorem \ref{thm:Normal form of a quadratic function}) and a continuous group homomorphism $\alpha$. Since $\chi_{v'}$ is also a character, the function  $\exp\left( \pii \left( A^\transpose C_{M} - C_{A^\transpose M A}\right)g \right)$ is a character too (as characters are a group under multiplication), and it follows that  $v_{A,M}=( A^\transpose C_{M} - C_{A^\transpose M A} )/2$ is congruent to some element of $G^\bullet$ \footnote{This statement can also be proven (more laboriously) by explicit evaluation, using arguments similar to those in the proof of lemma \ref{lemma:B-representations can always be constructed}.}; we obtain that $v'= A^\transpose v + v_{A,M}$ is an element of $G^\bullet$. Finally, we obtain that  $\xi_{M',v'}$ is a normal form of $\xi_{M,v}\circ \alpha$, using the relationship $\xi_{M,v}\circ \alpha (g)=Q_{M'}(g) f(g)=Q_{M'}(g) \chi_{v'}(g)=\xi_{M',v'}(g)$.

\end{document}